\documentclass[reprint,amsmath,amssymb,aps,pre,longbibliography]{revtex4-2}
\usepackage{xr-hyper}
\usepackage{tikz-network}
\usepackage[caption=false,font=footnotesize,indention=0cm]{subfig}
\usepackage{bm}
\usepackage{tabularx}
\usepackage{multirow}
\usepackage{hyperref}
\usepackage{algorithm}
\usepackage{algpseudocode}
\usepackage{cleveref}
\usepackage{amsthm}
\usepackage{eqparbox}

\theoremstyle{definition}
\newtheorem{theorem}{Theorem}
\newtheorem{lemma}{Lemma}%[section]

\def\R{\mathbb{R}}
\def\N{\mathbb{N}}

\begin{document}

\title{Core-periphery detection in multilayer networks}%

\author{Kai Bergermann}%
\email[]{kai.bergermann@math.tu-chemnitz.de}
\affiliation{Department of Mathematics, Technische Universität Chemnitz, 09107 Chemnitz, Germany}
\author{Francesco Tudisco}
\email[]{f.tudisco@ed.ac.uk}
\affiliation{School of Mathematics, The University of Edinburgh, Edinburgh EH93FD, UK}
\affiliation{School of Mathematics, Gran Sasso Science Institute, 67100 L’Aquila, Italy}
\date{\today}%

\begin{abstract}
Multilayer networks provide a powerful framework for modeling complex systems that capture different types of interactions between the same set of entities across multiple layers. Core-periphery detection involves partitioning the nodes of a network into core nodes, which are highly connected across the network, and peripheral nodes, which are densely connected to the core but sparsely connected among themselves. In this paper, we propose a new model of core-periphery in multilayer network and a nonlinear spectral method that simultaneously detects the corresponding core and periphery structures of both nodes and layers in weighted and directed multilayer networks. Our method reveals novel structural insights in three empirical multilayer networks from distinct application areas: the citation network of complex network scientists, the European airlines transport network, and the world trade network.
\end{abstract}

\maketitle

\section{Introduction}

The study of networks is a powerful approach to understanding complex systems across various disciplines, including physics, the natural sciences, economics, engineering, and sociology \cite{watts1998collective,barabasi1999emergence,newman2003structure,boccaletti2006complex}. Multilayer network models offer the flexibility to represent different types of relationships between the same set of nodes across multiple layers \cite{kivela2014multilayer,boccaletti2014structure}. In the early 2010s, several breakthroughs established the study of multilayer networks as a distinct sub-discipline within network science \cite{mucha2010community,de2013mathematical,gomez2013diffusion,granell2013dynamical,battiston2014structural}. Since then, the generalization of concepts and techniques from single-layer networks to multilayer networks has become an active and evolving field of research.

A crucial aspect of analyzing complex networks involves examining structural properties such as community structure, metric structure, degree distributions, and centralities \cite{newman2003structure,boccaletti2006complex}. The core-periphery (also referred to as rich club or core-fringe) detection problem involves dividing the nodes of a network into two subsets: the core, which is well-connected to the entire network, and the periphery, which is densely connected to the core but sparsely connected internally. Borgatti and Everett's seminal work first formalized core-periphery structure, characterizing it as an ideal L-shape in the network's adjacency matrix \cite{borgatti2000models}. Since then, numerous methods have been proposed for detecting core-periphery structures, and such structures have been identified in various types of networks, including physical, social, computer, biological, brain, and transport networks \cite{colizza2006detecting,csermely2013structure,rombach2014core,zhang2015identification,lu2016h,battiston2018multiplex,tudisco2019nonlinear,tudisco2019fast,higham2022core,bergermann2024nonlinear,polanco2023hierarchical}.

Core-periphery detection can be formulated as a combinatorial optimization problem, the exact solution of which becomes computationally infeasible for even moderately sized networks. To address this, standard methods often rely on a relaxed formulation where one computes a positive coreness vector $\bm{x} \in \mathbb{R}_{>0}^n$, where larger values of $\bm{x}_i$ indicate a higher likelihood that node $i$ belongs to the core. Many core-periphery detection methods exist for standard, single-layer networks and it was noted in \cite{tudisco2019nonlinear, tudisco2019fast} that a large number of available approaches, including e.g.\ \cite{borgatti2000models,rombach2014core,lu2016h, boyd2010computing,mondragon2017network}, can be formalized in terms of the following objective function:
\begin{equation}\label{eq:kernel_obj_function}
 f(\bm{x}) = \sum_{i,j=1}^n \bm{A}_{ij} \kappa(\bm{x}_i,\bm{x}_j),
\end{equation}
where $\bm{A} \in \mathbb{R}_{\geq 0}^{n \times n}$ is the adjacency matrix and $\kappa : \mathbb{R} \times \mathbb{R} \to \mathbb{R}$ is a suitable kernel function.

While well-understood in the standard single-layer setting, core-periphery detection in multilayer networks remains a relatively underdeveloped area. Existing approaches, such as those based on degree counts, offer limited insights into the more intricate structures of multilayer systems \cite{battiston2018multiplex}. A recent approach introduced in \cite{bergermann2024nonlinear} proposes a nonlinear spectral method to optimize a multilayer version of \eqref{eq:kernel_obj_function}, specifically for unweighted, undirected multiplex networks without inter-layer edges. This method has demonstrated superior performance compared to degree-based approaches in several multiplex networks, including genetic, transportation, and social systems. However, extending this method to general multilayer networks that incorporate inter-layer edges is non-trivial.

In this paper, we introduce a model to detect and quantify core and periphery sets for both nodes and layers in general multilayer networks, accommodating arbitrary inter- and intra-layer edges, potentially with weights and directions. Our model is based on the maximization of a generalized core-periphery kernel function, resulting in a challenging non-convex and nonlinear optimization problem. Despite this, we show that the model can be applied in practical settings by proposing a tailored core-periphery detection algorithm that efficiently solves the optimization task and determines core and periphery sizes for both layers and nodes. We provide theoretical results demonstrating the convergence of the algorithm to either global or local optima. We validate the proposed method on three empirical networks representing collaboration, transportation, and commercial trade data, revealing novel structural insights.

\section{Multilayer core-periphery model}

Defining the notion of core and periphery in multilayer networks is a non-trivial task, as nodes may belong to the core in some layers while belonging to the periphery in others, or peripheral nodes may be weakly connected within layers but strongly connected across layers.

In this section, we propose a novel model for detecting core-periphery structures in general multilayer networks. The key idea is to simultaneously detect core structures among both nodes and layers by optimizing a coreness vector for each, using a provably convergent and efficiently implementable optimization procedure.

For the representation of multilayer networks with $n$ nodes and $L$ layers, we define a fourth-order adjacency tensor
\begin{equation*}
\mathcal{A} \in \mathbb{R}_{\geq 0}^{n \times n \times L \times L}, \quad \big[\mathcal{A}_{ij}^{k\ell}\big]_{i,j=1,\dots,n}^{k,\ell=1,\dots,L},
\end{equation*}
which encodes possibly weighted and directed edges from node $i$ in layer $k$ to node $j$ in layer $\ell$.

Our goal is to optimize entry-wise positive node and layer coreness vectors, denoted by $\bm{x} \in \mathbb{R}_{>0}^n$ and $\bm{c} \in \mathbb{R}_{>0}^L$, respectively. High values in these vectors correspond to core nodes and layers, while low values correspond to peripheral ones. To prevent the blow-up of the objective function \eqref{eq:objective_function}, we impose the norm constraints $\|\bm{x}\|_p = \|\bm{c}\|_q = 1$ using the $p$-norm and $q$-norm, respectively. The choice of $p$ and $q$ does not affect the model, but will allow us to provide convergence guarantees of the proposed computational strategy.

We extend the objective function in \eqref{eq:kernel_obj_function} to multilayer networks by summing the entries of the adjacency tensor, weighted by a kernel function that incorporates both the node and layer coreness vectors. As our kernel function $\kappa$, we select the product of two smoothed maximum kernels, which encourages at least one large node coreness score $\bm{x}_i, \bm{x}_j$ and one large layer coreness score $\bm{c}_k, \bm{c}_\ell$ whenever the edge weight $\mathcal{A}_{ij}^{k\ell}$ is significant \cite{tudisco2019nonlinear,tudisco2019fast,higham2022core,bergermann2024nonlinear}. The objective function to maximize is given by:
\begin{equation}\label{eq:objective_function}
f_{\alpha,\beta} (\bm{x},\bm{c}) = \sum_{i,j=1}^n \sum_{k,\ell=1}^L \mathcal{A}_{ij}^{k\ell} \big( \bm{x}_i^\alpha + \bm{x}_j^\alpha \big)^{1/\alpha} \big( \bm{c}_k^\beta + \bm{c}_\ell^\beta \big)^{1/\beta}.
\end{equation}
Here, the scalar parameters $\alpha, \beta > 1$ control the smoothness of the kernel. In the limit $\alpha \to \infty$, the term $(\bm{x}_i^\alpha + \bm{x}_j^\alpha)^{1/\alpha}$ converges to $\max\{\bm{x}_i,\bm{x}_j\}$, yielding the standard maximum kernel. Notably, when $\alpha = \beta$, the roles of $\bm{x}$ and $\bm{c}$ are interchanged upon transforming the original network into its dual, as explored in Ref.~\cite{presigny2024node}. The choice $\alpha\neq\beta$, however, allows for additional model flexibility.

\begin{figure}
\begin{algorithm}[H]
\raggedright{
	\begin{tabular}{lll}
		%		\vspace{1mm}
		\textbf{Input}:
		& $\mathcal{A}\in\R^{n \times n \times L \times L}_{\geq 0},$ & Adjacency tensor.\\%\vspace{1mm}
		& $\bm{x}_0\in\R^n_{>0},$ & Initial node coreness vector.\\
		& $\bm{c}_0\in\R^L_{>0},$ & Initial layer coreness vector.\\
	\end{tabular}\\
	\begin{tabular}{ll}
		\textbf{Parameters}: & $\alpha,\beta,p,q>1, \mathrm{tol}\in\R_{>0}, \mathrm{maxIter}\in\N.$
	\end{tabular}
	\begin{algorithmic}[1]
		\State $\bm{x}_0 =\bm{x}_0/\| \bm{x}_0 \|_{\frac{p}{p-1}}$
		\State $\bm{c}_0 = \bm{c}_0/\| \bm{c}_0 \|_{\frac{q}{q-1}}$
		\For{$k=1:\mathrm{maxIter}$}
		%		\vspace{1mm}
		\State $\widetilde{\bm{x}} = \nabla_{\bm{x}} f_{\alpha,\beta}(\bm{x}_{k-1}, \bm{c}_{k-1})$
		%		\vspace{1mm}
		\State $\bm{x}_{k} = \left( \widetilde{\bm{x}}/\| \widetilde{\bm{x}} \|_{\frac{p}{p-1}} \right)^{\frac{1}{p-1}}$
		%		\vspace{1mm}
		\State $\widetilde{\bm{c}} = \nabla_{\bm{c}} f_{\alpha,\beta}(\bm{x}_{k-1}, \bm{c}_{k-1})$
		%		\vspace{1mm}
		\State $\bm{c}_{k} = \left( \widetilde{\bm{c}}/\| \widetilde{\bm{c}} \|_{\frac{q}{q-1}} \right)^{\frac{1}{q-1}}$
		%		\vspace{1mm}
		\If{$\| \bm{x}_{k} - \bm{x}_{k-1} \| < \mathrm{tol}$ \textbf{and} $\| \bm{c}_{k} - \bm{c}_{k-1} \| < \mathrm{tol}$}
		\State \textbf{break}
		\EndIf
		\EndFor
	\end{algorithmic}
	%	\vspace{1em}
	\begin{tabular}{lll}
		%		\vspace{1mm}
		\textbf{Output}: & $\bm{x}\in\mathcal{S}_p^+,$ & Optimized node coreness vector.\\
		& $\bm{c}\in\mathcal{S}_q^+,$ & Optimized layer coreness vector.
	\end{tabular}
}
\caption{Nonlinear spectral method for core-periphery detection in multilayer networks.}\label{alg}
\end{algorithm}
\end{figure}

To enforce the norm constraints, we normalize the objective function \eqref{eq:objective_function} by dividing by $\|\bm{x}\|_p \|\bm{c}\|_q$. The optimization problem is then solved by setting the gradients with respect to both $\bm{x}$ and $\bm{c}$ to zero, yielding the following fixed-point equations:
\begin{equation}\label{eq:fixed_point_equations}
\bm{x} = J_{p^\ast} (\nabla_{\bm{x}} f_{\alpha,\beta}(\bm{x},\bm{c})), \quad \bm{c} = J_{q^\ast} (\nabla_{\bm{c}} f_{\alpha,\beta}(\bm{x},\bm{c})),
\end{equation}
using similar arguments as in Refs.~\cite{tudisco2019nonlinear,bergermann2024nonlinear}. Here, $J_p(\bm{x}) = \|\bm{x}\|_p^{1-p} \bm{x}^{p-1}$ is the gradient of the $p$-norm, and $p^\ast$ is the H\"older conjugate of $p$, defined by $1/p + 1/p^\ast = 1$.

Our proposed method consists of alternating fixed-point iterations of the two equations in \eqref{eq:fixed_point_equations}, which are summarized in \Cref{alg}. This method can be efficiently implemented in a few lines of code. We refer to it as a nonlinear spectral method, as it can be interpreted as a power iteration for a generalized nonlinear eigenvalue problem \cite{tudisco2019nonlinear,bergermann2024nonlinear}.

The global convergence of \Cref{alg} is ensured by nonlinear and multilinear Perron-Frobenius theory \cite{gautier2019contractivity,gautier2019perron,gautier2023nonlinear}. Specifically, we define the following $2 \times 2$ coefficient matrix:
\begin{equation}\label{eq:coefficient_matrix}
\bm{M} = \begin{bmatrix}
\frac{2|\alpha-1|}{p-1} & \frac{2}{p-1} \\
\frac{2}{q-1} & \frac{2|\beta-1|}{q-1}
\end{bmatrix},
\end{equation}
and its spectral radius $\rho(\bm{M})$. We can then state the following convergence result:
\begin{theorem}\label{thm:unique_solution}
    For parameters $\alpha, \beta, p, q > 1$ such that
    \begin{equation*}
    \rho(\bm{M}) < 1,
    \end{equation*}
    \Cref{alg} converges to the global maximum of $f_{\alpha,\beta}$ defined in \eqref{eq:objective_function}, for any initial vectors $\bm{x}_0 \in \mathbb{R}_{>0}^n$ and $\bm{c}_0 \in \mathbb{R}_{>0}^L$, with a linear rate of convergence.
\end{theorem}
\begin{proof}
    The proof is provided in \Cref{sec:proof}.
\end{proof}

It has been observed that for single-layer and multiplex networks, selecting relatively small values of $p$ and $q$ may yield core-periphery structures closer to the ideal L-shape \cite{bergermann2024nonlinear}.
As such values do not necessarily satisfy the assumption of \Cref{thm:unique_solution}, the global convergence of \Cref{alg} cannot be guaranteed in general.
In this situation, however, convergence to a local optimum is still ensured by a previous result \cite[Thm.~3.2]{bergermann2024nonlinear}.

\begin{figure*}[t]
	\subfloat[Original supra-adjacency]{
		\includegraphics[width=0.19\textwidth]{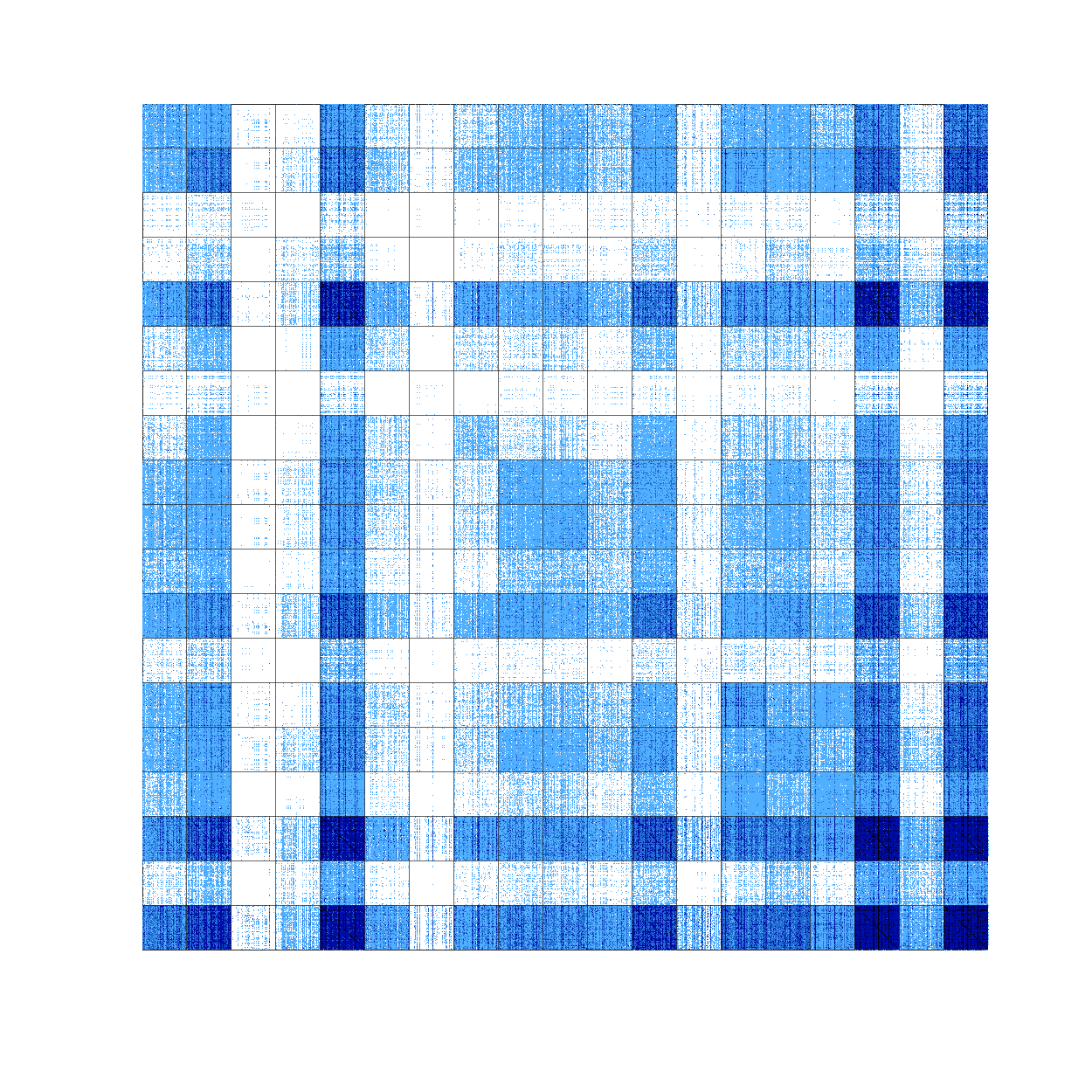}\label{fig:OA_spy_plots_original}
	}
	\subfloat[Full permuted supra-adjacency for $p=q=22$]{
		\includegraphics[width=0.19\textwidth]{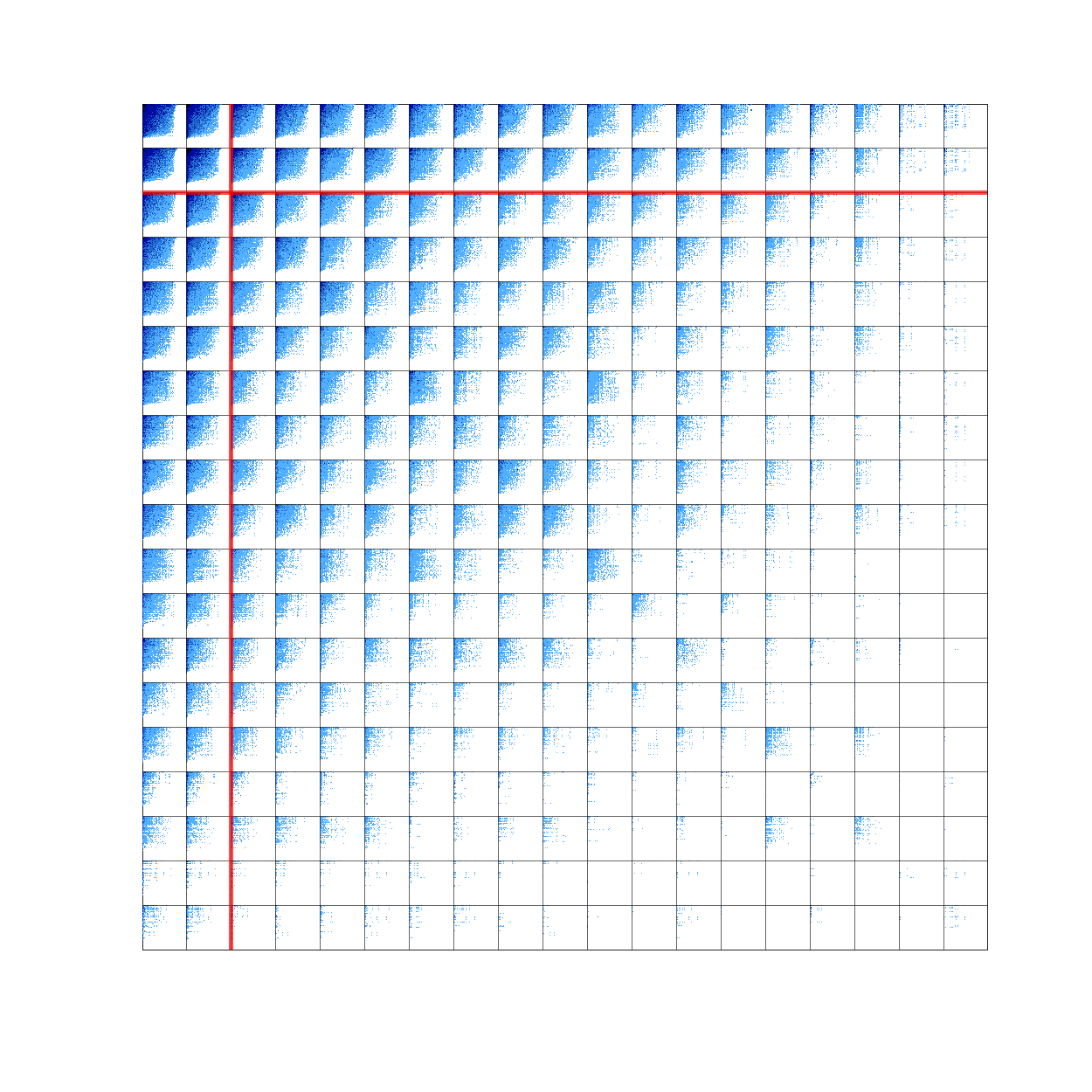}
	}
	\subfloat[Upper left $5\times 5$ blocks of permuted supra-adjacency for $p=q=22$]{
		\includegraphics[width=0.19\textwidth]{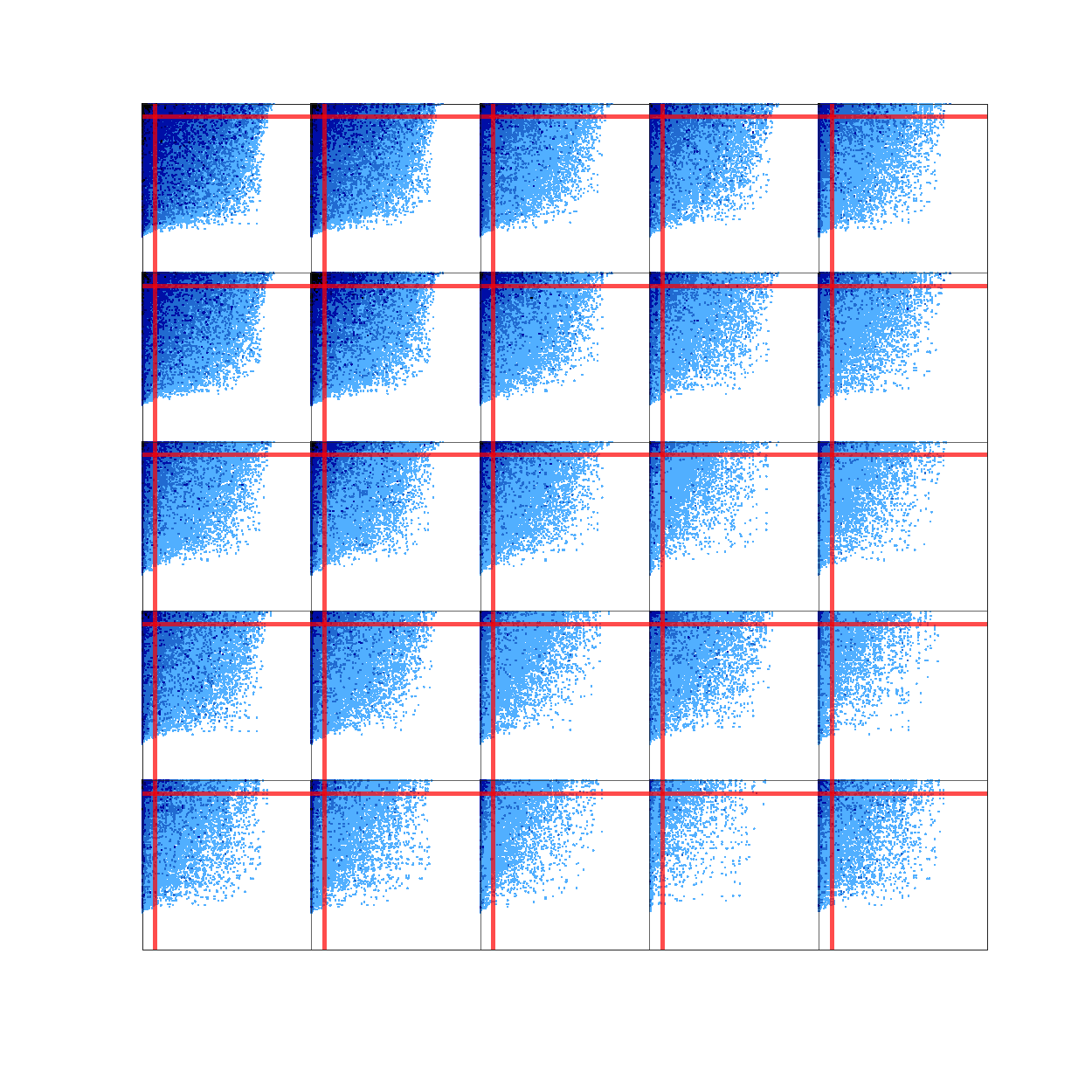}
	}
	\subfloat[Full permuted supra-adjacency for $p=q=2$]{
		\includegraphics[width=0.19\textwidth]{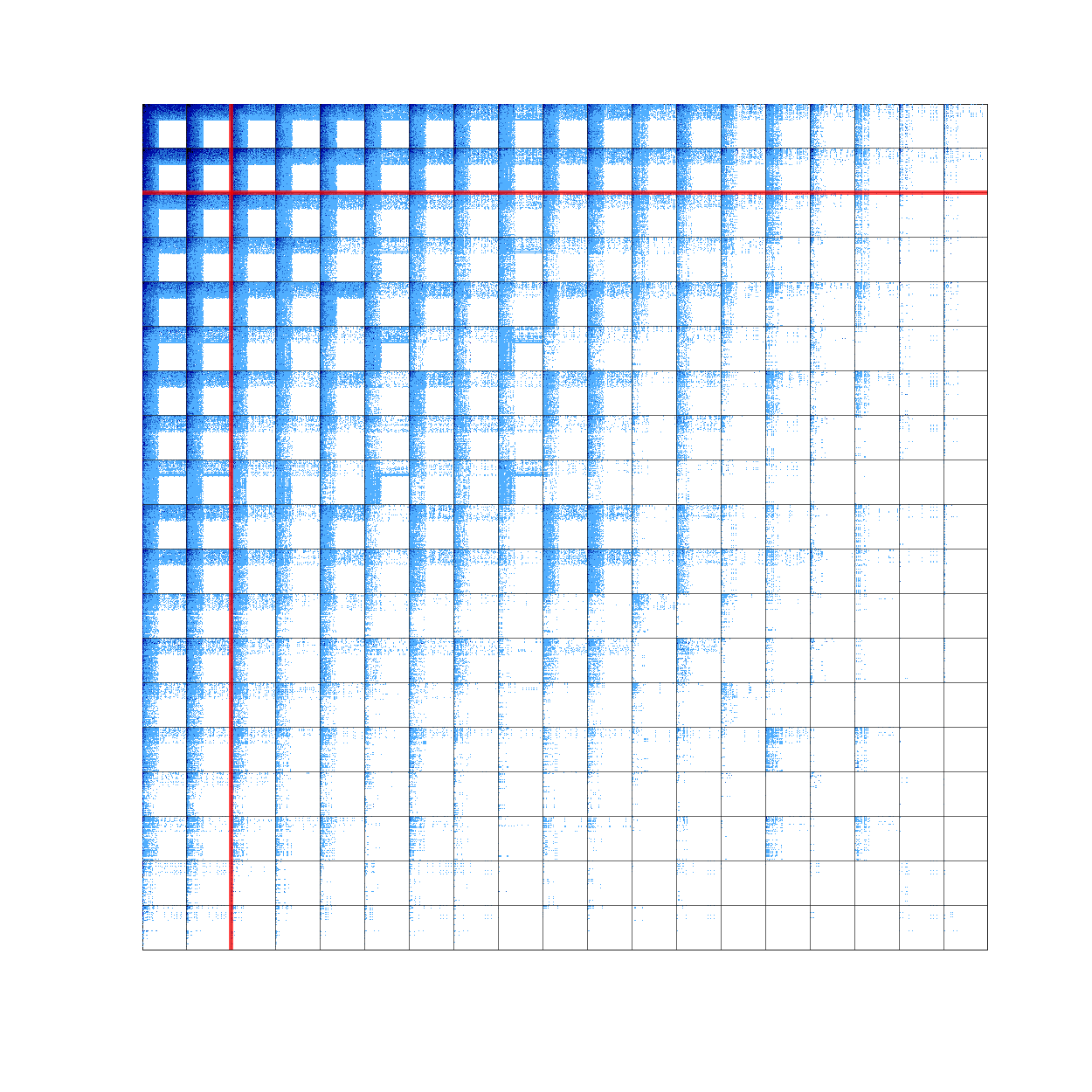}
	}
	\subfloat[Upper left $5\times 5$ blocks of permuted supra-adjacency for $p=q=2$]{
		\includegraphics[width=0.19\textwidth]{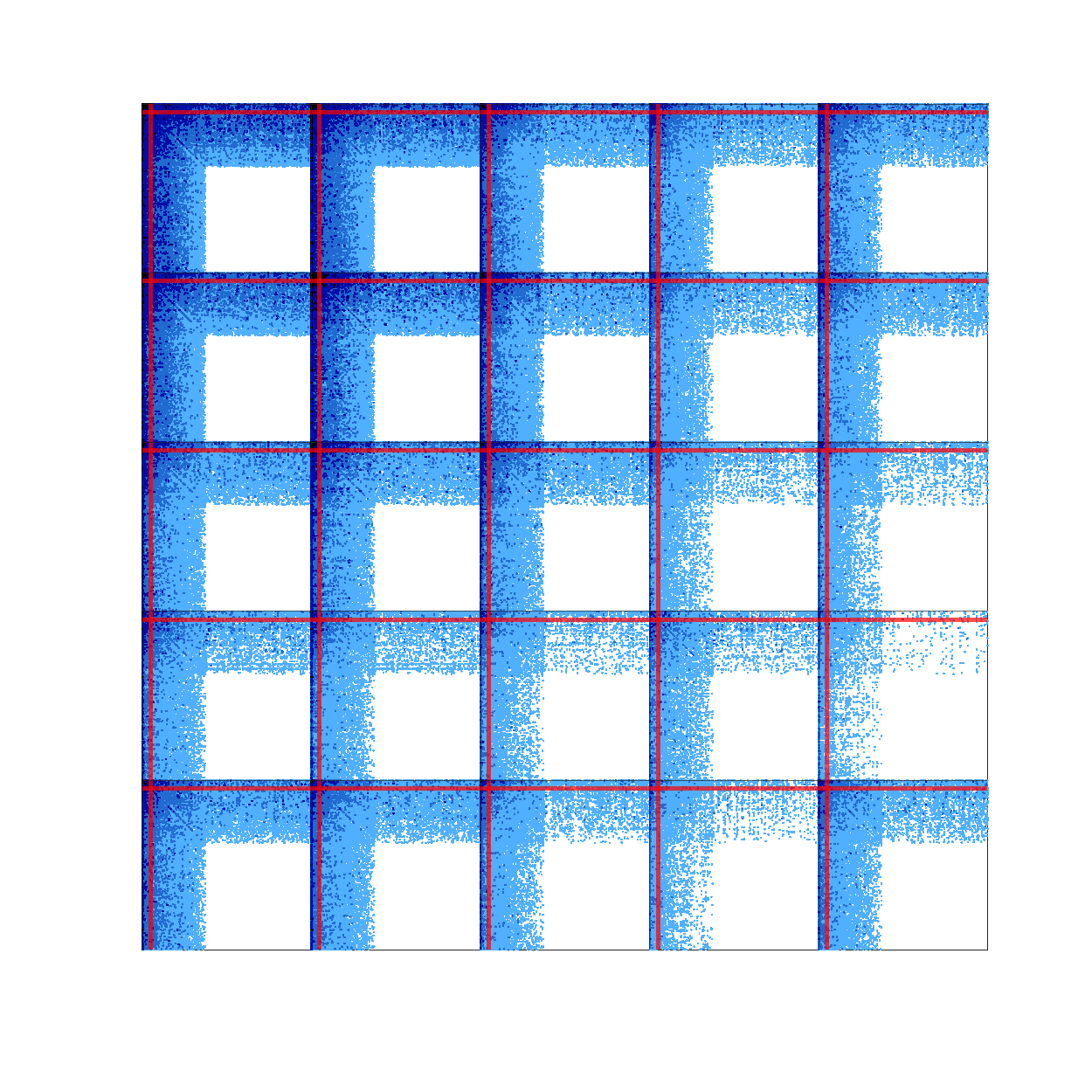}
	}
	\caption{Supra-adjacency matrix plots for the weighted OpenAlex multilayer citation network of complex network scientists of the year 2023.
	Dark fonts indicate large edge weights.
	Red lines in panels b) and d) indicate the layer core sizes $s^\ast_{\mathrm{layer}}=2$ and red lines in panels c) and e) indicate the node core sizes $s^\ast_{\mathrm{node}}=4\,058$ and $s^\ast_{\mathrm{node}}=2\,566$, respectively, in each block.}\label{fig:OA_spy_plots}
\end{figure*}

\pgfkeys{/pgf/number format/.cd,1000 sep={}}
\begin{figure*}[t]
    \includegraphics[width=.98\textwidth]{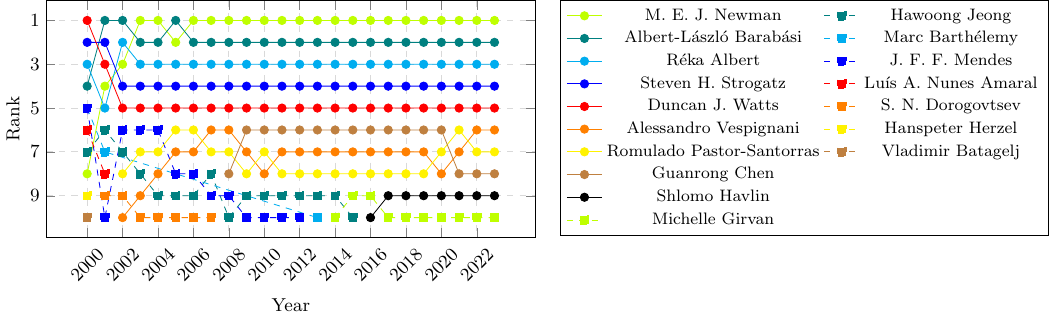}
	\caption{Top $10$ authors by node coreness score in the weighted OpenAlex citation multilayer network over the years $2000$ to $2023$ for the parameters $p=q=22$.}\label{fig:OA_author_rankings}
\end{figure*}

\section{Determining core sizes}\label{sec:core_size}

Once the node and layer coreness vectors $\bm{x}$ and $\bm{c}$ have been optimized, the next challenge is determining where to draw the boundary between core and periphery.

In single-layer networks, core-periphery structures are often visualized by permuting the rows and columns of the adjacency matrix $\bm{A} \in \mathbb{R}_{\geq 0}^{n \times n}$ in descending order of the nodes' coreness scores. A recent method proposes identifying the core size by sweeping over the permuted adjacency matrix and comparing it to Borgatti and Everett's idealized L-shape model for different core sizes \cite{higham2022core}.

We extend this approach to multilayer networks by generalizing the single-layer L-shape model. The multilayer network is represented by flattening the adjacency tensor into the supra-adjacency matrix, a block matrix of size $L \times L$, where each block is of size $n \times n$. Block $(k,\ell)$ represents edges between nodes in layer $k$ and layer $\ell$ \cite{mucha2010community,gomez2013diffusion,de2013mathematical,kivela2014multilayer}. To facilitate comparison, we normalize the supra-adjacency matrix by dividing it by its largest entry, allowing edge weights of $1$ to represent present edges, $0$ missing edges, and weights in $(0,1)$ partially present edges.

We propose a two-level permutation of the supra-adjacency matrix. First, the block rows and columns are permuted using the layer coreness vector $\bm{c}$, while each block is further permuted according to the node coreness vector $\bm{x}$, as done in the single-layer case. A strong multilayer core-periphery structure is characterized by dense blocks forming an L-shape at the block level, with each block itself also exhibiting an inner L-shape. \Cref{fig:OA_spy_plots} provides an illustrative example.

To identify the core size in single-layer networks, the sweeping procedure evaluates $n$ binary vectors $\bar{\bm{x}} \in \{0,1\}^n$, where each binary vector $s = 1,\dots,n$ contains non-zeros in the first $s$ entries \cite{higham2022core}. Given the permuted binary adjacency matrix $\bm{A}_{ij}$, the objective function
\begin{equation}\label{eq:QUBO_single_layer}
\sum_{i,j=1}^n \left( \frac{\bm{A}_{ij}}{n_1} \max \{ \bar{\bm{x}}_i, \bar{\bm{x}}_j\} + \frac{1 - \bm{A}_{ij}}{n_2} (1-\max \{ \bar{\bm{x}}_i, \bar{\bm{x}}_j\})\right),
\end{equation}
is evaluated for each $\bar{\bm{x}}$, where $n_1 = \sum_{i,j=1}^n \bm{A}_{ij}$ is the total number of present edges, and $n_2 = n^2 - n_1$ is the number of absent edges. Equation \eqref{eq:QUBO_single_layer} compares the adjacency matrix with an ideal L-shape model for core size $s$, with present edges in the first $s$ rows and columns, and absent edges in the bottom-right $(n-s) \times (n-s)$ block. The normalization factors $n_1$ and $n_2$ ensure that both contributions are weighted according to the overall structure of the network. The core size $s^\ast$ is selected as the value of $s$ that maximizes \eqref{eq:QUBO_single_layer} \cite{higham2022core,bergermann2024nonlinear}.

For multilayer networks, we aim to detect both the node core size $s^\ast_{\mathrm{node}}$ and the layer core size $s^\ast_{\mathrm{layer}}$ via separate sweeping procedures. To determine the node core size, we compute the similarity of each block in the permuted supra-adjacency matrix to an ideal L-shape using the layer coreness values $\gamma_{k\ell} = \max\{\bm{c}_k,\bm{c}_\ell\}$ as weights. Using the same binary vectors $\bar{\bm{x}} \in \{0,1\}^n$ as in the single-layer case, the objective function is:
\begin{equation}\label{eq:QUBO_node}
\sum_{k,\ell=1}^L \gamma_{k\ell} \sum_{i,j=1}^n \left( \frac{\mathcal{A}_{ij}^{k\ell}}{n_1^{(k\ell)}} \bar{\chi}_{ij} + \frac{1 - \mathcal{A}_{ij}^{k\ell}}{n_2^{(k\ell)}} (1-\bar{\chi}_{ij}) \right),
\end{equation}
where $n_1^{(k\ell)} = \sum_{i,j=1}^n \mathcal{A}_{ij}^{k\ell}$, $n_2^{(k\ell)} = n^2 - n_1^{(k\ell)}$, and $\bar{\chi}_{ij} = \max\{ \bar{\bm{x}}_i, \bar{\bm{x}}_j \}$. The node core size $s^\ast_{\mathrm{node}}$ is chosen as the number of non-zero entries in $\bar{\bm{x}}$ that maximizes \eqref{eq:QUBO_node}.

For the layer core size, we perform a similar sweep over binary layer coreness vectors $\bar{\bm{c}} \in \{0,1\}^L$ and define the corresponding objective function as:
\begin{equation}\label{eq:QUBO_layer}
\sum_{i,j=1}^n \chi_{ij} \sum_{k,\ell=1}^L \left( \frac{\mathcal{A}_{ij}^{k\ell}}{n_1^{(ij)}} \bar{\gamma}_{k\ell} + \frac{1 - \mathcal{A}_{ij}^{k\ell}}{n_2^{(ij)}} (1-\bar{\gamma}_{k\ell}) \right),
\end{equation}
where $n_1^{(ij)} = \sum_{k,\ell=1}^L \mathcal{A}_{ij}^{k\ell}$, $n_2^{(ij)} = L^2 - n_1^{(ij)}$, $\chi_{ij} = \max \{ \bm{x}_i, \bm{x}_j\}$, and $\bar{\gamma}_{k\ell} = \max \{ \bar{\bm{c}}_k, \bar{\bm{c}}_\ell\}$. This respects the symmetry between nodes and layers imposed by \eqref{eq:objective_function}.

Normalizing \eqref{eq:QUBO_node} by $\frac{1}{\sum_{k,\ell=1}^L \gamma_{k\ell}}$ and \eqref{eq:QUBO_layer} by $\frac{1}{\sum_{i,j=1}^n \chi_{ij}}$ ensures that the objective function values lie in the range $[-1, 1]$, where larger values indicate stronger multilayer core-periphery structures, with configurations closer to the ideal two-level L-shape. Note that the described procedure relies on the node and layer coreness rankings generated by \Cref{alg} that are used to obtain the permuted supra-adjacency matrix. Details on the efficient numerical evaluation of \eqref{eq:QUBO_node} and \eqref{eq:QUBO_layer} are provided in Section V of the supplementary materials.

\section{Numerical results on empirical multilayer networks}

We test \Cref{alg} on three real-world multilayer networks: the citation network of complex network scientists, the European airlines transport network, and the world trade network. In addition, we report results on synthetic networks in \Cref{sec:experiments_synthetic}. In all experiments, we use the parameters $\alpha = \beta = 10$ \footnote{We found a value of $\alpha=\beta=10$ to be a reasonable compromise between the two objectives of the kernels in \eqref{eq:objective_function} to closely match the maximum kernel and not being too restrictive in terms of the choice of the parameters $p$ and $q$ with respect to the condition $\rho(\bm{M})<1$ in \Cref{thm:unique_solution}.}, \texttt{tol} = $10^{-8}$, \texttt{maxIter} = $200$, with initial vectors $\bm{x}_0 = \bm{1} \in \mathbb{R}^n$ and $\bm{c}_0 = \bm{1} \in \mathbb{R}^L$, where $\bm{1}$ denotes the vector of all ones. A \texttt{julia} implementation is publicly available under \url{https://github.com/COMPiLELab/MLCP} and all the network data is publicly available under \url{https://doi.org/10.5281/zenodo.14231869}.

\subsection{Citation multilayer network of Complex Networks scientists}
We constructed the citation network of complex network scientists using publicly available data from Open\-Alex \cite{priem2022openalex}. We collected data on all recorded complex network science publications by filtering for the level-2 concept \cite{openalex_note} ``Complex Network'', resulting in $38,346$ works published before January 1, 2024. After author disambiguation to handle hyphenated names, we identified $n = 53,423$ unique authors, representing the nodes in the multilayer network. The layers are instead formed by the level-0 concepts in the OpenAlex topics hierarchy. We detected $L = 19$ level-0 concepts, including ``Computer Science,'' ``Mathematics,'' and ``Physics''.

Edges originate from one complex network science paper citing another. Directed edges were inserted from all citing authors and disciplines associated with the citing paper toward all cited authors and disciplines of the cited paper, resulting in $1.55 \cdot 10^7$ non-zero entries in the adjacency tensor $\mathcal{A}$. We weighted citations such that the total edge weight of one paper citing another is $1$, with each citation weight divided by the product of the number of citing authors, citing disciplines, cited authors, and cited disciplines. Additional results, including analyses on ``unweighted'' multilayer networks (which omit the citation weighting), are provided in Section \ref{sec:results_OA}. Both weighting approaches yield sparse networks with $\mathcal{O}(nL)$ edges.

\Cref{fig:OA_spy_plots} shows the original as well as the permuted supra-adjacency matrices of the weighted and directed Open\-Alex citation network. Both the outer and the inner permutation, i.e., the permutation of and within blocks shows a strong concentration of edges with large weights in the upper left corner of supra-adjacency matrix and all individual blocks, respectively. Using $\alpha = \beta = 10$ in \Cref{alg}, the parameters $p = q = 22$ satisfy the assumptions of \Cref{thm:unique_solution}, ensuring convergence to the global optimum. With $p = q = 2$, convergence is only guaranteed to a local optimum, yet the intra-block reorderings exhibit strong L-shapes. Both parameter settings detect $s^\ast_{\mathrm{layer}} = 2$ disciplines, ``Computer Science'' and ``Mathematics,'' as the layer core. For $p = q = 22$, $s^\ast_{\mathrm{node}} = 4,058$ authors belong to the node core, whereas for $p = q = 2$, $s^\ast_{\mathrm{node}} = 2,566$. Despite the large problem size of $n = 53,423$ and $L = 19$, \Cref{alg} required only $135$ seconds on a standard laptop, completing in 26 iterations for $p = q = 22$.

We also constructed a sequence of adjacency tensors for publications up to the years $2000$ to $2023$. \Cref{fig:OA_author_rankings} shows the evolution of the top $10$ author coreness rankings over this period.
After strong reorderings following several landmark publications in the early $2000$s, the top $5$ core authors remain constant from year $2006$ onward.

\begin{figure}[t]
	\begin{center}
	\includegraphics[width=0.4\textwidth,clip,trim=330pt 80pt 290pt 90pt]{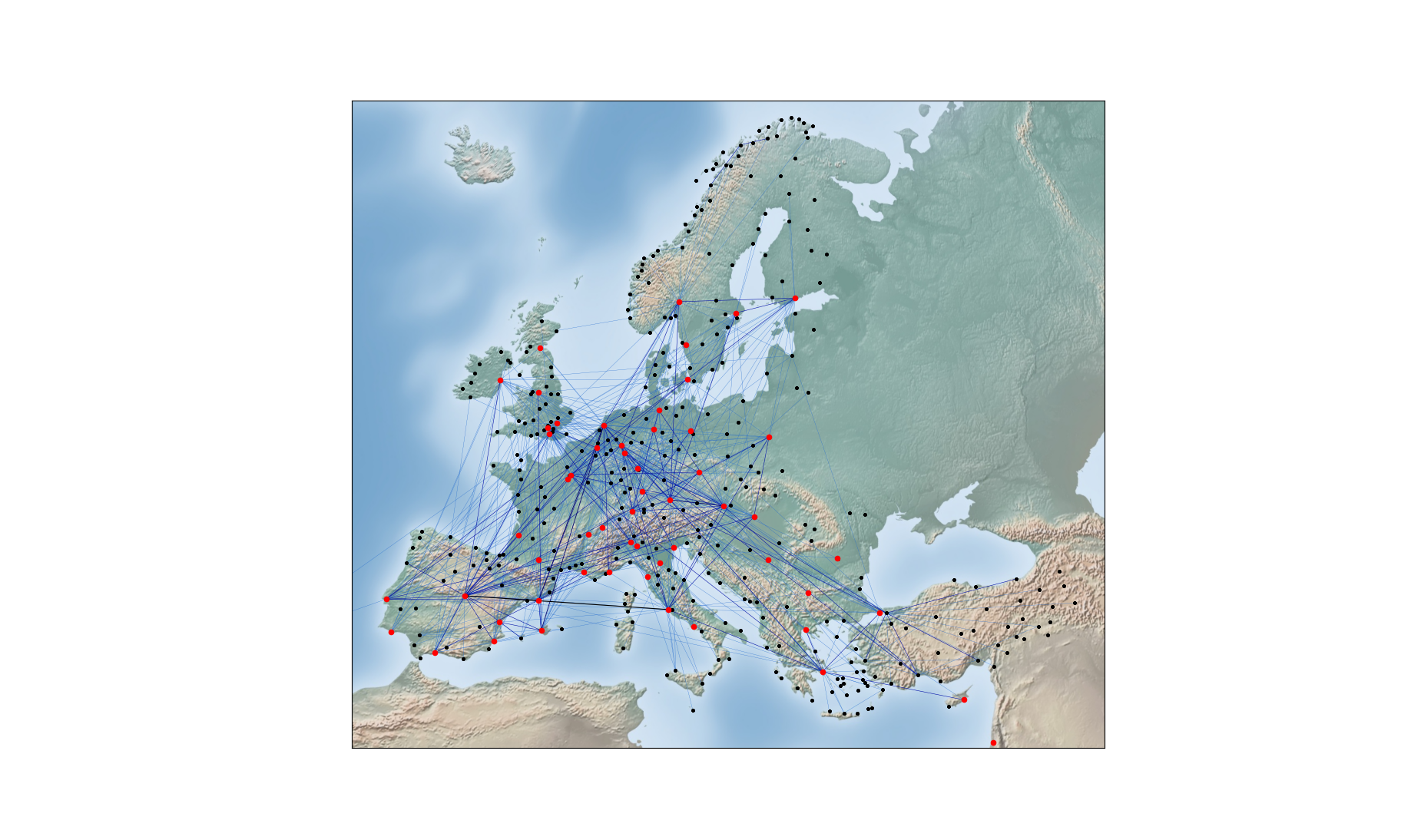}
	\end{center}
	\caption{Plot of the aggregated European Airlines network.
		Core airports are marked red while periphery airports are marked black.
		Only edges with aggregated intra-layer edge weight larger than $1$ are recorded.
		Dark fonts indicate large edge weights.
		The figure was created with matplotlib's basemap library.}\label{fig:european_airlines}
\end{figure}

\subsection{European airlines multilayer network}\label{sec:numerics_EUAir}
We analyze the largest connected component of the European Airlines (EUAir) network as of 2012 \cite{cardillo2013emergence}. This multiplex network consists of $n = 417$ airports and $L = 37$ airlines. Intra-layer edges are undirected and unweighted, indicating the existence of flight connections between airports for a given airline. Inter-layer edges follow the approach in Ref.~\cite{bergermann2021orientations}, where nodes (airports) in different layers (airlines) are connected if the airport is served by both airlines, allowing passengers to switch between airlines.

We ran \Cref{alg} with $p = q = 22$, detecting $s^\ast_{\mathrm{node}} = 57$ core airports (see Table \ref{tab:euair_airport_core}) and $s^\ast_{\mathrm{layer}} = 4$ core airlines: ``Lufthansa,'' ``easyJet,'' ``Ryanair,'' and ``Air Berlin.'' \Cref{fig:european_airlines} illustrates the core airports alongside intra-layer connections with aggregated edge weights greater than $1$.

An interesting observation is that Ryanair, ranked highly by several centrality measures \cite{bergermann2022fast}, is the third-highest layer by coreness score for $p = q = 22$, but drops to rank $18$ out of $L = 37$ layers for $p = q = 2$. This reflects Ryanair's tendency to serve peripheral airports, despite having numerous connections.

\begin{figure}[t]
        \includegraphics[width=.49\textwidth]{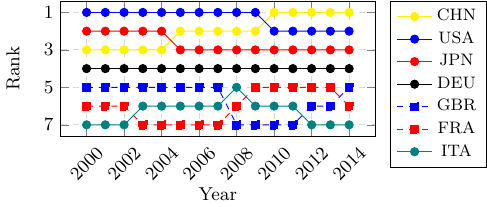}
		\caption{Rankings of country cores by layer coreness scores over the years $2000$ to $2014$ in the WIOD world trade multilayer network.}\label{fig:WIOD_cores_country}
\end{figure}
\begin{figure}[t]
    \includegraphics[width=.49\textwidth]{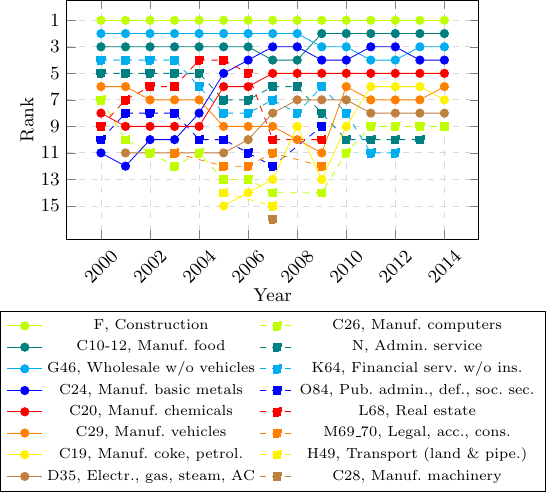}
	\caption{Rankings of industry cores by node coreness scores over the years $2000$ to $2014$ in the WIOD world trade multilayer network.}\label{fig:WIOD_cores_industry}
\end{figure}

\subsection{World trade multilayer network}
We consider the world trade network, recorded in the World Input Output Database (WIOD) \cite{timmer2015illustrated,wiod2021} for the years 2000 to 2014. Nodes correspond to $n = 56$ industries, and layers correspond to $L = 43$ countries. A separate multilayer network was constructed for each year. Directed and weighted edges represent trade volumes between industries in different countries, measured in millions of USD. Although the network is not fully connected, it is dense with $\mathcal{O}(n^2 L^2)$ edges, but the edge weights vary widely.

\Cref{fig:WIOD_cores_country,fig:WIOD_cores_industry} illustrates the evolution of country and industry cores from 2000 to 2014, showing their coreness score rankings. \Cref{fig:WIOD_cores_country} shows a consistent country core of $s^\ast_{\mathrm{layer}} = 7$ countries over all years, with notable shifts in rankings, such as China's rise to the top position. In contrast, \Cref{fig:WIOD_cores_industry} shows that the industry core size $s^\ast_{\mathrm{node}}$ varies between $9$ and $16$, with a sharp drop from $16$ in 2007 to $10$ in 2008, reflecting the global financial crisis. Noteworthy trends include the increasing coreness of the ``Manufacturing of Basic Metals'' industry and the decreasing coreness of ``Financial Services'' over the years.

\section{Conclusion}

In this work, we proposed a new model of core-periphery structure in general multilayer networks along with a provably convergent nonlinear spectral method for detecting the modeled cores and peripheries. Our method simultaneously optimizes both node and layer coreness vector, allowing for a comprehensive understanding of core-periphery dynamics across multiple layers. We also introduced a procedure for determining the sizes of node and layer cores and applied our framework to three real-world multilayer networks from diverse application domains, uncovering novel structural insights into citation, transportation, and world trade networks.

\onecolumngrid
\appendix

\section{Results on synthetic networks}\label{sec:experiments_synthetic}

We test our nonlinear spectral method (Algorithm 1 in the main text) on a range of synthetic multilayer networks with prescribed core-periphery structure.
We use directed stochastic block model networks with two edge probabilities---$p_c$ for edges within the upper-left L-shape representing connections between core nodes in Borgatti and Everett's ideal core-periphery model \cite{borgatti2000models} and $p_p$ for the remaining edges in the bottom right square---for all adjacency tensor slices $\mathcal{A}^{(k,l)}\in\R^{n\times n}, l,k=1,\dots,L$, or equivalently, all blocks of the supra adjacency matrix of the multilayer network.
Before running the core-periphery detection method, we obscure the planted structure by applying the same random row and column permutation to all slices or blocks, respectively.
The results of the nonlinear spectral method are visualized by spy plots of the supra-adjacency matrix with the two-level permutation based on the node and layer coreness vectors described in the main text, cf., e.g., Figure 1 in the main text or \Cref{fig:OA_spy_plots_unweighted}.

\begin{figure*}[b]
	\subfloat[All informative layers, $q_{s^\ast_{\mathrm{node}}}=0.586$.]{
		\includegraphics[width=0.19\textwidth]{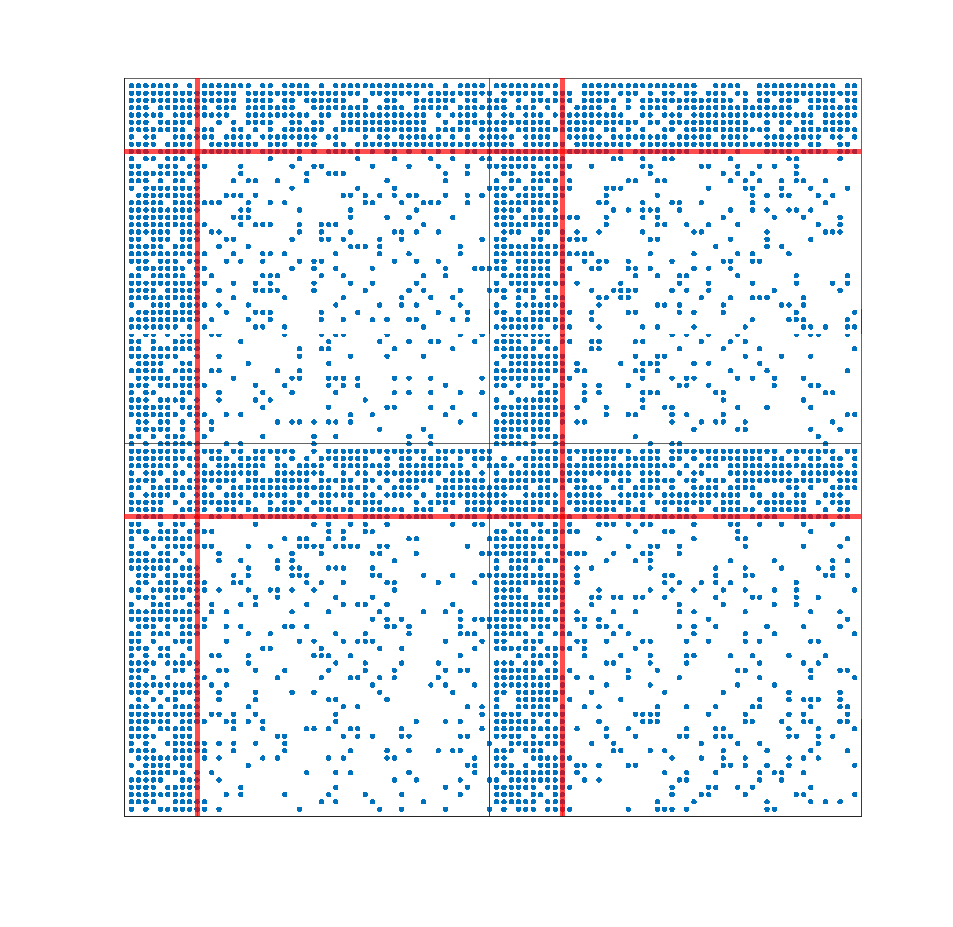}
	}
	\subfloat[Three informative layers, one noise layer, $q_{s^\ast_{\mathrm{node}}}=0.440$.]{
		\includegraphics[width=0.19\textwidth]{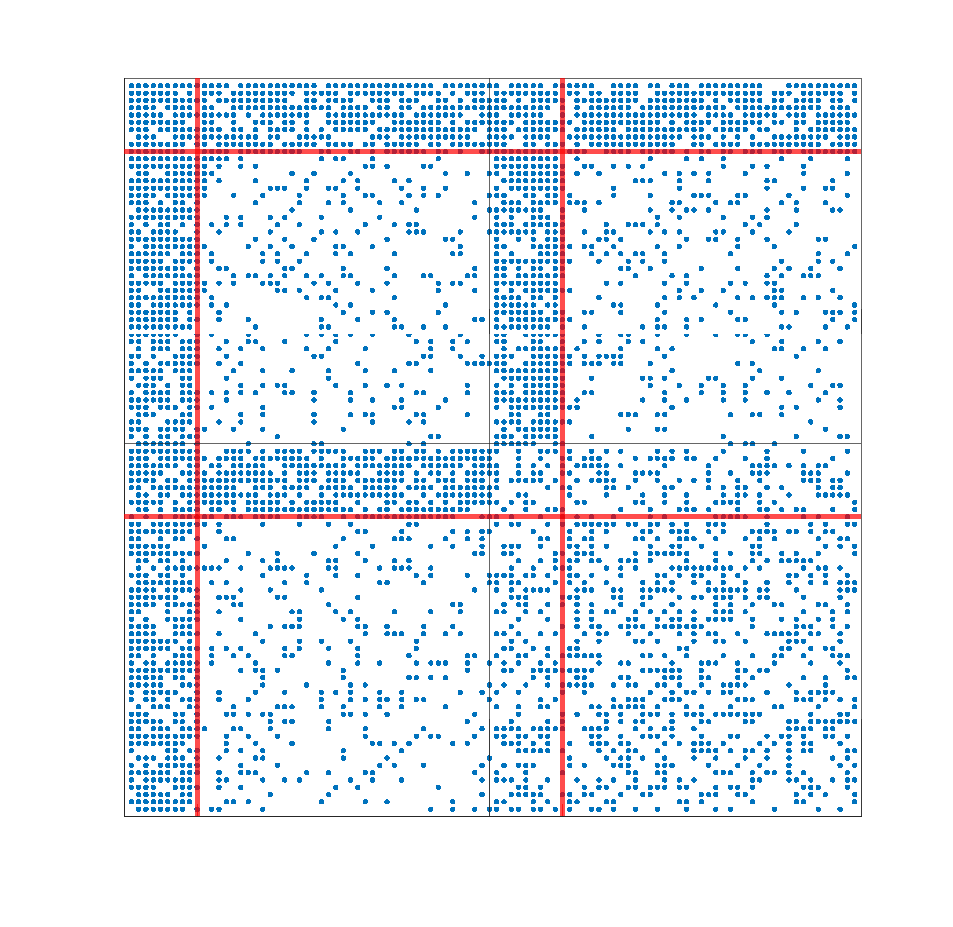}
	}
	\subfloat[Two informative layers, two noise layer, $q_{s^\ast_{\mathrm{node}}}=0.304$.]{
		\includegraphics[width=0.19\textwidth]{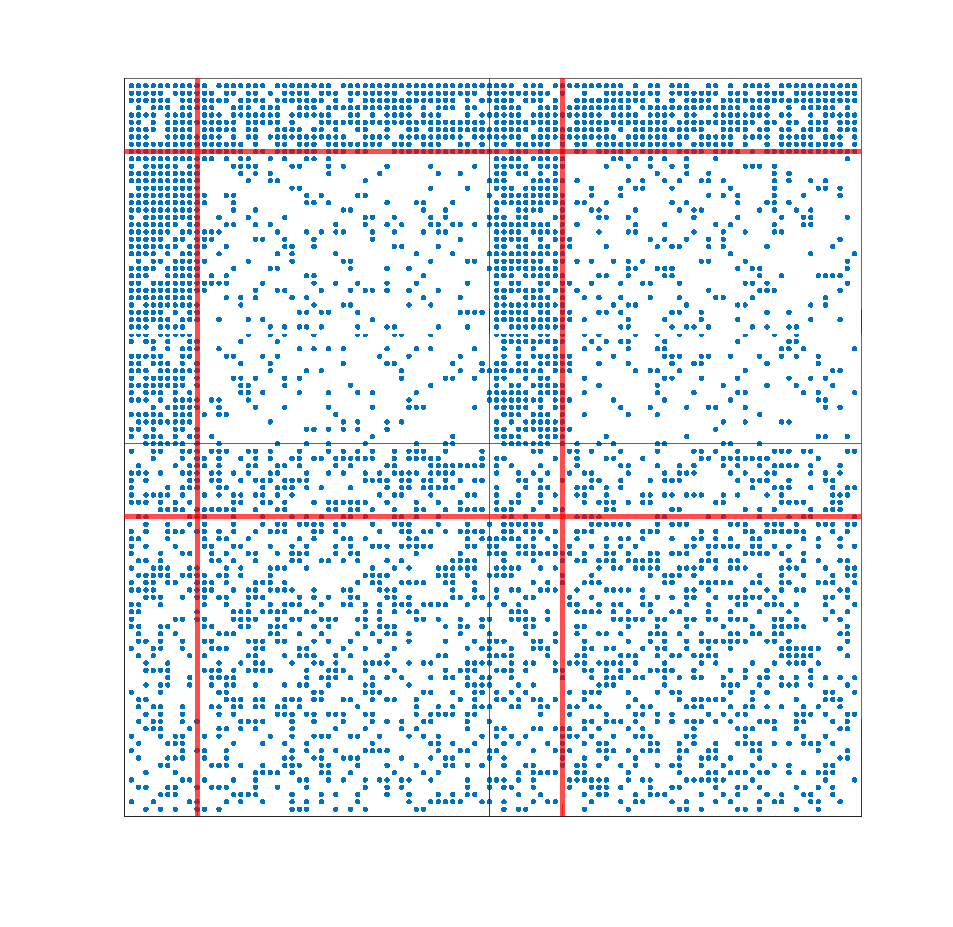}
	}
	\subfloat[One informative layer, three noise layers, $q_{s^\ast_{\mathrm{node}}}=0.159$.]{
		\includegraphics[width=0.19\textwidth]{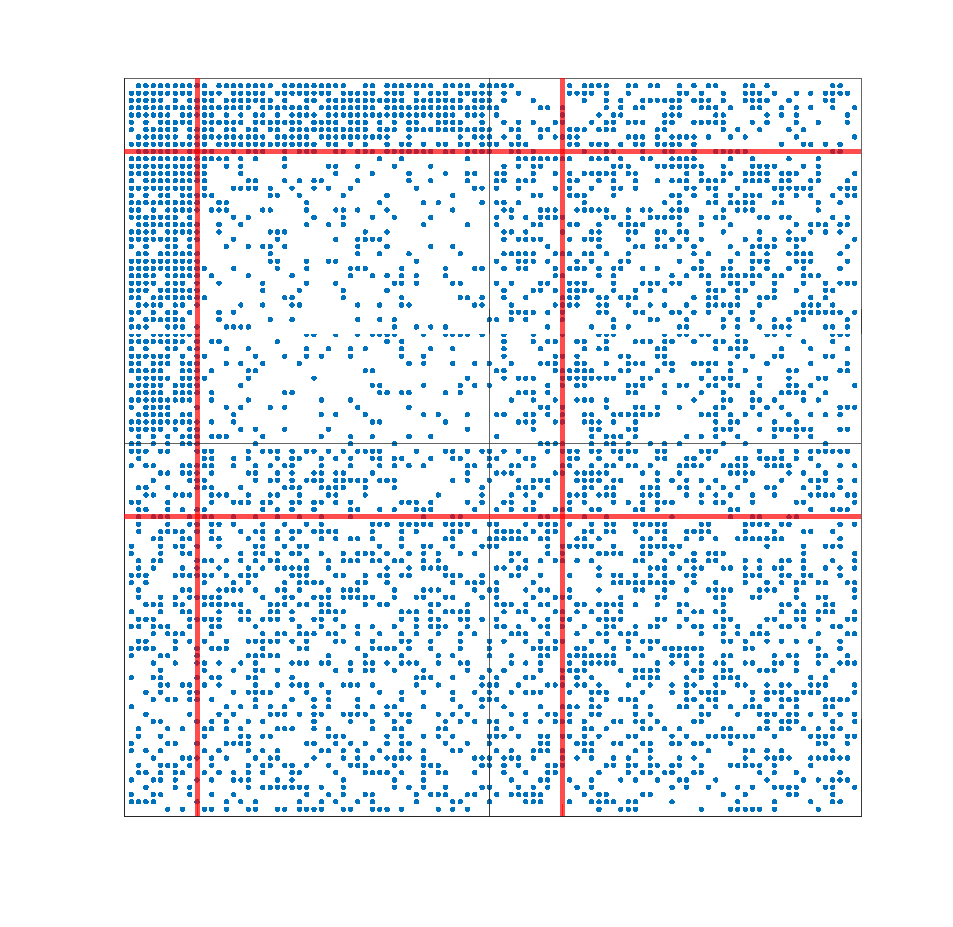}
	}
	\subfloat[Four noise layers, $q_{s^\ast_{\mathrm{node}}}=0.055$.]{
		\includegraphics[width=0.19\textwidth]{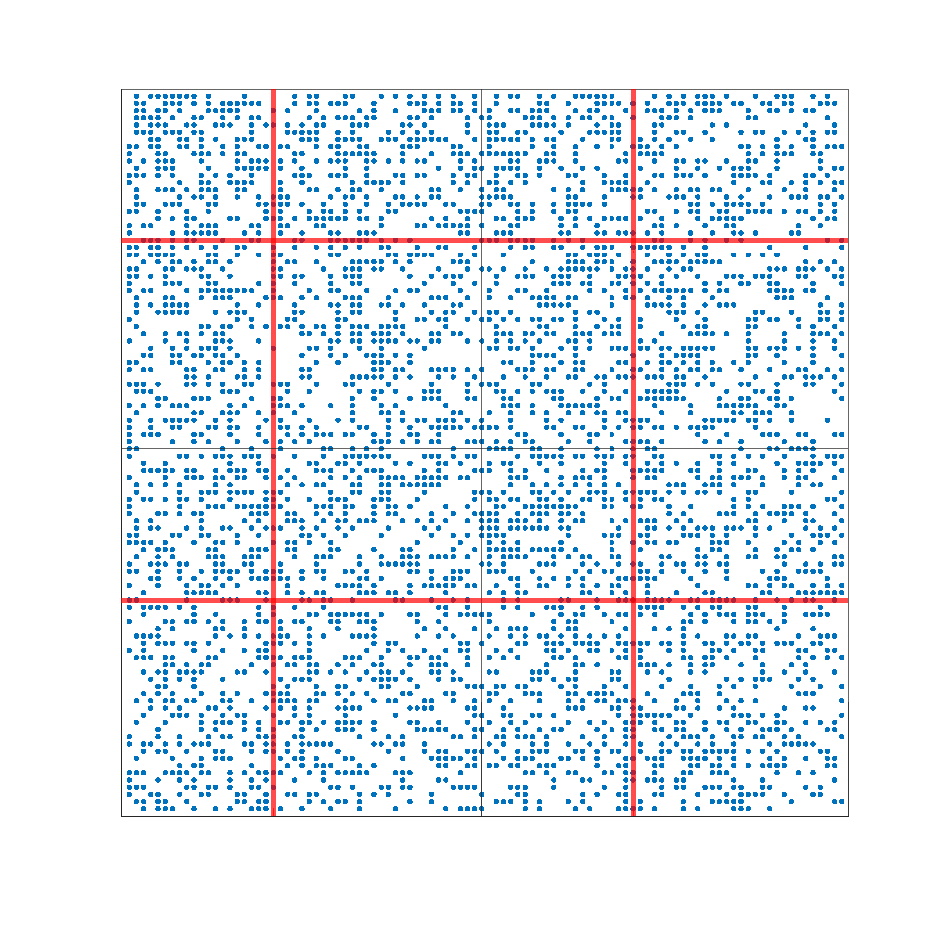}\label{fig:SBM_2_layers_spy_plots_all_noise}
	}
	\caption{Reordered supra-adjacency matrix plots for stochastic block model multilayer networks with $L=2$ layers and $n=50$ nodes.
		We vary the numbers of informative and noisy blocks in the supra-adjacency matrix and visualize the detected node core sizes by red lines.}\label{fig:SBM_2_layers_spy_plots}
\end{figure*}

In the first set of experiments, we choose $L=2$ layers with $n=50$ nodes each and assign a mixture of informative (with respect to the planted core-periphery structure) and noisy blocks in the corresponding supra-adjacency matrices.
For the informative layers, we prescribe a node core of size $10$ with prescribed edge probabilities $p_c=0.8$ and $p_p=0.2$.
For the noise layers, we choose $p_c=p_p=0.4$ to assign approximately equal numbers of edges to both informative and noisy layers.
\Cref{fig:SBM_2_layers_spy_plots} shows that the planted node core is perfectly recovered even if only one of the four blocks in the supra-adjacency matrix is informative.
In \Cref{fig:SBM_2_layers_spy_plots_all_noise}, we choose all blocks to be noisy and similarly to many community detection methods \cite{macmahon2015community} the nonlinear spectral method attempts to find the configuration that most closely resembles a core-periphery structure.
In all examples considered in \Cref{fig:SBM_2_layers_spy_plots}, however, the highest value $q_{s^\ast_{\mathrm{node}}}$ of the QUBO objective function \eqref{eq:QUBO_objective_node} obtained by \Cref{alg} provides an indication of the overall informativity of the multilayer network with respect to core-periphery structure.
Each replacement of an informative block by a noisy one significantly reduces the value $q_{s^\ast_{\mathrm{node}}}$ until it is close to $0$ in the setting of \Cref{fig:SBM_2_layers_spy_plots_all_noise}, which indicates a low confidence in the detected structure.

In the second set of experiments, we choose $L=4$ layers with $n=25$ nodes each and plant $5$ core nodes in each block.
In contrast to the first set of experiments, however, we plant a layer core by assigning the edge probabilities $p_c=0.8$ and $p_p=0.2$ only in a layer L-shape of core size $1$ (\Cref{fig:SBM_4_layers_spy_plots_1_core_layer_layer,fig:SBM_4_layers_spy_plots_1_core_layer_node}) or $2$ (\Cref{fig:SBM_4_layers_spy_plots_2_core_layers_layer,fig:SBM_4_layers_spy_plots_2_core_layers_node}) and reduce the edge probabilities in the remaining blocks representing peripheral blocks to $p_c=0.2$ and $p_p=0.05$.
\Cref{fig:SBM_4_layers_spy_plots} shows that both layer core sizes as well as the prescribed node core size are correctly detected.

\begin{figure*}[t]
	\subfloat[Layer core size $s^\ast_{\mathrm{layer}}=1$ for one core layer.]{
		\includegraphics[width=0.24\textwidth]{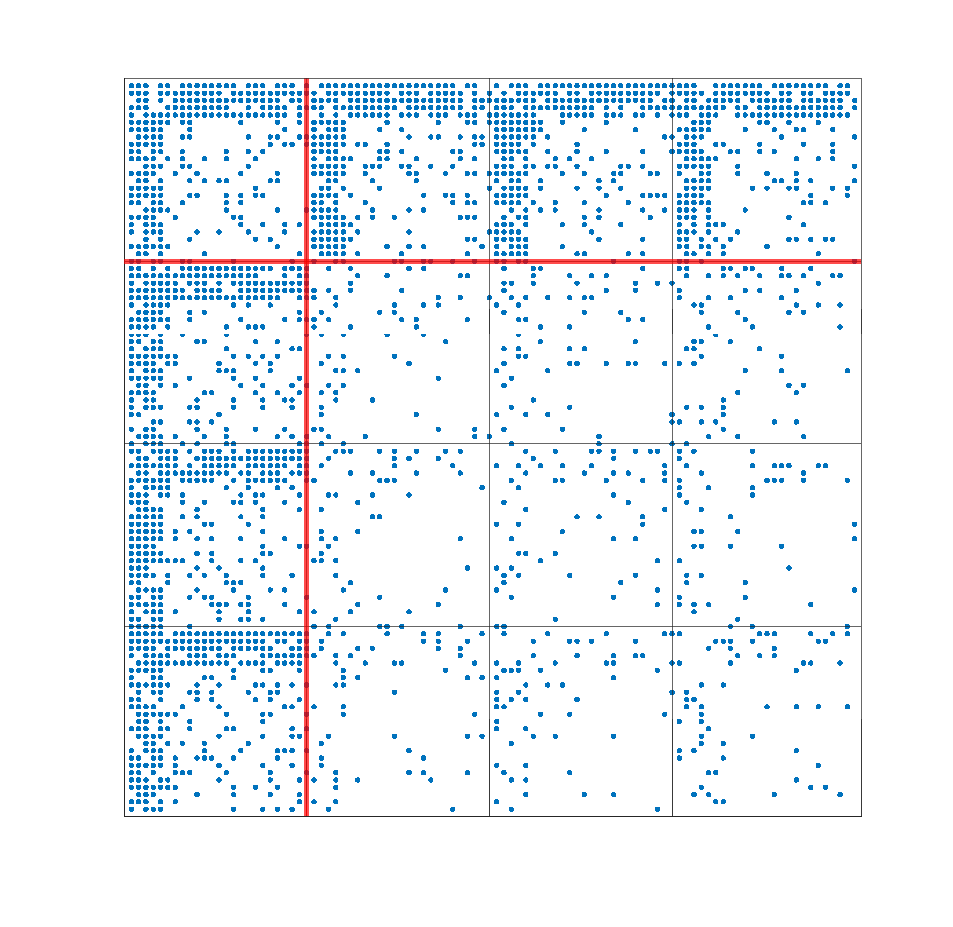}\label{fig:SBM_4_layers_spy_plots_1_core_layer_layer}
	}
	\subfloat[Node core size $s^\ast_{\mathrm{node}}=5$ for one core layer.]{
		\includegraphics[width=0.24\textwidth]{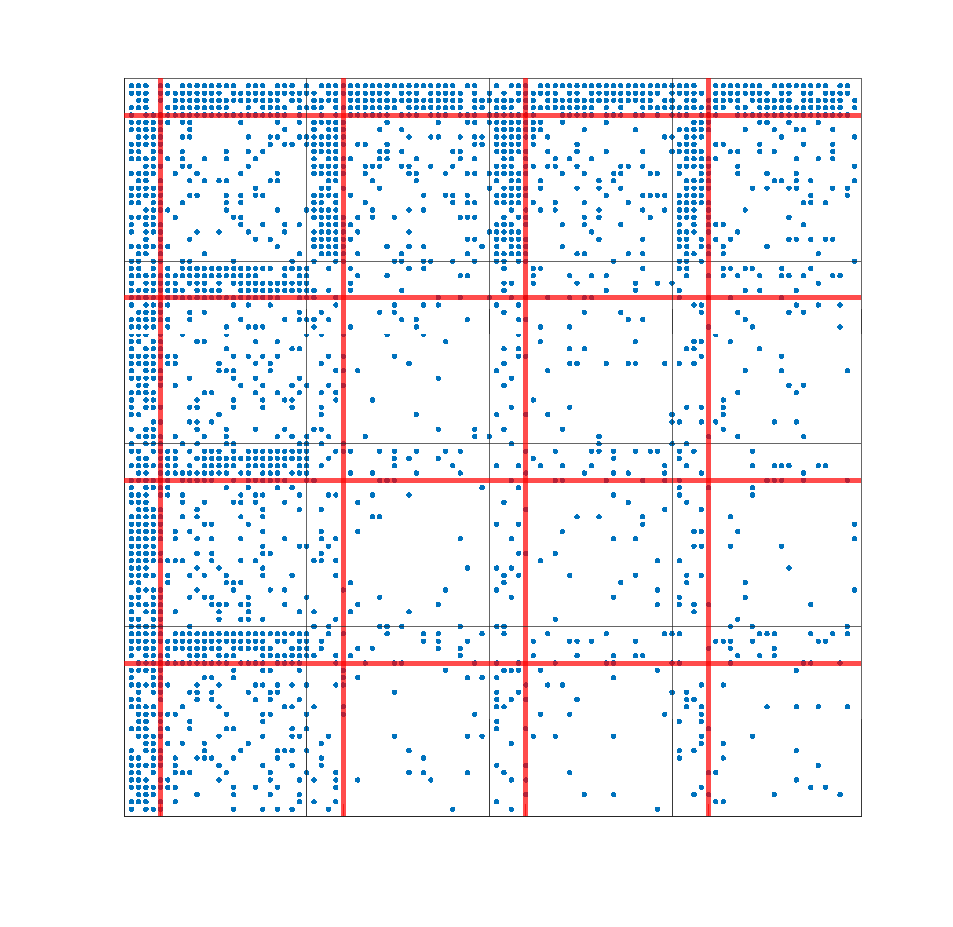}\label{fig:SBM_4_layers_spy_plots_1_core_layer_node}
	}
	\subfloat[Layer core size $s^\ast_{\mathrm{layer}}=2$ for two core layers.]{
		\includegraphics[width=0.24\textwidth]{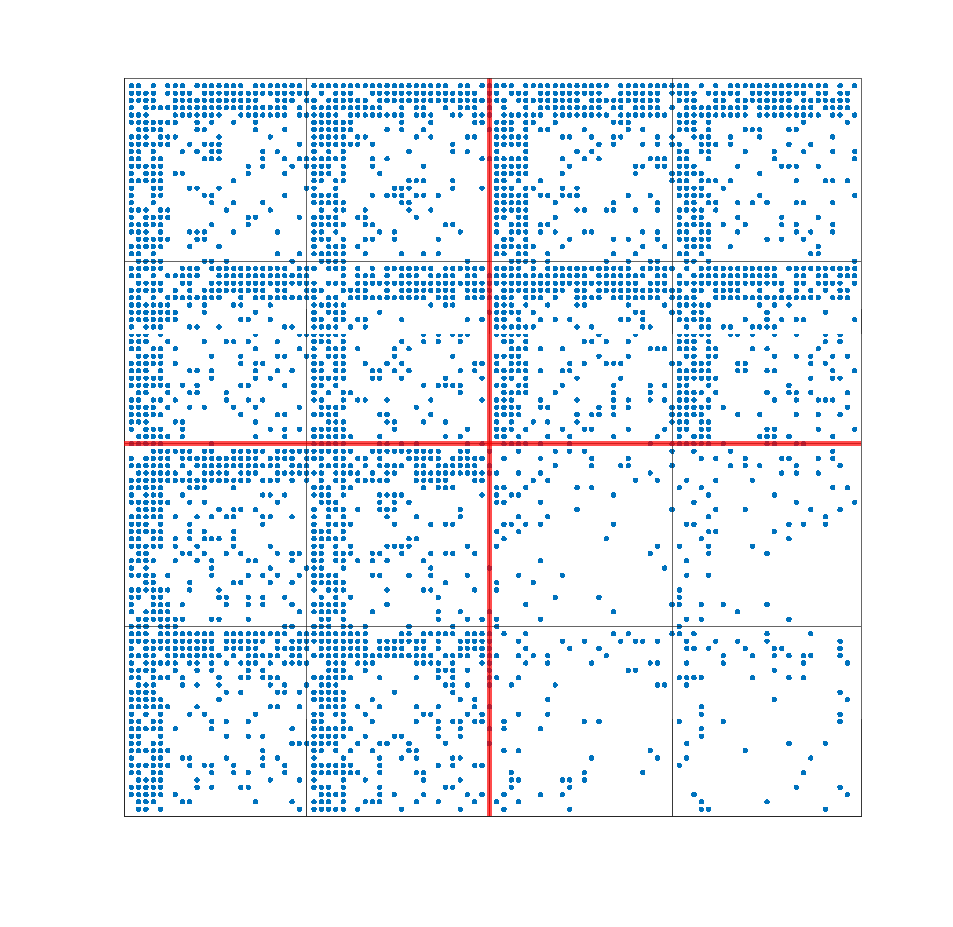}\label{fig:SBM_4_layers_spy_plots_2_core_layers_layer}
	}
	\subfloat[Node core size $s^\ast_{\mathrm{node}}=5$ for two core layers.]{
		\includegraphics[width=0.24\textwidth]{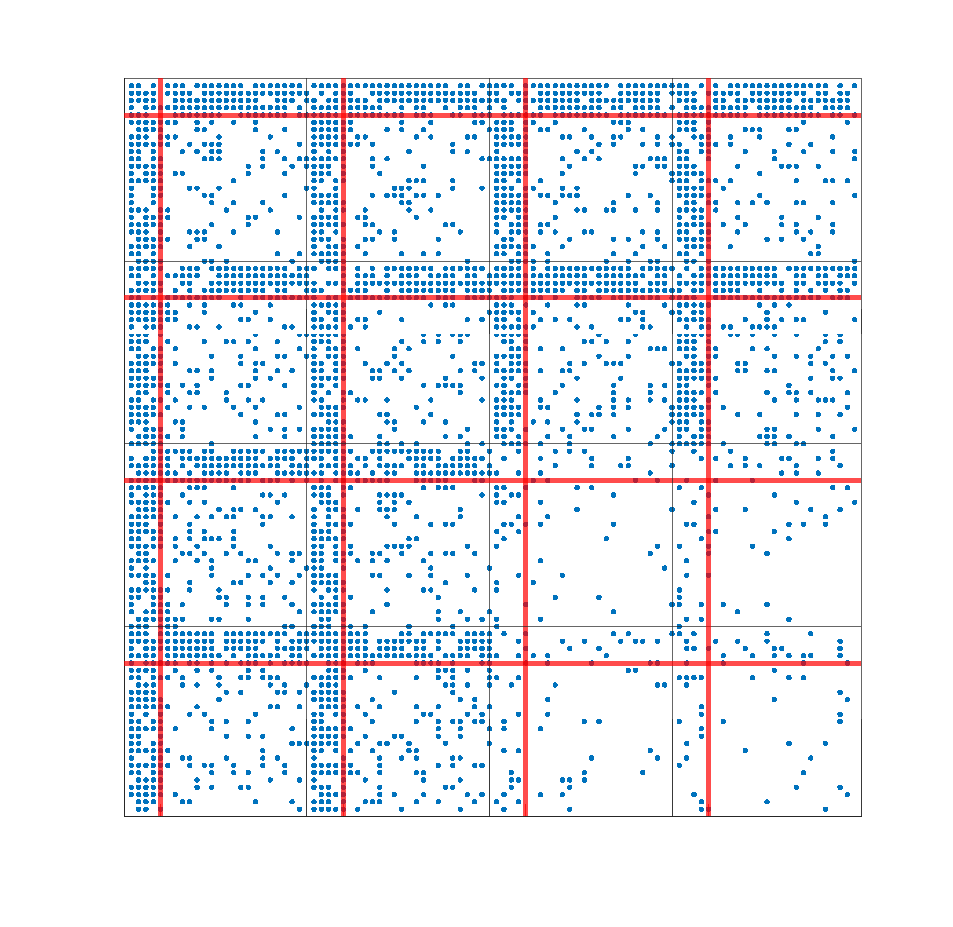}\label{fig:SBM_4_layers_spy_plots_2_core_layers_node}
	}
	\caption{Reordered supra-adjacency matrix plots for stochastic block model multilayer networks with $L=4$ layers and $n=25$ nodes.
		We plant a layer core of size $1$ in panels (a) and (b) and of size $2$ in panels (c) and (d) and visualize the detected layer core sizes in panels (a) and (c) and the detected node core sizes in panels (b) and (d) by red lines.}\label{fig:SBM_4_layers_spy_plots}
\end{figure*}

\section{Additional results on OpenAlex networks}\label{sec:results_OA}

Figure \ref{fig:OA_spy_plots} shows reordered supra-adjacency matrices for a weighted version of the OpenAlex multilayer citation network of complex network scientists of the year $2023$.
\Cref{tab:OA_top_19_rankings_weighted} additionally lists the corresponding top $19$ rankings of authors and disciplines according to node and layer coreness scores, respectively, for the two parameter settings $p=q=22$ and $p=q=2$.
The number $19$ is chosen since it corresponds to the number of layers in the multilayer network.
\Cref{tab:OA_top_19_rankings_weighted} shows slight differences in rankings between the two parameter settings.
Moreover, the coreness vector entries decay much more rapidly for the parameters $p=q=2$ since vector normalization in $p$- and $q$-norm has the effect of increasing the vector entries more strongly towards the value $1$ for larger values of $p$ and $q$.

\begin{table*}
	\centering
	\begin{tabular}{clc@{\hskip 3mm}lc@{\hskip 6mm}lc@{\hskip 3mm}lc}
		\hline\hline
		&\multicolumn{4}{c}{$p=q=22$}&\multicolumn{4}{c}{$p=q=2$}\\
		Rank & \multicolumn{2}{c}{Author (node)} & \multicolumn{2}{c}{Discipline (layer)} & \multicolumn{2}{c}{Author (node)} & \multicolumn{2}{c}{Discipline (layer)}\\\hline\hline
		1 & M. E. J. Newman & 0.887 & Computer science & 0.965 & M. E. J. Newman & 0.675 & Computer science & 0.907\\
		2 & Albert-László Barabási & 0.870 & Mathematics & 0.949 & Albert-László Barabási & 0.412 & Mathematics & 0.407\\
		3 & Réka Albert & 0.860 & Physics & 0.917 & Réka Albert & 0.328 & Physics &  0.106\\
		4 & Steven H. Strogatz & 0.850 & Engineering & 0.858 & Steven H. Strogatz & 0.246 & Biology & 0.014\\
		5 & Duncan J. Watts & 0.842 & Biology & 0.841 & Duncan J. Watts & 0.208 & Engineering & 0.008\\
		6 & Alessandro Vespignani & 0.819 & Geography & 0.802 & Guanrong Chen & 0.117 & Psychology & 0.003\\
		7 & Rom.\ Pastor-Satorras & 0.819 & Psychology & 0.801 & Rom.\ Pastor-Satorras & 0.114 & Geography & 0.002\\
		8 & Guanrong Chen & 0.819 & Sociology & 0.764 & Alessandro Vespignani & 0.114 & Sociology & 0.001\\
		9 & Shlomo Havlin & 0.814 & Economics & 0.762 & Michelle Girvan & 0.106 & Medicine & $4\cdot 10^{-4}$\\
		10 & Michelle Girvan & 0.809 & Business & 0.760 & Shlomo Havlin & 0.101 & Business & $3\cdot 10^{-4}$\\
		11 & Marc Barthélemy & 0.804 & Medicine & 0.746 & Vito Latora & 0.082 & Economics & $2\cdot 10^{-4}$\\
		12 & Vito Latora & 0.803 & Materials science & 0.727 & Marc Barthélemy & 0.080 & Materials science & $7\cdot 10^{-5}$\\
		13 & Hawoong Jeong & 0.797 & Political science & 0.716 & Hawoong Jeong & 0.068 & Political science & $4\cdot 10^{-5}$\\
		14 & Yamir Moreno & 0.795 & Chemistry & 0.688 & Yamir Moreno & 0.064 & Chemistry & $2\cdot 10^{-5}$\\
		15 & Ginestra Bianconi & 0.792 & Geology & 0.683 & Ginestra Bianconi & 0.059 & Geology & $1\cdot 10^{-5}$\\
		16 & H. Eugene Stanley & 0.791 & Philosophy & 0.648 & H. Eugene Stanley & 0.054 & Philosophy & $7\cdot 10^{-6}$\\
		17 & S. N. Dorogovtsev & 0.788 & Environm.\ science & 0.631 & Àlex Arenas & 0.053 & Environm.\ science & $2\cdot 10^{-6}$\\
		18 & J. F. F. Mendes & 0.787 & Art & 0.595 & S. N. Dorogovtsev & 0.052 & History & $3\cdot 10^{-7}$\\
		19 & Àlex Arenas & 0.785 & History & 0.575 & Tao Zhou & 0.051 & Art & $1\cdot 10^{-8}$\\\hline\hline
	\end{tabular}
	\caption{Top $L=19$ authors and disciplines for the weighted OpenAlex citation multilayer network of the year 2023 (corresponding to the results from Figure \ref{fig:OA_author_rankings}) including the respective coreness scores for the two parameter settings $p=q=22$ and $p=q=2$.}\label{tab:OA_top_19_rankings_weighted}
\end{table*}

A natural question to ask is what additional information is gained by modeling the weighted OpenAlex citation network as a multilayer network as described in the main text in comparison to analyzing the aggregated single-layer network with adjacency matrix entries $[\bm{A}_{\mathrm{agg}}]_{ij}=\sum_{k,l=1}^L \mathcal{A}_{ij}^{kl}$ for $i,j=1,\dots,n$.
Clearly, such an approach prohibits the detection and ranking of core layers, i.e., scientific disciplines.
For a comparison of the obtained node coreness vectors, we ran the single-layer nonlinear spectral method \cite{tudisco2019nonlinear} with parameters $\alpha=10$ and $p=22$ on the aggregated network, which is equivalent to running the nonlinear spectral method for multilayer networks on the same network with $L=1$.
\Cref{fig:OA_single_layer_NSM_spy} shows that a valid core-periphery reordering is obtained that bears a partial resemblance although no immediate correspondence to the blocks of the multilayer results displayed in Figure 1b in the main text.
The comparison of the top $19$ authors (nodes) in terms of their coreness scores displayed in \Cref{fig:OA_single_layer_NSM_rankings}, however, highlights strong differences to the author (node) rankings in the $p=q=22$ case in \Cref{tab:OA_top_19_rankings_weighted}.
Since aggregating the multilayer network corresponds to the layer coreness vector $\boldsymbol{c}=\boldsymbol{1}$ in the multilayer setting, it could be argued that the results shown in \Cref{fig:OA_single_layer_NSM} are biased in favor of authors (nodes) primarily publishing in disciplines (layers) considered peripheral by our multilayer approach.

\begin{figure*}
	\subfloat[Reordered aggregated adjacency matrix.]{
		\raisebox{-.555\height}{
			\includegraphics[width=0.4\textwidth]{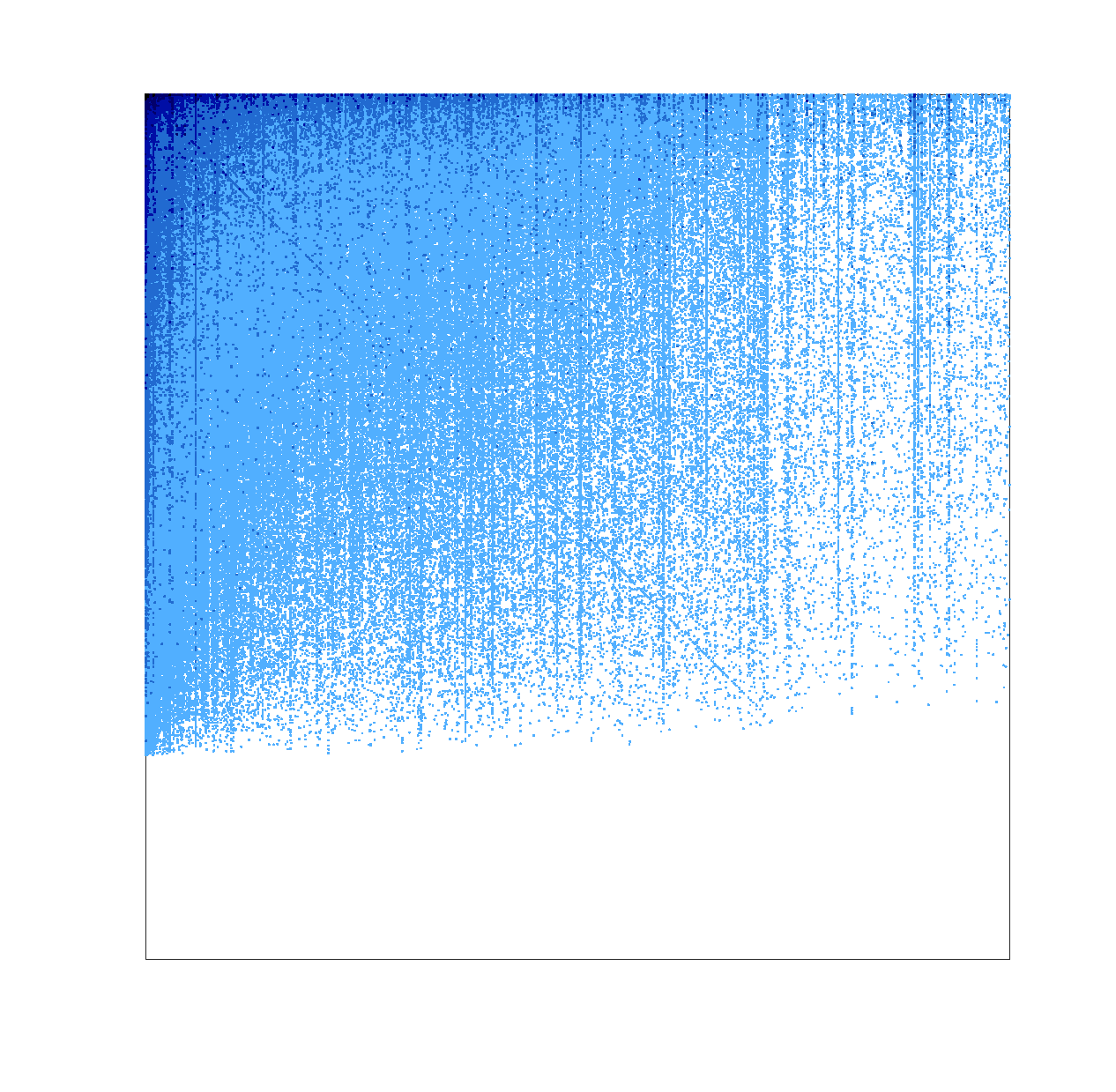}\label{fig:OA_single_layer_NSM_spy}
		}
	}
	\subfloat[Top $19$ authors.]{
		\begin{tabular}{clc}
			\hline\hline
			&\multicolumn{2}{c}{$p=22$}\\
			Rank & \multicolumn{2}{c}{Author (node)}\\\hline\hline
			1 & Shlomo Havlin & 0.806\\
			2 & Guanrong Chen & 0.796\\
			3 & Stefano Boccaletti & 0.795\\
			4 & H. Eugene Stanley & 0.791\\
			5 & Jürgen Kurths & 0.786\\
			6 & Tao Zhou & 0.785\\
			7 & Bing--Hong Wang & 0.782\\
			8 & Matjaž Perc & 0.781\\
			9 & Yamir Moreno & 0.779\\
			10 & Ginestra Bianconi & 0.777\\
			11 & Vito Latora & 0.775\\
			12 & Zhongzhi Zhang & 0.774\\
			13 & Jesús Gómez‐Gardeñes & 0.772\\
			14 & Ying--Cheng Lai & 0.771\\
			15 & Jianxi Gao & 0.768\\
			16 & Àlex Arenas & 0.767\\
			17 & Manlio De Domenico & 0.765\\
			18 & Wen Xu Wang & 0.764\\
			19 & Ming Tang & 0.763\\\hline\hline
		\end{tabular}\label{fig:OA_single_layer_NSM_rankings}
	}
	\caption{Plot of the reordered adjacency matrix and top $19$ core authors of the single-layer nonlinear spectral method \cite{tudisco2019nonlinear} applied to the aggregated weighted OpenAlex citation multilayer network.}\label{fig:OA_single_layer_NSM}
\end{figure*}

Furthermore, Figure \ref{fig:OA_author_rankings} shows top 10 author rankings by node coreness score over time for the weighted version of the OpenAlex multilayer citation network of complex network scientists of the year $2023$.
The weighting of edges from author-discipline pairs of a citing paper to the author-discipline pair of a cited paper is performed such that the total edge weight of one paper amounts to $1$.
This is achieved by additively inserting edge weights of $1$ over the product of the numbers of citing authors, citing disciplines, cited authors, and cited disciplines for each of the $38\,346$ recorded complex network science papers.
The largest observed edge weight in the weighted adjacency tensor is $17.59$.

In this section, we repeat the same experiments for ``unweighted'' edges, i.e., we simply replace the weight of $1$ over the product of the numbers of citing authors, citing disciplines, cited authors, and cited disciplines by $1$.
Note that this procedure still yields a weighted multilayer network with each entry in the adjacency tensor $\mathcal{A}$ representing the total number of citations from a given author in a given discipline to another author in a another discipline (self-loops by self-citations are included).
With this modeling approach, the total edge weight introduced by a given paper is larger when more authors and disciplines are involved in both the citing and cited paper.
The largest edge weight observed in the unweighted adjacency tensor is $876$.

\begin{figure*}
	\subfloat[Original supra-adjacency]{
		\includegraphics[width=0.19\textwidth]{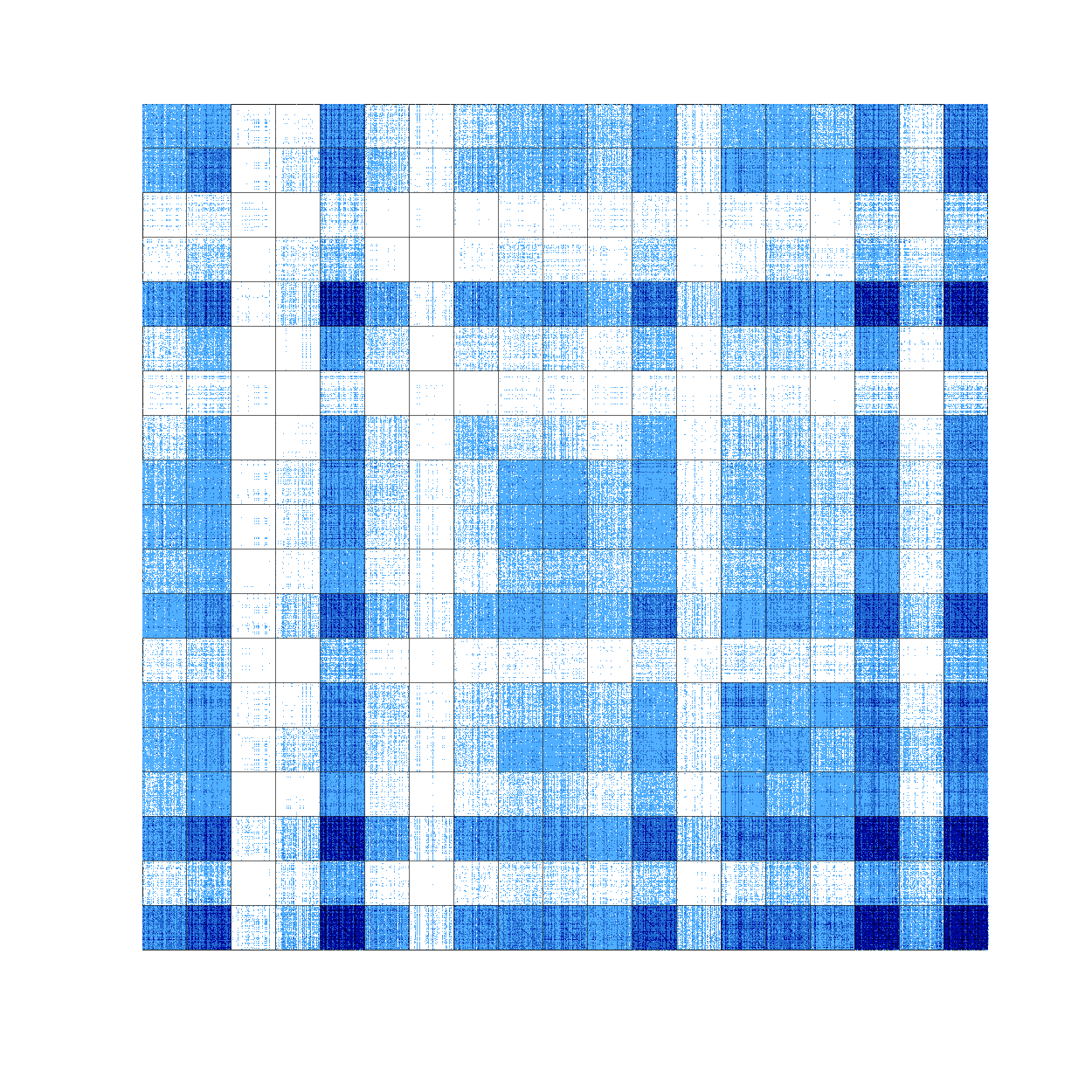}
	}
	\subfloat[Full reordered supra-adjacency for $p=q=22$]{
		\includegraphics[width=0.19\textwidth]{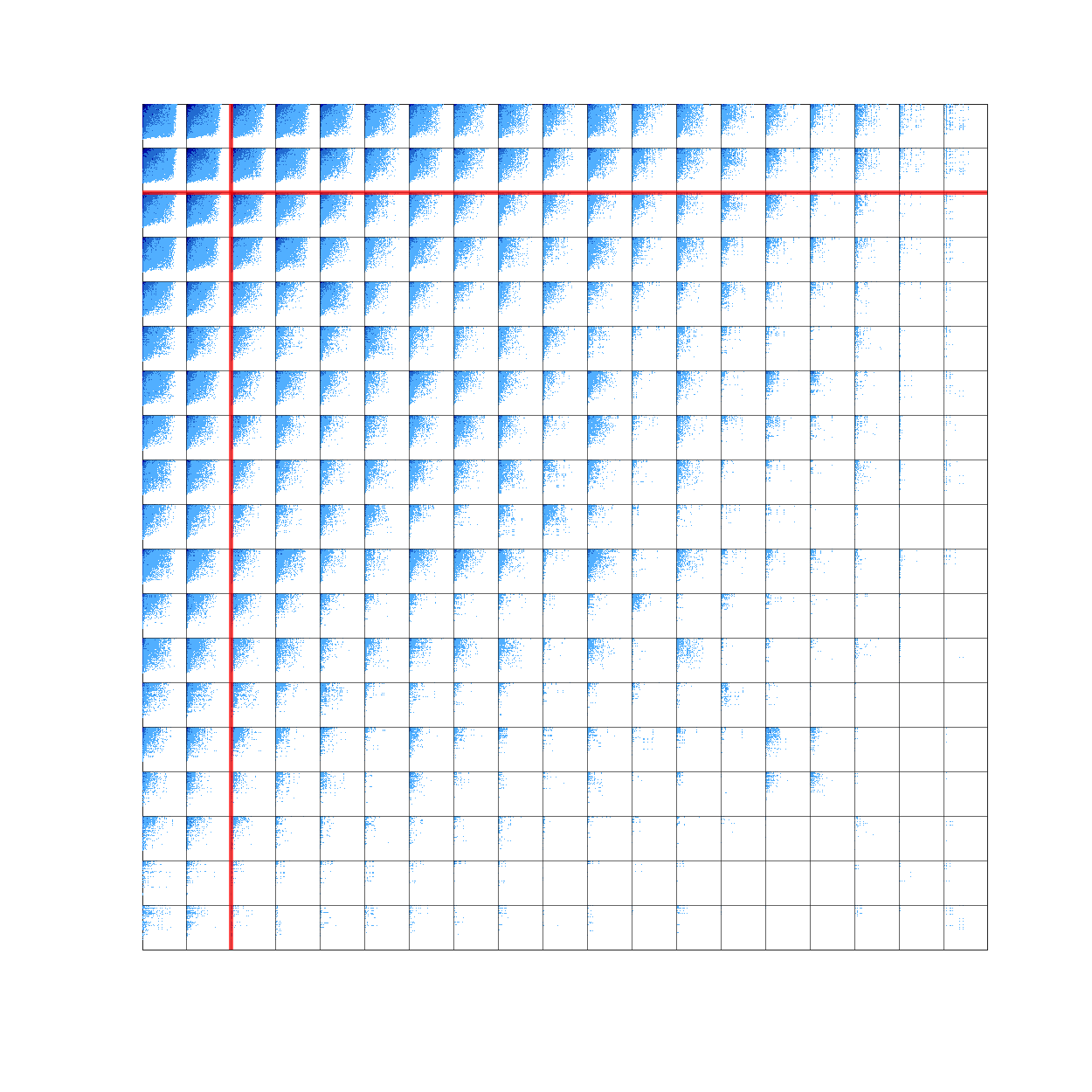}
	}
	\subfloat[Upper left $5\times 5$ blocks of reordered supra-adjacency for $p=q=22$]{
		\includegraphics[width=0.19\textwidth]{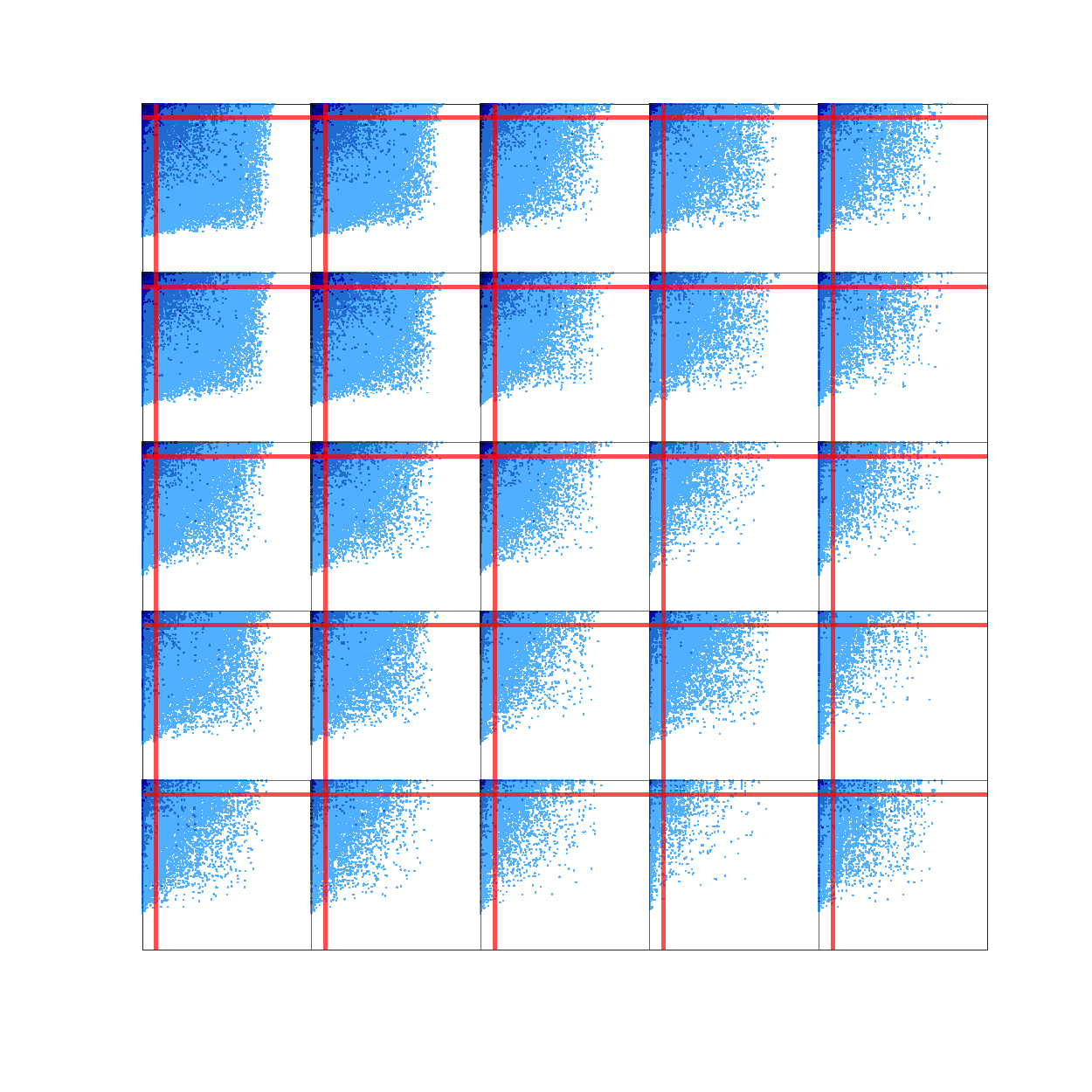}
	}
	\subfloat[Full reordered supra-adjacency for $p=q=2$]{
		\includegraphics[width=0.19\textwidth]{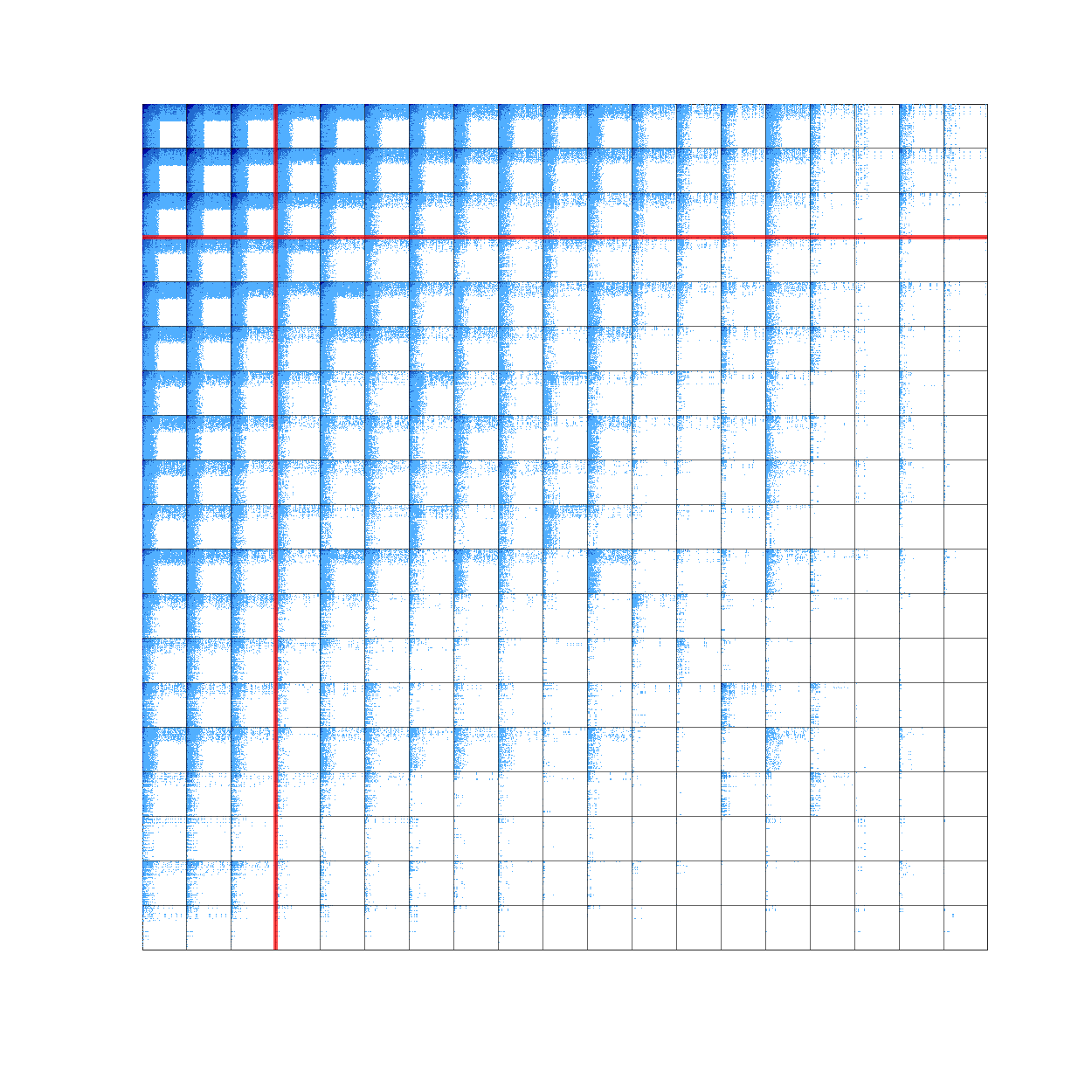}
	}
	\subfloat[Upper left $5\times 5$ blocks of reordered supra-adjacency for $p=q=2$]{
		\includegraphics[width=0.19\textwidth]{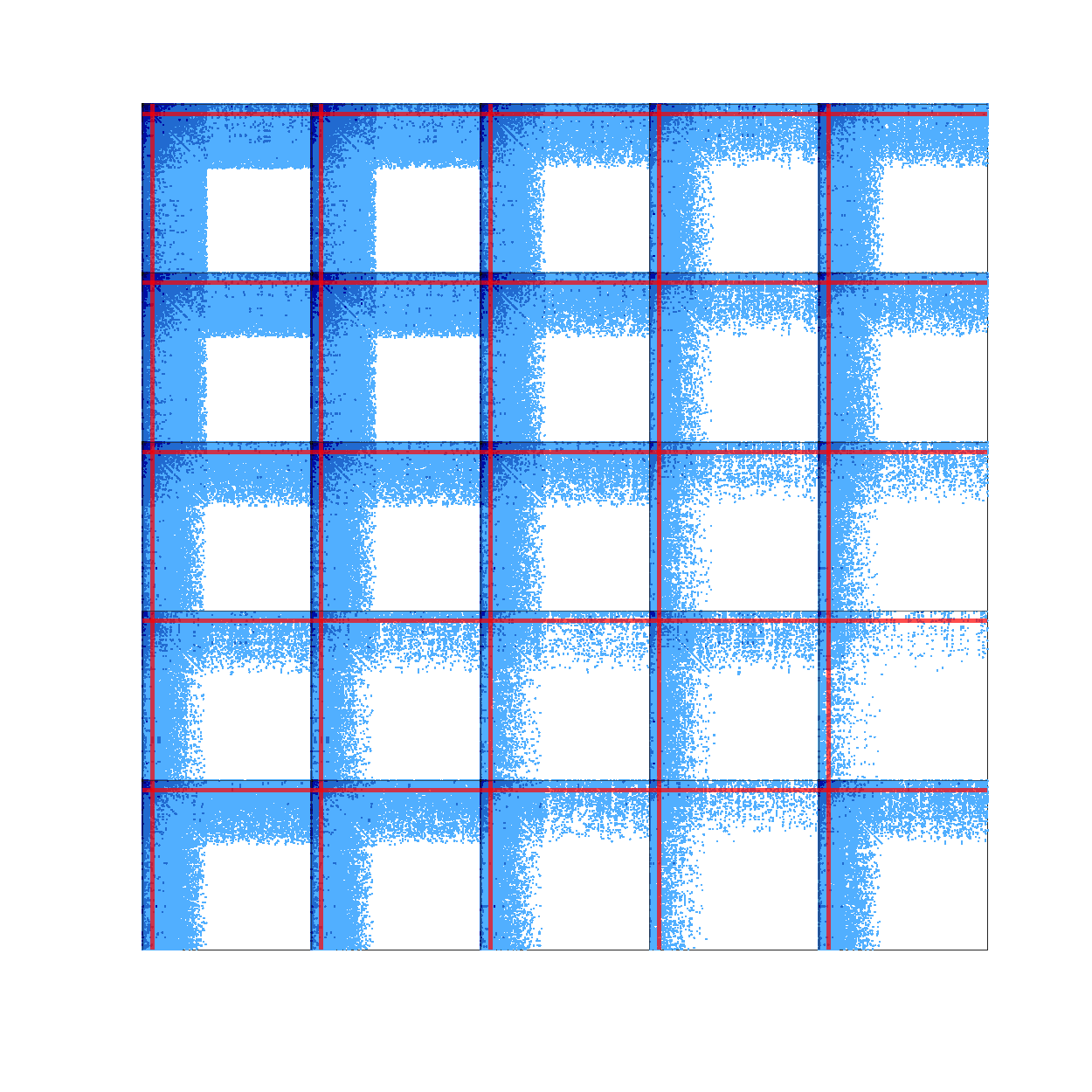}
	}
	\caption{Supra-adjacency matrix plots for the OpenAlex multilayer citation network of complex network scientists of the year 2023 in which individual citations are not weighted by the numbers of authors and disciplines.
		Dark fonts indicate large edge weights.
		Red lines in panels b) and d) indicate the layer core sizes $s^\ast_{\mathrm{layer}}=2$ and $s^\ast_{\mathrm{layer}}=3$, respectively, and red lines in panels c) and e) indicate the node core size $s^\ast_{\mathrm{node}}=4\,398$ and $s^\ast_{\mathrm{node}}=3\,051$, respectively, in each block.}\label{fig:OA_spy_plots_unweighted}
\end{figure*}

\begin{table*}
	\centering
	\begin{tabular}{clc@{\hskip 3mm}lc@{\hskip 6mm}lc@{\hskip 3mm}lc}
		\hline\hline
		&\multicolumn{4}{c}{$p=q=22$}&\multicolumn{4}{c}{$p=q=2$}\\
		Rank & \multicolumn{2}{c}{Author (node)} & \multicolumn{2}{c}{Discipline (layer)} & \multicolumn{2}{c}{Author (node)} & \multicolumn{2}{c}{Discipline (layer)}\\\hline\hline
		1 & Albert-László Barabási & 0.856 & Computer science & 0.959 & Albert-László Barabási & 0.452 & Computer science & 0.866\\
		2 & Réka Albert & 0.842 & Mathematics & 0.946 & Réka Albert & 0.330 & Mathematics & 0.456\\
		3 & M. E. J. Newman & 0.838 & Physics & 0.923 & M. E. J. Newman & 0.310 & Physics & 0.200\\
		4 & Shlomo Havlin & 0.835 & Engineering & 0.871 & Shlomo Havlin & 0.265 & Biology & 0.034\\
		5 & Steven H. Strogatz & 0.830 & Biology & 0.860 & Steven H. Strogatz & 0.239 & Engineering & 0.020\\
		6 & Duncan J. Watts & 0.823 & Psychology & 0.825 & Duncan J. Watts & 0.204 & Geography & 0.007\\
		7 & H. Eugene Stanley & 0.821 & Geography & 0.822 & H. Eugene Stanley & 0.179 & Psychology & 0.005\\
		8 & Alessandro Vespignani & 0.817 & Economics & 0.790 & Alessandro Vespignani & 0.169 & Economics & 0.002\\
		9 & Yamir Moreno & 0.817 & Sociology & 0.786 & Yamir Moreno & 0.168 & Sociology & 0.002\\
		10 & Guanrong Chen & 0.815 & Medicine & 0.782 & Guanrong Chen & 0.167 & Medicine & 0.001\\
		11 & Stefano Boccaletti & 0.814 & Business & 0.781 & Vito Latora & 0.159 & Business & $6\cdot 10^{-4}$\\
		12 & Vito Latora & 0.814 & Materials science & 0.751 & Stefano Boccaletti & 0.157 & Materials science & $4\cdot 10^{-4}$\\
		13 & Rom.\ Pastor-Satorras & 0.813 & Political science & 0.730 & Rom.\ Pastor-Satorras & 0.152 & Chemistry & $2\cdot 10^{-4}$\\
		14 & Hawoong Jeong & 0.799 & Chemistry & 0.719 & Hawoong Jeong & 0.109 & Geology & $1\cdot 10^{-4}$\\
		15 & Marc Barthélemy & 0.799 & Geology & 0.709 & Marc Barthélemy & 0.105 & Political science & $4\cdot 10^{-5}$\\
		16 & Àlex Arenas & 0.797 & Environm.\ science & 0.656 & Àlex Arenas & 0.105 & Environm.\ science & $1\cdot 10^{-5}$\\
		17 & Jürgen Kurths & 0.791 & Philosophy & 0.643 & Jürgen Kurths & 0.089 & History & $4\cdot 10^{-6}$\\
		18 & Sergey V. Buldyrev & 0.790 & Art & 0.588 & Tao Zhou & 0.088 & Philosophy & $3\cdot 10^{-6}$\\
		19 & Tao Zhou & 0.790 & History & 0.586 & Mario Chávez & 0.083 & Art & $5\cdot 10^{-8}$\\\hline\hline
	\end{tabular}
	\caption{Top $L=19$ authors and disciplines for the unweighted OpenAlex citation multilayer network of the year 2023 (corresponding to the results from \Cref{fig:OA_auhor_rankings_unweighted}) including the respective coreness scores for the two parameter settings $p=q=22$ and $p=q=2$.}\label{tab:OA_top_19_rankings_unweighted}
\end{table*}

\Cref{fig:OA_spy_plots_unweighted} repeats the experiments from Figure \ref{fig:OA_spy_plots} for the unweighted OpenAlex citation network.
It shows the original supra-adjacency matrix as well as reordered supra-adjacency matrices and their principal $5\times 5$ blocks for the parameter choices $p=q=22$ and $p=q=2$.
Again, for the choice $\alpha=\beta=10$ the parameters $p=q=22$ guarantee convergence to the global optimum of the objective function $f_{\alpha,\beta}(\bm{x},\bm{c})$ while $p=q=2$ lead to locally optimal solutions for initial conditions $\bm{x}_0=\bm{1}\in\R^n$ and $\bm{c}_0=\bm{1}\in\R^L$.
We observe differences in the row and column permutations as well as the core-periphery assignment of nodes and layers in comparison to the weighted case.
However, since the color code has been adapted according to the different scale of edge weights, the plots in \Cref{fig:OA_spy_plots_unweighted} are visually similar to the weighted results presented in Figure \ref{fig:OA_spy_plots}.
In the parameter setting $p=q=22$, the $s^\ast_{\mathrm{layer}}=2$ disciplines ``Computer Science'' and ``Mathematics'' are detected as the layer core while $s^\ast_{\mathrm{node}}=4\,398$ authors are assigned to the node core.
For $p=q=2$, the $s^\ast_{\mathrm{layer}}=3$ disciplines ``Computer Science'', ``Mathematics'', and ``Physics'' are detected as the layer core while $s^\ast_{\mathrm{node}}=3\,051$ authors are assigned to the node core.

Additionally, \Cref{tab:OA_top_19_rankings_unweighted} again lists the top $19$ authors and disciplines of the unweighted OpenAlex multilayer citation network of complex network scientists of the year $2023$ for the two parameters $p=q=22$ and $p=q=2$.
Similarly to the weighted case, we observe slight deviations in rankings and expectedly strong differences in the magnitude of coreness vector entries between the two parameter settings.

Finally, \Cref{fig:OA_auhor_rankings_unweighted} shows similar deviations from the weighted case in the ranking of top 10 node/author coreness scores over the years $2000$ to $2023$.
The absence of weighted edges per paper favors authors citing and cited by papers with relatively high numbers of authors and disciplines.

\pgfkeys{/pgf/number format/.cd,1000 sep={}}
\begin{figure*}
	\includegraphics[width=.99\textwidth]{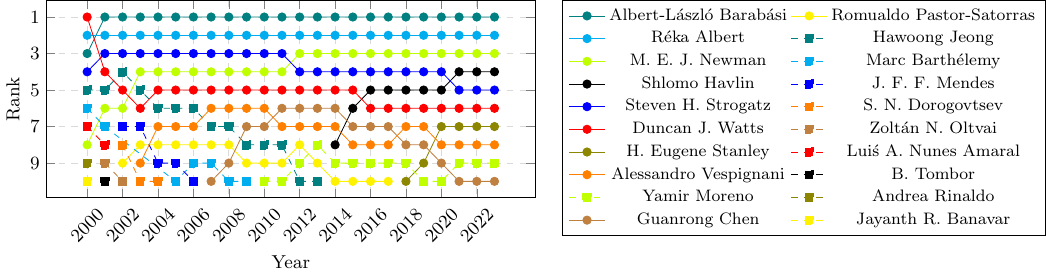}
	\caption{Top $10$ authors by node coreness score in the unweighted OpenAlex citation multilayer network over the years $2000$ to $2023$ for the parameters $p=q=22$.}\label{fig:OA_auhor_rankings_unweighted}
\end{figure*}

\section{Additional results on EUAir network}\label{sec:results_EUAir}

\begin{figure}[b]
	\begin{center}
		\includegraphics[width=0.49\textwidth,clip,trim=330pt 80pt 290pt 90pt]{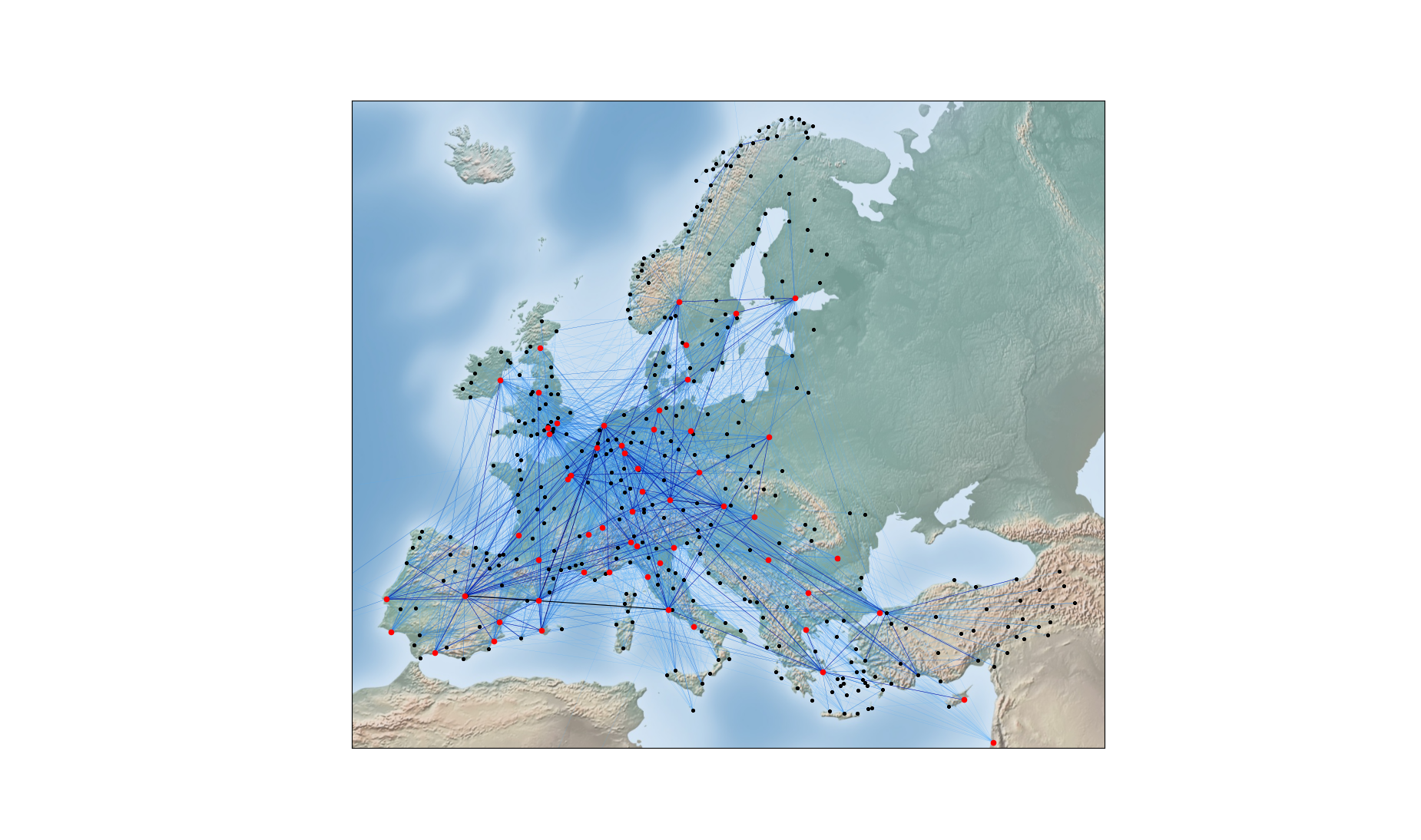}
	\end{center}
	\caption{Plot of the aggregated European Airlines network.
		Core airports are marked red while periphery airports are marked black.
		Dark fonts indicate large edge weights.
		The figure was created with matplotlib's basemap library.}\label{fig:european_airlines2}
\end{figure}

Figure \ref{fig:european_airlines2} shows a map of the European Airlines (EUAir) network in which core nodes detected by the nonlinear spectral method for multilayer networks with parameters $p=q=22$ are marked red and periphery nodes are marked black.
Additionally, edges of the aggregated network with adjacency matrix entries $[\bm{A}_{\mathrm{agg}}]_{ij}=\sum_{l=1}^L \mathcal{A}_{ij}^{ll}$ for $i,j=1,\dots,n$ larger than one are shown with a color coding assigning darker fonts to larger edge weights.
This means that for better visibility, only pairs of airports are connected by an edge in Figure \ref{fig:european_airlines} when at least two airlines offer a flight connection between them.
For completeness, \Cref{fig:european_airlines2} shows the same map plot also including edges with weight $1$.
The largest edge weight in the aggregated network is $5$ indicating that the airports Leonardo da Vinci–Fiumicino and Adolfo Suárez Madrid–Barajas are connected by $5$ different airlines.

While \Cref{sec:numerics_EUAir} reports the $s^\ast_{\mathrm{layer}}=4$ core airlines ``Lufthansa'', ``easyJet'', ``Ryanair'', and ``Air Berlin'', we list the $s^\ast_{\mathrm{node}}=57$ core airports in \Cref{tab:euair_airport_core}.

Again, we study what additional information is gained by modeling the EUAir network as a general multiplex network as described in \Cref{sec:numerics_EUAir} in comparison to analyzing the aggregated single-layer network.
Running the singe-layer nonlinear spectral method \cite{tudisco2019nonlinear} with parameters $\alpha=10$ and $p=22$ on the aggregated network, which is equivalent to running the nonlinear spectral method for multilayer networks on the same network with $L=1$ yields a very similar node core size of $s^\ast_{\mathrm{node}}=59$.
The ranking of some airports in terms of node coreness scores is also comparable.
However, the available layer, i.e., airline information is lost in aggregation.
Consequently, the single-layer approach yields no layer coreness scores and does not allow the determination of the airline core.
Furthermore, our modeling approach of inter-layer edges reflects information on possible changes of airlines, the presence of which should arguably make an airport more likely to belong to the core.
A concrete example are the node coreness rankings of the airports London Gatwick and London Stansted.
While London Gatwick occupies rank $39$ in the multilayer case, it ranks $6$th in the single-layer nonlinear spectral method applied to the aggregated network.
Similarly, London Stansted occupies rank $45$ in the multilayer case while it ranks $5$th in the single-layer case.
The reason is that although London Gatwick is very well-connected within the easyJet layer and London Stansted is very well-connected within the Ryanair layer, both airports are served by relatively few airlines.
In particular, London Gatwick permits connections to $10$ different airlines while London Stansted permits only $6$.
In comparison, the top $10$ airports in the multilayer case offer connections to between $23$ to $30$ different airlines.

Note that information on airline changes is also impossible to encode in a third-order tensor model for which a nonlinear spectral method for multiplex networks is available in the literature \cite{bergermann2024nonlinear}.

\begin{table*}
	\begin{tabular}{cl@{\hskip 20pt}cl@{\hskip 20pt}cl}
		\hline\hline
		Rank & Airport & Rank & Airport & Rank & Airport\\\hline\hline
		1&Barcelona Int.\ &20&Václav Havel  Prague&39&London Gatwick \\
		2&Leonardo da Vinci–Fiumicino &21&Stockholm-Arlanda &40&Marseille Provence \\
		3&Amsterdam  Schiphol&22&Berlin-Tegel &41&Lyon Saint-Exupéry \\
		4&Malpensa Int.\ &23&Geneva Cointrin Int.\ &42&Cologne Bonn \\
		5&Adolfo Suárez Madrid–Barajas &24&Ben Gurion Int.\ &43&Gothenburg-Landvetter \\
		6&Munich &25&London Heathrow &44&Valencia \\
		7&Brussels &26&Málaga &45&London Stansted \\
		8&Vienna Int.\ &27&Helsinki Vantaa &46&Belgrade Nikola Tesla \\
		9&Frankfurt am Main &28&Oslo Gardermoen &47&Thessaloniki Macedonia Int.\ \\
		10&Charles de Gaulle Int.\ &29&Stuttgart &48&Naples Int.\ \\
		11&Eleftherios Venizelos Int.\ &30&Henri Coandă Int.\ &49&Alicante Int.\ \\
		12&Copenhagen Kastrup &31&Humberto Delgado (Lisbon) &50&Toulouse-Blagnac \\
		13&Düsseldorf &32&Dublin &51&Edinburgh \\
		14&Zürich &33&Atatürk Int.\ &52&Paris-Orly \\
		15&Warsaw Chopin &34&Bologna Guglielmo Marconi &53&Hannover \\
		16&Venice Marco Polo &35&Palma De Mallorca &54&Faro \\
		17&Budapest Liszt Ferenc Int.\ &36&Manchester &55&Pisa Int.\ \\
		18&Hamburg &37&Sofia &56&Bordeaux-Mérignac \\
		19&Nice-Côte d'Azur &38&Milano Linate &57&Larnaca Int.\ \\\hline\hline
	\end{tabular}
	\caption{Ranking of core nodes (airports) in the EUAir multiplex network for $p=q=22$.}\label{tab:euair_airport_core}
\end{table*}

\section{Additional results on WIOD networks}\label{sec:results_WIOD}

\begin{figure*}
	\begin{center}
		\subfloat[Original supra-adjacency]{
			\includegraphics[width=0.32\textwidth]{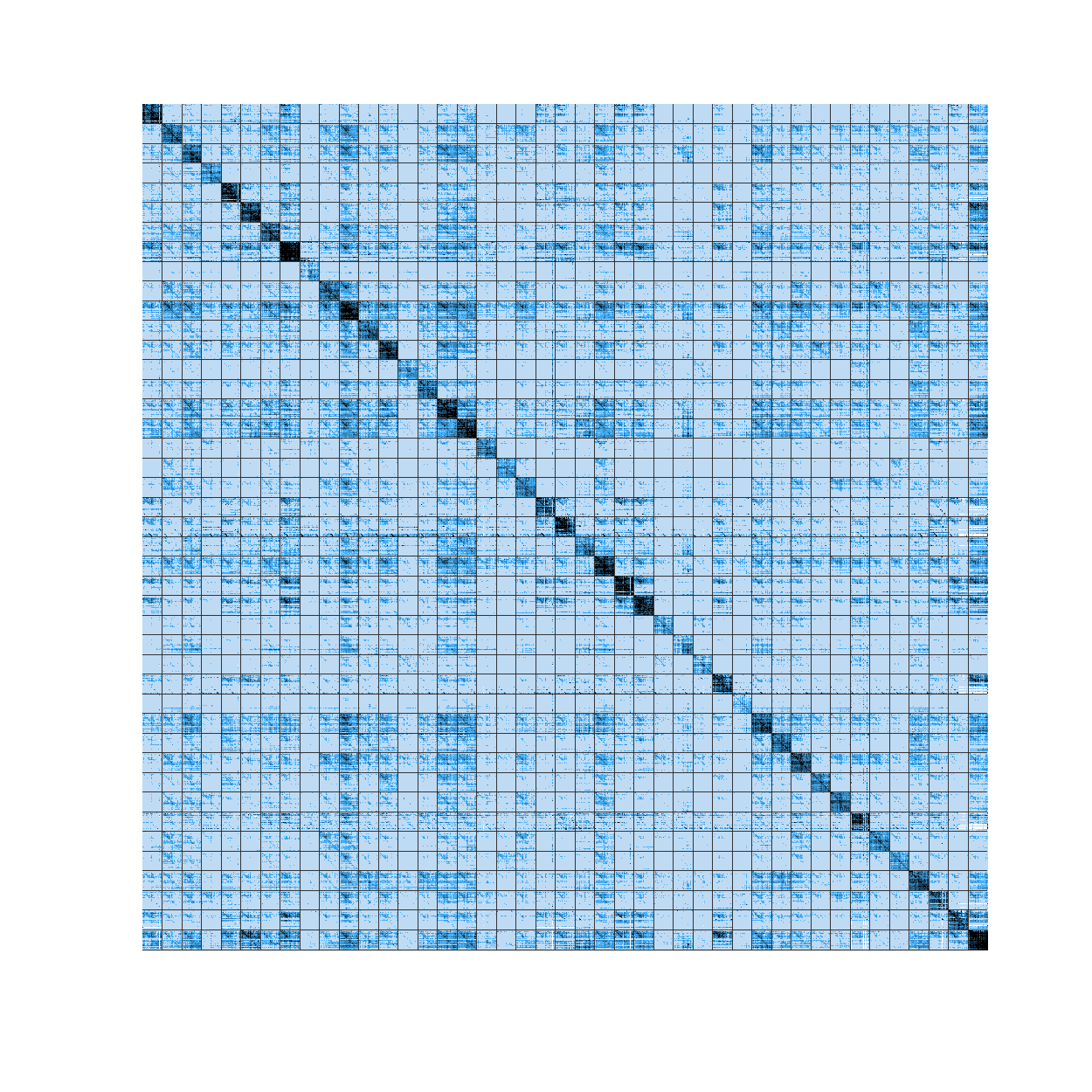}
		}
		\subfloat[Full reordered supra-adjacency]{
			\includegraphics[width=0.32\textwidth]{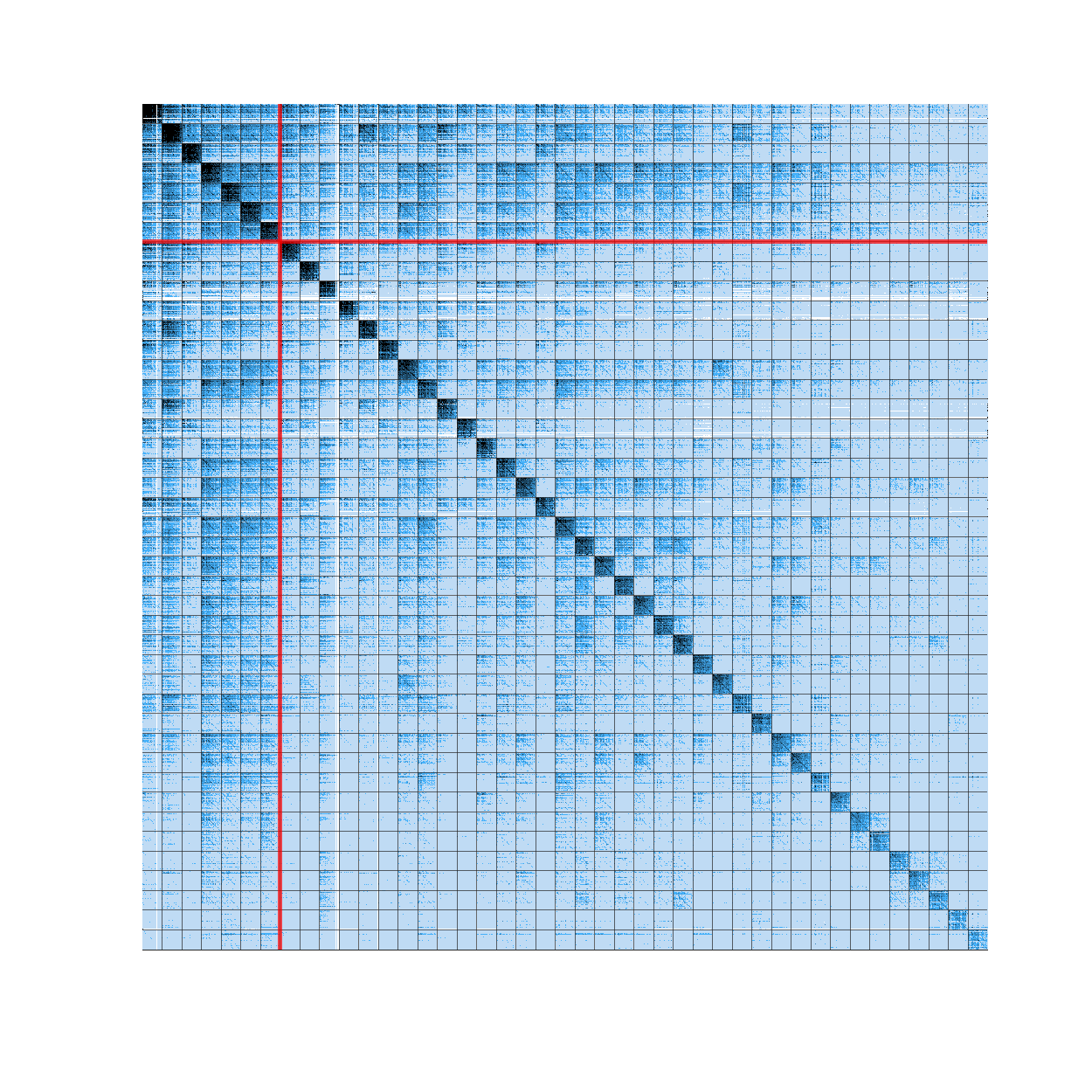}
		}
		\subfloat[Upper left $10\times 10$ blocks of reordered supra-adjacency]{
			\includegraphics[width=0.32\textwidth]{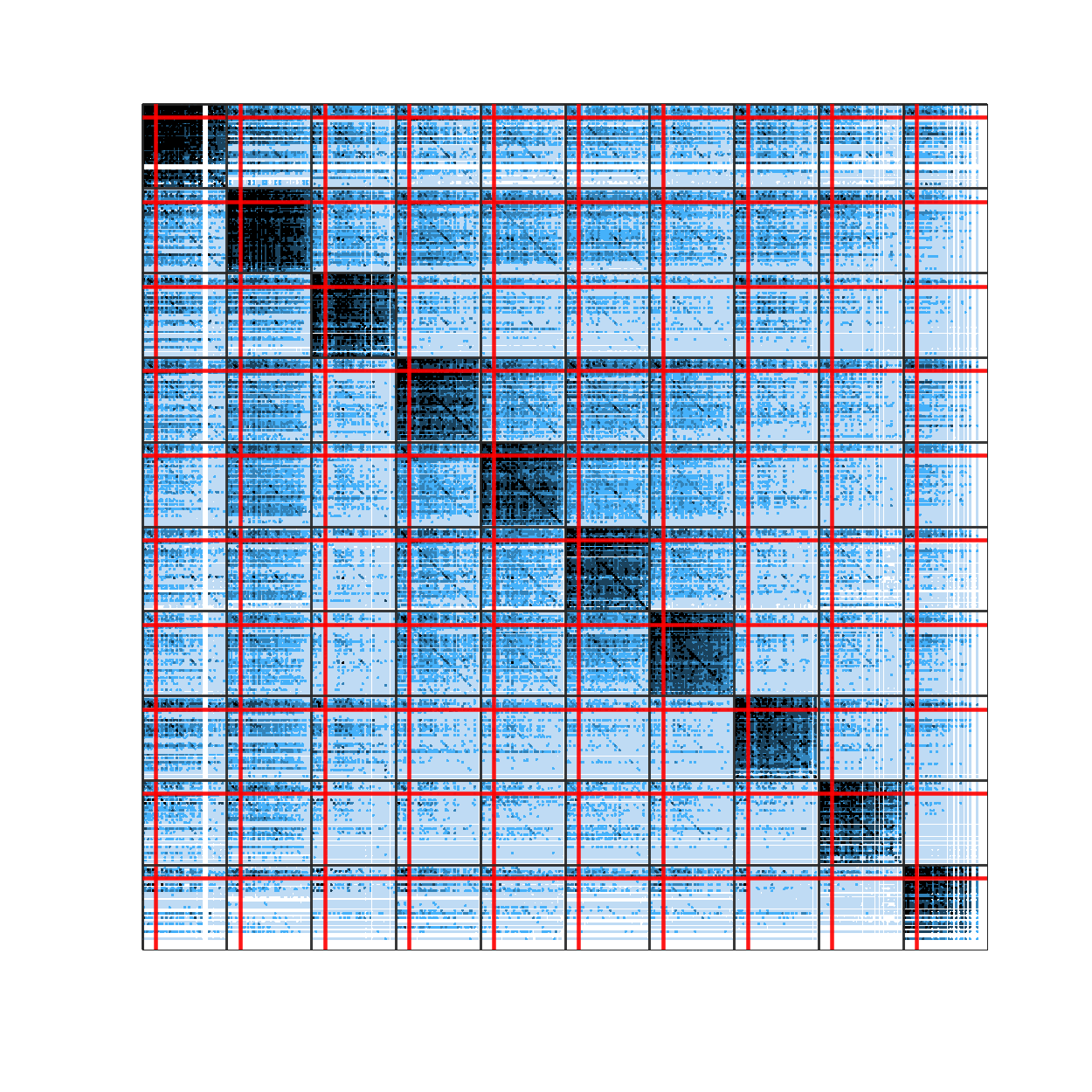}
		}
	\end{center}
	\caption{Supra-adjacency matrix plots for the WIOD world trade multilayer network of the year $2014$ for the parameters $p=q=22$.
		Dark fonts indicate large edge weights.
		The red lines in panel b) indicate the layer core size $s^\ast_{\mathrm{layer}}=7$ and the red lines in panel c) indicate the node core size $s^\ast_{\mathrm{node}}=9$ in each block.}\label{fig:WIOD_spy}
\end{figure*}

Figures \ref{fig:WIOD_cores_country} and \ref{fig:WIOD_cores_industry} show the evolution of the core layers and nodes for the WIOD world trade multilayer network over the years $2000$ to $2014$.
In this multilayer network, industries are modeled as nodes while countries are modeled as layers.
The results were obtained by the nonlinear spectral method for multilayer networks in the provably convergent setting with parameters $p=q=22$.

In addition, \Cref{fig:WIOD_spy} shows spy plots of the original supra-adjacency matrix as well as its reordered version in which blocks are permuted according to the layer coreness vector $\bm{c}$ while each block is internally permuted according to the node coreness vector $\bm{x}$.
Since the network is almost fully connected, sparse L-shapes can not be obtained.
Instead, edge weights are color-coded to assign darker fonts to larger edge weights.
Our method manages to find node and layer coreness vectors that permute blocks with larger edge weights into the outer L-shape of layers while at the same time permuting larger edge weights into the inner L-shape of nodes within each block.
Note that the diagonal blocks of the supra-adjacency matrix are dominant with respect to edge weights since intra-country trade volumes tend to be higher than inter-country ones.

\section{Efficient evaluation of the QUBO objective function}\label{sec:QUBO}

Equations \eqref{eq:QUBO_node} and \eqref{eq:QUBO_layer} define the QUBO objective functions
\begin{align}
\sum_{k,l=1}^L \max \{ \bm{c}_k, \bm{c}_l\} \sum_{{i,j=1}}^n \left( \frac{\mathcal{A}_{ij}^{kl}}{n_1^{(kl)}} \max \{ \bar{\bm{x}}_i, \bar{\bm{x}}_j\} + \frac{1 - \mathcal{A}_{ij}^{kl}}{n_2^{(kl)}} (1-\max \{ \bar{\bm{x}}_i, \bar{\bm{x}}_j\}) \right),\label{eq:QUBO_objective_node}\\
\sum_{i,j=1}^n \max \{ \bm{x}_i, \bm{x}_j\} \sum_{{k,l=1}}^L \left( \frac{\mathcal{A}_{ij}^{kl}}{n_1^{(ij)}} \max \{ \bar{\bm{c}}_k, \bar{\bm{c}}_l\} + \frac{1 - \mathcal{A}_{ij}^{kl}}{n_2^{(ij)}} (1-\max \{ \bar{\bm{c}}_k, \bar{\bm{c}}_l\}) \right).\label{eq:QUBO_objective_layer}
\end{align}
for the determination of node and layer core sizes $s^\ast_{\mathrm{node}}$ and $s^\ast_{\mathrm{layer}}$, respectively, by evaluating the objective functions for $n$ binary vectors $\bar{\bm{x}}\in\{0,1\}^n$ and $L$ binary vectors $\bar{\bm{c}}\in\{0,1\}^L$, respectively, as described in \Cref{sec:core_size}.
The term QUBO stands for ``quadratic unconstrained binary optimization'' problem as which the single-layer QUBO objective function proposed in Ref.~\cite{higham2022core} and stated in (5) in the main text was first introduced in the literature.
A naive implementation consists of $5$ nested for-loops, which becomes computationally burdensome for medium- to large-scale networks.
For example, this procedure requires several days of runtime for the OpenAlex multilayer citation network of complex network scientists for implementations in both julia and Matlab.

\begin{figure}
	\begin{algorithm}[H]%[htb!]
		%	\vspace{0.5em}
		\raggedright{
			\begin{tabular}{lll}
				%		\vspace{1mm}
				\textbf{Input}:
				& $\mathcal{A}\in\R^{n \times n \times L \times L}_{\geq 0},$ & Adjacency tensor.\\%\vspace{1mm}
				& $\bm{x}\in\R^n_{>0},$ & Optimized node coreness vector.\\
				& $\bm{c}\in\R^L_{>0},$ & Optimized layer coreness vector.\\
			\end{tabular}
		}
		\vspace{2mm}
		
		\begin{algorithmic}[1]
			\State Set $c_1=c_2=0$
			\noindent\For{$k=1:L$} \hfill\% \textbf{Preparation}
			\noindent\State Compute $\bm{e}_k^T\bm{C}$ \hfill\% Row vector of length $L$
			\noindent\State Compute $\bm{d}^{(kl)} = \mathcal{A}^{(kl)}\bm{1}_n + (\mathcal{A}^{(kl)})^T\bm{1}_n$ for $l=1,\dots,L$ \hfill\% $L$ column vectors of length $n$
			\noindent\State Compute $n_1^{(kl)}$ and $n_2^{(kl)}$ for $l=1,\dots,L$ using $\bm{d}^{(kl)}$ \hfill\% $2L$ scalars
			\noindent\State Compute $\bm{e}_k^T\bm{E}$ using $\bm{d}^{(kl)},n_1^{(kl)},$ and $n_2^{(kl)}$ \hfill\% Row vector of length $L$
			\noindent\State Compute the $k$th block row $\bm{D}_k$ of $\bm{D}$ using $\bm{d}^{(kl)},n_1^{(kl)},$ and $n_2^{(kl)}$ \hfill\% $L\times nL$ block matrix with $L$ diagonal blocks
			\noindent\State Compute the $k$th block row $\bm{A}_k$ of $\bm{A}$ using $\bm{d}^{(kl)},n_1^{(kl)},$ and $n_2^{(kl)}$ \hfill\% $L\times nL$ block matrix (typically sparse)
			\State Update $c_1 = c_1 + 2(n-1/2)\bm{e}_k^T\bm{E}\bm{1}_L$
			\State Update $c_2 = c_2 + \bm{e}_k^T\bm{E}\bm{1}_L$
			\State Clear $\bm{e}_k^T\bm{C}, \bm{e}_k^T\bm{E}, \bm{d}^{(kl)}, n_1^{(kl)}$ and $n_2^{(kl)}$ from memory
			\EndFor
			\noindent\State Set $\bm{q}=\bm{0}\in\R^n$ \hfill\% Empty vector of node QUBO scores
			\noindent\State Set $\bar{\bm{x}}=\bm{0}\in\R^n$ \hfill\% Empty binary vector
			\noindent\For{$i=1:n$} \hfill\% \textbf{Evaluation}
			\noindent\State Set $i^\ast$ as the index of the $i$th largest entry of $\bm{x}$ \hfill\% Index for the next additional non-zero entry in $\bar{\bm{x}}$
			\noindent\State $\bar{\bm{x}} = \bar{\bm{x}} + \bm{e}_{i^\ast}$ \hfill\% Update binary vector $\bar{\bm{x}}$
			\State $\bm{q}_i = \frac{1}{\sum_{k,l=1}^L \max \{ \bm{c}_k, \bm{c}_l\}} \left[ (\bm{1}_L\otimes \bar{\bm{x}})^T \bm{D} (\bm{1}_L\otimes \bar{\bm{x}}) - c_1 (\bar{\bm{x}}^T\bar{\bm{x}}) - (\bm{1}_L\otimes \bar{\bm{x}})^T \bm{A} (\bm{1}_L\otimes \bar{\bm{x}}) + c_2 (\bar{\bm{x}}^T\bar{\bm{x}})^2 \right]$
			\EndFor
			\State $s^\ast_{\mathrm{node}} = \text{argmax}_{i=1,\dots,n}~\bm{q}_i$
		\end{algorithmic}
		\vspace{2mm}
		\begin{tabular}{lll}
			%		\vspace{1mm}
			\textbf{Output}: & $s^\ast_{\mathrm{node}},$ & Node core size.\\
			& $\bm{q}_{s^\ast_{\mathrm{node}}},$ & Node QUBO score corresponding to $s^\ast_{\mathrm{node}}$.
		\end{tabular}
		\caption{Efficient QUBO objective function evaluation.}\label{alg2}
	\end{algorithm}
\end{figure}

In this section, we describe an efficient procedure for evaluating \eqref{eq:QUBO_objective_node} and \eqref{eq:QUBO_objective_layer} and summarize it in \Cref{alg2}.
For the OpenAlex network, julia or Matlab implementations of \Cref{alg2} reduces the runtime from several days to approximately $20$ minutes on a standard desktop computer with an AMD Ryzen 5 5600X 6-Core processor and 16GB memory.

Since \eqref{eq:QUBO_objective_layer} is obtained by exchanging the roles of the vectors $\bm{x}$ and $\bm{c}$ and consequently the indices $i,j$ and $k,l$ in \eqref{eq:QUBO_objective_node}, we restrict ourselves to describing our procedure in the setting of \eqref{eq:QUBO_objective_node}.
This case builds on the flattening of the fourth-order adjacency tensor \mbox{$\mathcal{A}\in\R^{n\times n\times L\times L}$} into the supra-adjacency matrix, i.e., the entry $(i,j,k,l)$ in $\mathcal{A}$ corresponds to entry $((k-1)n+i,(l-1)n+j)$ in the supra-adjacency matrix.
Equation \eqref{eq:QUBO_objective_layer} instead requires the flattening that maps the index $(i,j,k,l)$ in $\mathcal{A}$ to the entry $((i-1)L+k,(j-1)L+l)$ in the corresponding matrix representation.
In this matrix representation that has recently been introduced as the dual multilayer network \cite{presigny2024node}, blocks correspond to nodes while edges within blocks correspond to layers.

The starting point of the derivation of \Cref{alg2} is rewriting the sum over the indices $i$ and $j$ in \eqref{eq:QUBO_objective_node} as an inner product.
In fact, this sum exactly corresponds to the single-layer QUBO objective function proposed in Ref.~\cite{higham2022core} and stated in \eqref{eq:QUBO_single_layer}.
Consequently, following the same steps as in \cite[Section 3]{higham2022core} and \cite[Section 5]{bergermann2024nonlinear}, we can rewrite \eqref{eq:QUBO_objective_node} as
\begin{equation}\label{eq:QUBO_bilinear}
\sum_{k,l=1}^L \max \{ \bm{c}_k, \bm{c}_l\}~\bar{\bm{x}}^T \bm{Q}^{(kl)} \bar{\bm{x}},
\end{equation}
where
\begin{equation}\label{eq:QUBO_Qkl}
\bm{Q}^{(kl)} = \left( \frac{1}{n_1^{(kl)}} + \frac{1}{n_2^{(kl)}} \right)\text{diag}\left[ \mathcal{A}^{(kl)}\bm{1}_n + (\mathcal{A}^{(kl)})^T\bm{1}_n \right] - 2 \frac{n-1/2}{n_2^{(kl)}} \bm{I}_n - \left( \frac{1}{n_1^{(kl)}} + \frac{1}{n_2^{(kl)}} \right) \mathcal{A}^{(kl)} + \frac{1}{n_2^{(kl)}}\bm{1}_n\bm{1}_n^T,
\end{equation}
with $\mathcal{A}^{(kl)}\in\R^{n\times n}$ denoting the $(k,l)$-block of the supra-adjacency matrix (or equivalently, the $n\times n$ slice of $\mathcal{A}$ in which the indices $k$ and $l$ are fixed), $\bm{1}_n\in\R^n$ denoting the vector of all ones, and $\bm{I}_n\in\R^{n\times n}$ denoting the identity matrix.
Note that due to the non-symmetry of the adjacency tensor $\mathcal{A}$, the first diagonal term of $\bm{Q}^{(kl)}$ can generally not be simplified into the form stated in  Ref.'s~\cite{higham2022core,bergermann2024nonlinear}.

The key to the fast evaluation of \eqref{eq:QUBO_objective_node} is further rewriting \eqref{eq:QUBO_bilinear} in terms of a block matrix $\bm{Q}\in\R^{nL\times nL}$, i.e., 
\begin{equation}\label{eq:QUBO_one_bilinear_form}
\sum_{k,l=1}^L \max \{ \bm{c}_k, \bm{c}_l\}~\bar{\bm{x}}^T \bm{Q}^{(kl)} \bar{\bm{x}} = (\bm{1}_L\otimes\bar{\bm{x}})^T\bm{Q}(\bm{1}_L\otimes\bar{\bm{x}}),
\end{equation}
where
\begin{equation}\label{eq:QUBO_Q_block_matrix}
\bm{Q} = (\bm{C}\otimes\bm{I}_n) \odot \bm{D} + \bm{E}\otimes \bm{I}_n + (\bm{C}\otimes\bm{I}_n) \odot \bm{A} + \bm{E}\otimes (\bm{1}_n\bm{1}_n^T),
\end{equation}
and in which the $(k,l)$-block of size $n\times n$ represents the summand $k,l$ in \eqref{eq:QUBO_bilinear}.
Here, $\otimes$ denotes the Kronecker (tensor) product while $\odot$ denotes the Hadamard (elementwise) matrix product.
The matrices $\bm{A}\in\R^{nL\times nL}, \bm{C}\in\R^{L\times L}, \bm{D}\in\R^{nL\times nL},$ and $\bm{E}\in\R^{L\times L}$ in \eqref{eq:QUBO_Q_block_matrix} are defined as
\begin{equation*}
\bm{A} = \begin{bmatrix}
\left( \frac{1}{n_1^{(11)}} + \frac{1}{n_2^{(11)}} \right)\mathcal{A}^{(11)} & \cdots & \left( \frac{1}{n_1^{(1L)}} + \frac{1}{n_2^{(1L)}} \right)\mathcal{A}^{(1L)}\\
\vdots&\ddots&\vdots\\
\left( \frac{1}{n_1^{(L1)}} + \frac{1}{n_2^{(L1)}} \right)\mathcal{A}^{(L1)} & \cdots & \left( \frac{1}{n_1^{(LL)}} + \frac{1}{n_2^{(LL)}} \right)\mathcal{A}^{(LL)}
\end{bmatrix}, \qquad
\bm{C} = \begin{bmatrix}
\bm{c}_1 & \max\{\bm{c}_1,\bm{c}_2\} & \cdots & \max\{\bm{c}_1,\bm{c}_L\}\\
\max\{\bm{c}_2,\bm{c}_1\} & \bm{c}_2 & \cdots & \max\{\bm{c}_2,\bm{c}_L\}\\
\vdots&\vdots&\ddots&\vdots\\
\max\{\bm{c}_L,\bm{c}_1\} & \max\{\bm{c}_L,\bm{c}_2\} & \cdots & \bm{c}_L
\end{bmatrix},
\end{equation*}
\begin{equation*}
\bm{D} = \begin{bmatrix}
\left( \frac{1}{n_1^{(11)}} + \frac{1}{n_2^{(11)}} \right)\text{diag}\left[ \mathcal{A}^{(11)}\bm{1}_n + (\mathcal{A}^{(11)})^T\bm{1}_n \right] & \cdots & \left( \frac{1}{n_1^{(1L)}} + \frac{1}{n_2^{(1L)}} \right)\text{diag}\left[ \mathcal{A}^{(1L)}\bm{1}_n + (\mathcal{A}^{(1L)})^T\bm{1}_n \right]\\
\vdots&\ddots&\vdots\\
\left( \frac{1}{n_1^{(L1)}} + \frac{1}{n_2^{(L1)}} \right)\text{diag}\left[ \mathcal{A}^{(L1)}\bm{1}_n + (\mathcal{A}^{(L1)})^T\bm{1}_n \right] & \cdots & \left( \frac{1}{n_1^{(LL)}} + \frac{1}{n_2^{(LL)}} \right)\text{diag}\left[ \mathcal{A}^{(LL)}\bm{1}_n + (\mathcal{A}^{(LL)})^T\bm{1}_n \right]
\end{bmatrix},
\end{equation*}
\begin{equation*}
\bm{E} = \begin{bmatrix}
\frac{\bm{c}_1}{n_2^{(11)}} & \frac{\max\{\bm{c}_1,\bm{c}_2\}}{n_2^{(12)}} & \cdots & \frac{\max\{\bm{c}_1,\bm{c}_L\}}{n_2^{(1L)}}\\
\frac{\max\{\bm{c}_2,\bm{c}_1\}}{n_2^{(21)}} & \frac{\bm{c}_2}{n_2^{(22)}} & \cdots & \frac{\max\{\bm{c}_2,\bm{c}_L\}}{n_2^{(2L)}}\\
\vdots&\vdots&\ddots&\vdots\\
\frac{\max\{\bm{c}_L,\bm{c}_1\}}{n_2^{(L1)}} & \frac{\max\{\bm{c}_L,\bm{c}_2\}}{n_2^{(L2)}} & \cdots & \frac{\bm{c}_L}{n_2^{(LL)}}
\end{bmatrix},
\end{equation*}
where entries including $n_1^{(kl)}=0$ or $n_2^{(kl)}=0$ need to be set to zero.

Note that $\bm{A}$ is a block-weighted version of the supra-adjacency matrix, i.e., is has the same number of non-zero entries as the adjacency tensor $\mathcal{A}$.
Furthermore, the number of non-zero entries in $\bm{D}$ is upper-bounded by the number of non-zero entries in $\mathcal{A}$.
For moderate network sizes, one can now assemble the matrix $\bm{Q}$ defined in \eqref{eq:QUBO_Q_block_matrix} and evaluate \eqref{eq:QUBO_one_bilinear_form} for all binary vectors $\bar{\bm{x}}$.
However, a further increase in efficiency is possible.

For the two terms in \eqref{eq:QUBO_Q_block_matrix} involving $\bm{E}$, straightforward computations show that
\begin{equation}\label{eq:QUBO_I_scalar}
(\bm{1}_L\otimes\bar{\bm{x}})^T(\bm{E}\otimes \bm{I}_n)(\bm{1}_L\otimes\bar{\bm{x}}) = (\bar{\bm{x}}^T\bar{\bm{x}})\bm{1}^T\bm{E1}
\end{equation}
and for binary vectors $\bar{\bm{x}}$,
\begin{equation}\label{eq:QUBO_11_scalar}
(\bm{1}_L\otimes\bar{\bm{x}})^T(\bm{E}\otimes \bm{1}_n\bm{1}_n^T)(\bm{1}_L\otimes\bar{\bm{x}}) = (\bar{\bm{x}}^T\bar{\bm{x}})^2\bm{1}^T\bm{E1},
\end{equation}
where $\bar{\bm{x}}^T\bar{\bm{x}}$ corresponds to the number of non-zero entries in $\bar{\bm{x}}$.
Consequently, the matrices $\bm{E}\otimes \bm{I}_n$ and $\bm{E}\otimes \bm{1}_n\bm{1}_n^T$ never need to be formed explicitly.
In fact, storing the dense matrices $\bm{C},\bm{E}\in\R^{L\times L}$ may be impossible due to memory constraints for sufficiently large numbers of layers.
Also recall that the described procedure can be used to evaluate \eqref{eq:QUBO_objective_layer} by exchanging the roles of $\bm{x}$ and $\bm{c}$ and consequently $n$ and $L$ in the matrix sizes.
For the OpenAlex network, storing a dense $53\,423 \times 53\,423$ matrix is impossible on (currently common) hardware with 16GB of memory.
The same is true for the $L^2$ quantities $n_1^{(kl)}$ and $n_2^{(kl)}$ in the setting of \eqref{eq:QUBO_objective_node} and the $n^2$ quantities $n_1^{(ij)}$ and $n_2^{(ij)}$ in the setting of \eqref{eq:QUBO_objective_layer}.

The preparation step of \Cref{alg2} in lines $2$ to $12$ circumvents this issue by assembling the quantities $\bm{C}, \bm{E}, n_1^{(kl)},$ and $n_2^{(kl)}$ row by row.
The sparse downstream quantities $\bm{A}$ and $\bm{D}$ can then be assembled and stored block row by block row.
Similarly, we define scalar quantities $c_1=2(n-1/2)\bm{1}^T\bm{E1}$ and $c_2=\bm{1}^T\bm{E1}$ for the evaluation of \eqref{eq:QUBO_I_scalar} and \eqref{eq:QUBO_11_scalar}.
Their values are updated with the currently available information on the $k$th row of $\bm{E}$.

The actual evaluation of the QUBO objective function values via \eqref{eq:QUBO_one_bilinear_form} can then cheaply be performed in lines $15$ to $19$ of \Cref{alg2}.
The runtime reduction from several days to approximately $20$ minutes on the OpenAlex network is achieved by replacing $5$ nested for-loops by two single for-loops performing efficient vector- and matrix-valued operations.

\section{Proof of Theorem 1}\label{sec:proof}

In this section, we derive the coefficient matrix $\bm{M}$ defined in \eqref{eq:coefficient_matrix}.
Furthermore, we give detailed references to previous results from Ref.~\cite{bergermann2024nonlinear} that also hold true in the case of the nonlinear spectral method for multilayer networks and complete the proof of Theorem \ref{thm:unique_solution}.
The starting point is the following Lemma.
\begin{lemma}\label{lemma:coefficient_matrix}
	The objective function $f_{\alpha,\beta} (\bm{x},\bm{c}) = \sum_{i,j=1}^n\sum_{k,l=1}^L \mathcal{A}_{ij}^{kl} (\bm{x}_i^\alpha + \bm{x}_j^\alpha)^{1/\alpha} (\bm{c}_k^\beta + \bm{c}_l^\beta)^{1/\beta}$ satisfies the elementwise inequalities
	\begin{equation}\label{eq:theta_multilayer_setting}
	\begin{array}{cc}
	|\nabla_{\bm{x}} \nabla_{\bm{x}} f_{\alpha,\beta} (\bm{x},\bm{c})\bm{x}| \leq \bm{\Theta}_{11} |\nabla_{\bm{x}} f_{\alpha,\beta}(\bm{x},\bm{c})|
	&
	|\nabla_{\bm{c}} \nabla_{\bm{x}} f_{\alpha,\beta} (\bm{x},\bm{c})\bm{c}| \leq \bm{\Theta}_{12} |\nabla_{\bm{x}} f_{\alpha,\beta}(\bm{x},\bm{c})|
	\\
	|\nabla_{\bm{x}} \nabla_{\bm{c}} f_{\alpha,\beta} (\bm{x},\bm{c})\bm{x}| \leq \bm{\Theta}_{21} |\nabla_{\bm{c}} f_{\alpha,\beta}(\bm{x},\bm{c})|
	&
	|\nabla_{\bm{c}} \nabla_{\bm{c}} f_{\alpha,\beta} (\bm{x},\bm{c})\bm{c}| \leq \bm{\Theta}_{22} |\nabla_{\bm{c}} f_{\alpha,\beta}(\bm{x},\bm{c})|
	\end{array}
	\end{equation}
	with the coefficient matrix
	\begin{equation*}
	\bm{\Theta} = \begin{bmatrix}
	\bm{\Theta}_{11} & \bm{\Theta}_{12}\\
	\bm{\Theta}_{21} & \bm{\Theta}_{22}
	\end{bmatrix}
	=
	\begin{bmatrix}
	|\alpha-1| & 2\\
	2 & |\beta-1|
	\end{bmatrix}.
	\end{equation*}
\end{lemma}
\begin{proof}
	We start by computing entries of the gradients of $f_{\alpha,\beta}(\bm{x},\bm{c})$ with respect to the variables $\bm{x}$ and $\bm{c}$,
	\begin{align}
	[\nabla_{\bm{x}} f_{\alpha,\beta}(\bm{x},\bm{c})]_m & = \bm{x}_m^{\alpha-1} \sum_j (\bm{x}_m^\alpha + \bm{x}_j^\alpha)^{1/\alpha-1} \sum_{k,l} (\mathcal{A}_{mj}^{kl} + \mathcal{A}_{jm}^{kl}) (\bm{c}_k^\beta + \bm{c}_l^\beta)^{1/\beta},\label{eq:grad_x}\\
	[\nabla_{\bm{c}} f_{\alpha,\beta}(\bm{x},\bm{c})]_p & = \bm{c}_p^{\beta-1} \sum_l (\bm{c}_p^\beta + \bm{c}_l^\beta)^{1/\beta-1} \sum_{i,j} (\mathcal{A}_{ij}^{pl} + \mathcal{A}_{ij}^{lp}) (\bm{x}_i^\alpha + \bm{x}_j^\alpha)^{1/\alpha}.\label{eq:grad_c}
	\end{align}
	The sum $\sum_{k,l} (\mathcal{A}_{mj}^{kl} + \mathcal{A}_{jm}^{kl}) (\bm{c}_k^\beta + \bm{c}_l^\beta)^{1/\beta}$ in \eqref{eq:grad_x} does not depend on $\bm{x}$.
	Hence, the coefficient $\bm{\Theta}_{11}=|\alpha-1|$ is obtained by computations analogous to those in the proof of \cite[Lemma B.1]{bergermann2024nonlinear} where we have $(\mathcal{A}_{mj}^{kl} + \mathcal{A}_{jm}^{kl})$ instead of $2\mathcal{A}_{mj}^{kl}$ due to the non-symmetry of $\mathcal{A}$.
	Furthermore, since \eqref{eq:grad_c} is obtained by exchanging the roles of $\bm{x}$ and $\alpha$ with those of $\bm{c}$ and $\beta$, analogous arguments yield $\bm{\Theta}_{22}=|\beta-1|$.
	
	To see $\bm{\Theta}_{12}=2$, we consider
	\begin{align*}
	[\nabla_{\bm{c}} \nabla_{\bm{x}} f_{\alpha,\beta}(\bm{x},\bm{c}) \bm{c}]_m
	& = 2 \bm{x}_m^{\alpha-1} \sum_j (\bm{x}_m^\alpha + \bm{x}_j^\alpha)^{1/\alpha-1} \sum_{p,l} \bm{c}_p^\beta (\mathcal{A}_{mj}^{kl} + \mathcal{A}_{jm}^{kl}) (\bm{c}_p^\beta + \bm{c}_l^\beta)^{1/\beta-1}\\
	& = 2 \bm{x}_m^{\alpha-1} \sum_j (\bm{x}_m^\alpha + \bm{x}_j^\alpha)^{1/\alpha-1} \sum_{p,l} \frac{\bm{c}_p^\beta}{\bm{c}_p^\beta + \bm{c}_l^\beta} (\mathcal{A}_{mj}^{kl} + \mathcal{A}_{jm}^{kl}) (\bm{c}_p^\beta + \bm{c}_l^\beta)^{1/\beta}\\
	& \leq 2 \bm{x}_m^{\alpha-1} \sum_j (\bm{x}_m^\alpha + \bm{x}_j^\alpha)^{1/\alpha-1} \sum_{p,l} (\mathcal{A}_{mj}^{kl} + \mathcal{A}_{jm}^{kl}) (\bm{c}_p^\beta + \bm{c}_l^\beta)^{1/\beta}\\
	& = 2 [\nabla_{\bm{x}} f_{\alpha,\beta}(\bm{x},\bm{c})]_m,
	\end{align*}
	since $\bm{c}$ is entry-wise positive.
	Again, $\bm{\Theta}_{21}=2$ is obtained by analogous arguments after exchanging the roles of $\bm{x}$ and $\alpha$ with those of $\bm{c}$ and $\beta$.
\end{proof}

Next, we define the vector-valued fixed point map underlying the nonlinear spectral method for multilayer networks following from \eqref{eq:fixed_point_equations} as
\begin{equation}\label{eq:fixed_point_map}
G^{\alpha,\beta}(\bm{x},\bm{c}) = \begin{bmatrix}
J_{p^\ast}(\nabla_{\bm{x}} f_{\alpha,\beta}(\bm{x},\bm{c}))\\
J_{q^\ast}(\nabla_{\bm{c}} f_{\alpha,\beta}(\bm{x},\bm{c}))
\end{bmatrix}
=
\begin{bmatrix}
G^{\alpha,\beta}_1(\bm{x},\bm{c})\\
G^{\alpha,\beta}_2(\bm{x},\bm{c})
\end{bmatrix},
\end{equation}
where $J_p(\bm{x}) := \nabla \|\bm{x}\|_p = \|\bm{x}\|_p^{1-p}\bm{x}^{p-1}$ for $\bm{x}$ entry-wise positive.
With this notation, we state the following result establishing the coefficient matrix $\bm{M}$, cf.~\eqref{eq:coefficient_matrix}.

\begin{lemma}
	Let $D_{\bm{x}} G^{\alpha,\beta}_i(\bm{x},\bm{c})$ and $D_{\bm{c}} G^{\alpha,\beta}_i(\bm{x},\bm{c})$ for $i=1,2$ denote the Jacobians w.r.t.\ $\bm{x}$ and $\bm{c}$, respectively.
	For elementwise divisions, the map $G^{\alpha,\beta}$ satisfies the element\-wise inequality
	\begin{equation}\label{eq:coeff_matrix}
	\begin{bmatrix}
	\left\|\frac{D_{\bm{x}}G^{\alpha,\beta}_1(\bm{x},\bm{c})\bm{x}}{G^{\alpha,\beta}_1(\bm{x},\bm{c})}\right\|_1 & \left\|\frac{D_{\bm{c}}G^{\alpha,\beta}_1(\bm{x},\bm{c})\bm{c}}{G^{\alpha,\beta}_1(\bm{x},\bm{c})}\right\|_1\\[.7em]
	\left\|\frac{D_{\bm{x}}G^{\alpha,\beta}_2(\bm{x},\bm{c})\bm{x}}{G^{\alpha,\beta}_2(\bm{x},\bm{c})}\right\|_1 & \left\|\frac{D_{\bm{c}}G^{\alpha,\beta}_2(\bm{x},\bm{c})\bm{c}}{G^{\alpha,\beta}_2(\bm{x},\bm{c})}\right\|_1
	\end{bmatrix}
	\leq
	\begin{bmatrix}
	\frac{2|\alpha-1|}{p-1} & \frac{2}{p-1}\\[.7em]
	\frac{2}{q-1} & \frac{2|\beta-1|}{q-1}
	\end{bmatrix}
	=: \bm{M}.
	\end{equation}
\end{lemma}
\begin{proof}
	The proof is analogous to that of \cite[Lemma B.2]{bergermann2024nonlinear}.
\end{proof}

Finally, Theorem \ref{thm:unique_solution} is proven analogously to \cite[Theorem 3.1]{bergermann2024nonlinear} for the coefficient matrix $\bm{M}$ defined in \eqref{eq:coeff_matrix}.
The proof relies on showing the contractivity of the fixed point map $G^{\alpha,\beta}$ with respect to a suitably defined Thompson metric with Lipschitz constant $\rho(\bm{M})$, i.e., the spectral radius of $\bm{M}$ determines the linear rate of convergence of Algorithm \ref{alg}.

\twocolumngrid


\begin{thebibliography}{38}%
	\makeatletter
	\providecommand \@ifxundefined [1]{%
		\@ifx{#1\undefined}
	}%
	\providecommand \@ifnum [1]{%
		\ifnum #1\expandafter \@firstoftwo
		\else \expandafter \@secondoftwo
		\fi
	}%
	\providecommand \@ifx [1]{%
		\ifx #1\expandafter \@firstoftwo
		\else \expandafter \@secondoftwo
		\fi
	}%
	\providecommand \natexlab [1]{#1}%
	\providecommand \enquote  [1]{``#1''}%
	\providecommand \bibnamefont  [1]{#1}%
	\providecommand \bibfnamefont [1]{#1}%
	\providecommand \citenamefont [1]{#1}%
	\providecommand \href@noop [0]{\@secondoftwo}%
	\providecommand \href [0]{\begingroup \@sanitize@url \@href}%
	\providecommand \@href[1]{\@@startlink{#1}\@@href}%
	\providecommand \@@href[1]{\endgroup#1\@@endlink}%
	\providecommand \@sanitize@url [0]{\catcode `\\12\catcode `\$12\catcode
		`\&12\catcode `\#12\catcode `\^12\catcode `\_12\catcode `\%12\relax}%
	\providecommand \@@startlink[1]{}%
	\providecommand \@@endlink[0]{}%
	\providecommand \url  [0]{\begingroup\@sanitize@url \@url }%
	\providecommand \@url [1]{\endgroup\@href {#1}{\urlprefix }}%
	\providecommand \urlprefix  [0]{URL }%
	\providecommand \Eprint [0]{\href }%
	\providecommand \doibase [0]{https://doi.org/}%
	\providecommand \selectlanguage [0]{\@gobble}%
	\providecommand \bibinfo  [0]{\@secondoftwo}%
	\providecommand \bibfield  [0]{\@secondoftwo}%
	\providecommand \translation [1]{[#1]}%
	\providecommand \BibitemOpen [0]{}%
	\providecommand \bibitemStop [0]{}%
	\providecommand \bibitemNoStop [0]{.\EOS\space}%
	\providecommand \EOS [0]{\spacefactor3000\relax}%
	\providecommand \BibitemShut  [1]{\csname bibitem#1\endcsname}%
	\let\auto@bib@innerbib\@empty
	%</preamble>
	\bibitem [{\citenamefont {Watts}\ and\ \citenamefont
		{Strogatz}(1998)}]{watts1998collective}%
	\BibitemOpen
	\bibfield  {author} {\bibinfo {author} {\bibfnamefont {D.~J.}\ \bibnamefont
			{Watts}}\ and\ \bibinfo {author} {\bibfnamefont {S.~H.}\ \bibnamefont
			{Strogatz}},\ }\bibfield  {title} {\bibinfo {title} {Collective dynamics of
			‘small-world’networks},\ }\href {https://doi.org/10.1038/30918}
	{\bibfield  {journal} {\bibinfo  {journal} {nature}\ }\textbf {\bibinfo
			{volume} {393}},\ \bibinfo {pages} {440} (\bibinfo {year}
		{1998})}\BibitemShut {NoStop}%
	\bibitem [{\citenamefont {Barab{\'a}si}\ and\ \citenamefont
		{Albert}(1999)}]{barabasi1999emergence}%
	\BibitemOpen
	\bibfield  {author} {\bibinfo {author} {\bibfnamefont {A.-L.}\ \bibnamefont
			{Barab{\'a}si}}\ and\ \bibinfo {author} {\bibfnamefont {R.}~\bibnamefont
			{Albert}},\ }\bibfield  {title} {\bibinfo {title} {Emergence of scaling in
			random networks},\ }\href {https://doi.org/10.1126/science.286.5439.509}
	{\bibfield  {journal} {\bibinfo  {journal} {science}\ }\textbf {\bibinfo
			{volume} {286}},\ \bibinfo {pages} {509} (\bibinfo {year}
		{1999})}\BibitemShut {NoStop}%
	\bibitem [{\citenamefont {Newman}(2003)}]{newman2003structure}%
	\BibitemOpen
	\bibfield  {author} {\bibinfo {author} {\bibfnamefont {M.~E.}\ \bibnamefont
			{Newman}},\ }\bibfield  {title} {\bibinfo {title} {The structure and function
			of complex networks},\ }\href {https://doi.org/10.1137/S003614450342480}
	{\bibfield  {journal} {\bibinfo  {journal} {SIAM review}\ }\textbf {\bibinfo
			{volume} {45}},\ \bibinfo {pages} {167} (\bibinfo {year} {2003})}\BibitemShut
	{NoStop}%
	\bibitem [{\citenamefont {Boccaletti}\ \emph {et~al.}(2006)\citenamefont
		{Boccaletti}, \citenamefont {Latora}, \citenamefont {Moreno}, \citenamefont
		{Chavez},\ and\ \citenamefont {Hwang}}]{boccaletti2006complex}%
	\BibitemOpen
	\bibfield  {author} {\bibinfo {author} {\bibfnamefont {S.}~\bibnamefont
			{Boccaletti}}, \bibinfo {author} {\bibfnamefont {V.}~\bibnamefont {Latora}},
		\bibinfo {author} {\bibfnamefont {Y.}~\bibnamefont {Moreno}}, \bibinfo
		{author} {\bibfnamefont {M.}~\bibnamefont {Chavez}},\ and\ \bibinfo {author}
		{\bibfnamefont {D.-U.}\ \bibnamefont {Hwang}},\ }\bibfield  {title} {\bibinfo
		{title} {Complex networks: Structure and dynamics},\ }\href
	{https://doi.org/10.1016/j.physrep.2005.10.009} {\bibfield  {journal}
		{\bibinfo  {journal} {Physics reports}\ }\textbf {\bibinfo {volume} {424}},\
		\bibinfo {pages} {175} (\bibinfo {year} {2006})}\BibitemShut {NoStop}%
	\bibitem [{\citenamefont {Kivel{\"a}}\ \emph {et~al.}(2014)\citenamefont
		{Kivel{\"a}}, \citenamefont {Arenas}, \citenamefont {Barthelemy},
		\citenamefont {Gleeson}, \citenamefont {Moreno},\ and\ \citenamefont
		{Porter}}]{kivela2014multilayer}%
	\BibitemOpen
	\bibfield  {author} {\bibinfo {author} {\bibfnamefont {M.}~\bibnamefont
			{Kivel{\"a}}}, \bibinfo {author} {\bibfnamefont {A.}~\bibnamefont {Arenas}},
		\bibinfo {author} {\bibfnamefont {M.}~\bibnamefont {Barthelemy}}, \bibinfo
		{author} {\bibfnamefont {J.~P.}\ \bibnamefont {Gleeson}}, \bibinfo {author}
		{\bibfnamefont {Y.}~\bibnamefont {Moreno}},\ and\ \bibinfo {author}
		{\bibfnamefont {M.~A.}\ \bibnamefont {Porter}},\ }\bibfield  {title}
	{\bibinfo {title} {Multilayer networks},\ }\href
	{https://doi.org/10.1093/comnet/cnu016} {\bibfield  {journal} {\bibinfo
			{journal} {Journal of complex networks}\ }\textbf {\bibinfo {volume} {2}},\
		\bibinfo {pages} {203} (\bibinfo {year} {2014})}\BibitemShut {NoStop}%
	\bibitem [{\citenamefont {Boccaletti}\ \emph {et~al.}(2014)\citenamefont
		{Boccaletti}, \citenamefont {Bianconi}, \citenamefont {Criado}, \citenamefont
		{Del~Genio}, \citenamefont {G{\'o}mez-Gardenes}, \citenamefont {Romance},
		\citenamefont {Sendina-Nadal}, \citenamefont {Wang},\ and\ \citenamefont
		{Zanin}}]{boccaletti2014structure}%
	\BibitemOpen
	\bibfield  {author} {\bibinfo {author} {\bibfnamefont {S.}~\bibnamefont
			{Boccaletti}}, \bibinfo {author} {\bibfnamefont {G.}~\bibnamefont
			{Bianconi}}, \bibinfo {author} {\bibfnamefont {R.}~\bibnamefont {Criado}},
		\bibinfo {author} {\bibfnamefont {C.~I.}\ \bibnamefont {Del~Genio}}, \bibinfo
		{author} {\bibfnamefont {J.}~\bibnamefont {G{\'o}mez-Gardenes}}, \bibinfo
		{author} {\bibfnamefont {M.}~\bibnamefont {Romance}}, \bibinfo {author}
		{\bibfnamefont {I.}~\bibnamefont {Sendina-Nadal}}, \bibinfo {author}
		{\bibfnamefont {Z.}~\bibnamefont {Wang}},\ and\ \bibinfo {author}
		{\bibfnamefont {M.}~\bibnamefont {Zanin}},\ }\bibfield  {title} {\bibinfo
		{title} {The structure and dynamics of multilayer networks},\ }\href
	{https://doi.org/10.1016/j.physrep.2014.07.001} {\bibfield  {journal}
		{\bibinfo  {journal} {Physics reports}\ }\textbf {\bibinfo {volume} {544}},\
		\bibinfo {pages} {1} (\bibinfo {year} {2014})}\BibitemShut {NoStop}%
	\bibitem [{\citenamefont {Mucha}\ \emph {et~al.}(2010)\citenamefont {Mucha},
		\citenamefont {Richardson}, \citenamefont {Macon}, \citenamefont {Porter},\
		and\ \citenamefont {Onnela}}]{mucha2010community}%
	\BibitemOpen
	\bibfield  {author} {\bibinfo {author} {\bibfnamefont {P.~J.}\ \bibnamefont
			{Mucha}}, \bibinfo {author} {\bibfnamefont {T.}~\bibnamefont {Richardson}},
		\bibinfo {author} {\bibfnamefont {K.}~\bibnamefont {Macon}}, \bibinfo
		{author} {\bibfnamefont {M.~A.}\ \bibnamefont {Porter}},\ and\ \bibinfo
		{author} {\bibfnamefont {J.-P.}\ \bibnamefont {Onnela}},\ }\bibfield  {title}
	{\bibinfo {title} {Community structure in time-dependent, multiscale, and
			multiplex networks},\ }\href {https://doi.org/10.1126/science.1184819}
	{\bibfield  {journal} {\bibinfo  {journal} {science}\ }\textbf {\bibinfo
			{volume} {328}},\ \bibinfo {pages} {876} (\bibinfo {year}
		{2010})}\BibitemShut {NoStop}%
	\bibitem [{\citenamefont {De~Domenico}\ \emph {et~al.}(2013)\citenamefont
		{De~Domenico}, \citenamefont {Sol{\'e}-Ribalta}, \citenamefont {Cozzo},
		\citenamefont {Kivel{\"a}}, \citenamefont {Moreno}, \citenamefont {Porter},
		\citenamefont {G{\'o}mez},\ and\ \citenamefont
		{Arenas}}]{de2013mathematical}%
	\BibitemOpen
	\bibfield  {author} {\bibinfo {author} {\bibfnamefont {M.}~\bibnamefont
			{De~Domenico}}, \bibinfo {author} {\bibfnamefont {A.}~\bibnamefont
			{Sol{\'e}-Ribalta}}, \bibinfo {author} {\bibfnamefont {E.}~\bibnamefont
			{Cozzo}}, \bibinfo {author} {\bibfnamefont {M.}~\bibnamefont {Kivel{\"a}}},
		\bibinfo {author} {\bibfnamefont {Y.}~\bibnamefont {Moreno}}, \bibinfo
		{author} {\bibfnamefont {M.~A.}\ \bibnamefont {Porter}}, \bibinfo {author}
		{\bibfnamefont {S.}~\bibnamefont {G{\'o}mez}},\ and\ \bibinfo {author}
		{\bibfnamefont {A.}~\bibnamefont {Arenas}},\ }\bibfield  {title} {\bibinfo
		{title} {Mathematical formulation of multilayer networks},\ }\href
	{https://doi.org/10.1103/PhysRevX.3.041022} {\bibfield  {journal} {\bibinfo
			{journal} {Physical Review X}\ }\textbf {\bibinfo {volume} {3}},\ \bibinfo
		{pages} {041022} (\bibinfo {year} {2013})}\BibitemShut {NoStop}%
	\bibitem [{\citenamefont {Gomez}\ \emph {et~al.}(2013)\citenamefont {Gomez},
		\citenamefont {Diaz-Guilera}, \citenamefont {Gomez-Gardenes}, \citenamefont
		{Perez-Vicente}, \citenamefont {Moreno},\ and\ \citenamefont
		{Arenas}}]{gomez2013diffusion}%
	\BibitemOpen
	\bibfield  {author} {\bibinfo {author} {\bibfnamefont {S.}~\bibnamefont
			{Gomez}}, \bibinfo {author} {\bibfnamefont {A.}~\bibnamefont {Diaz-Guilera}},
		\bibinfo {author} {\bibfnamefont {J.}~\bibnamefont {Gomez-Gardenes}},
		\bibinfo {author} {\bibfnamefont {C.~J.}\ \bibnamefont {Perez-Vicente}},
		\bibinfo {author} {\bibfnamefont {Y.}~\bibnamefont {Moreno}},\ and\ \bibinfo
		{author} {\bibfnamefont {A.}~\bibnamefont {Arenas}},\ }\bibfield  {title}
	{\bibinfo {title} {Diffusion dynamics on multiplex networks},\ }\href
	{https://doi.org/10.1103/PhysRevLett.110.028701} {\bibfield  {journal}
		{\bibinfo  {journal} {Physical review letters}\ }\textbf {\bibinfo {volume}
			{110}},\ \bibinfo {pages} {028701} (\bibinfo {year} {2013})}\BibitemShut
	{NoStop}%
	\bibitem [{\citenamefont {Granell}\ \emph {et~al.}(2013)\citenamefont
		{Granell}, \citenamefont {G{\'o}mez},\ and\ \citenamefont
		{Arenas}}]{granell2013dynamical}%
	\BibitemOpen
	\bibfield  {author} {\bibinfo {author} {\bibfnamefont {C.}~\bibnamefont
			{Granell}}, \bibinfo {author} {\bibfnamefont {S.}~\bibnamefont {G{\'o}mez}},\
		and\ \bibinfo {author} {\bibfnamefont {A.}~\bibnamefont {Arenas}},\
	}\bibfield  {title} {\bibinfo {title} {Dynamical interplay between awareness
			and epidemic spreading in multiplex networks},\ }\href
	{https://doi.org/10.1103/PhysRevLett.111.128701} {\bibfield  {journal}
		{\bibinfo  {journal} {Physical review letters}\ }\textbf {\bibinfo {volume}
			{111}},\ \bibinfo {pages} {128701} (\bibinfo {year} {2013})}\BibitemShut
	{NoStop}%
	\bibitem [{\citenamefont {Battiston}\ \emph {et~al.}(2014)\citenamefont
		{Battiston}, \citenamefont {Nicosia},\ and\ \citenamefont
		{Latora}}]{battiston2014structural}%
	\BibitemOpen
	\bibfield  {author} {\bibinfo {author} {\bibfnamefont {F.}~\bibnamefont
			{Battiston}}, \bibinfo {author} {\bibfnamefont {V.}~\bibnamefont {Nicosia}},\
		and\ \bibinfo {author} {\bibfnamefont {V.}~\bibnamefont {Latora}},\
	}\bibfield  {title} {\bibinfo {title} {Structural measures for multiplex
			networks},\ }\href {https://doi.org/10.1103/PhysRevE.89.032804} {\bibfield
		{journal} {\bibinfo  {journal} {Physical Review E}\ }\textbf {\bibinfo
			{volume} {89}},\ \bibinfo {pages} {032804} (\bibinfo {year}
		{2014})}\BibitemShut {NoStop}%
	\bibitem [{\citenamefont {Borgatti}\ and\ \citenamefont
		{Everett}(2000)}]{borgatti2000models}%
	\BibitemOpen
	\bibfield  {author} {\bibinfo {author} {\bibfnamefont {S.~P.}\ \bibnamefont
			{Borgatti}}\ and\ \bibinfo {author} {\bibfnamefont {M.~G.}\ \bibnamefont
			{Everett}},\ }\bibfield  {title} {\bibinfo {title} {Models of core/periphery
			structures},\ }\href {https://doi.org/10.1016/S0378-8733(99)00019-2}
	{\bibfield  {journal} {\bibinfo  {journal} {Social networks}\ }\textbf
		{\bibinfo {volume} {21}},\ \bibinfo {pages} {375} (\bibinfo {year}
		{2000})}\BibitemShut {NoStop}%
	\bibitem [{\citenamefont {Colizza}\ \emph {et~al.}(2006)\citenamefont
		{Colizza}, \citenamefont {Flammini}, \citenamefont {Serrano},\ and\
		\citenamefont {Vespignani}}]{colizza2006detecting}%
	\BibitemOpen
	\bibfield  {author} {\bibinfo {author} {\bibfnamefont {V.}~\bibnamefont
			{Colizza}}, \bibinfo {author} {\bibfnamefont {A.}~\bibnamefont {Flammini}},
		\bibinfo {author} {\bibfnamefont {M.~A.}\ \bibnamefont {Serrano}},\ and\
		\bibinfo {author} {\bibfnamefont {A.}~\bibnamefont {Vespignani}},\ }\bibfield
	{title} {\bibinfo {title} {Detecting rich-club ordering in complex
			networks},\ }\href {https://doi.org/10.1038/nphys209} {\bibfield  {journal}
		{\bibinfo  {journal} {Nature physics}\ }\textbf {\bibinfo {volume} {2}},\
		\bibinfo {pages} {110} (\bibinfo {year} {2006})}\BibitemShut {NoStop}%
	\bibitem [{\citenamefont {Csermely}\ \emph {et~al.}(2013)\citenamefont
		{Csermely}, \citenamefont {London}, \citenamefont {Wu},\ and\ \citenamefont
		{Uzzi}}]{csermely2013structure}%
	\BibitemOpen
	\bibfield  {author} {\bibinfo {author} {\bibfnamefont {P.}~\bibnamefont
			{Csermely}}, \bibinfo {author} {\bibfnamefont {A.}~\bibnamefont {London}},
		\bibinfo {author} {\bibfnamefont {L.-Y.}\ \bibnamefont {Wu}},\ and\ \bibinfo
		{author} {\bibfnamefont {B.}~\bibnamefont {Uzzi}},\ }\bibfield  {title}
	{\bibinfo {title} {Structure and dynamics of core/periphery networks},\
	}\href {https://doi.org/10.1093/comnet/cnt016} {\bibfield  {journal}
		{\bibinfo  {journal} {Journal of Complex Networks}\ }\textbf {\bibinfo
			{volume} {1}},\ \bibinfo {pages} {93} (\bibinfo {year} {2013})}\BibitemShut
	{NoStop}%
	\bibitem [{\citenamefont {Rombach}\ \emph {et~al.}(2014)\citenamefont
		{Rombach}, \citenamefont {Porter}, \citenamefont {Fowler},\ and\
		\citenamefont {Mucha}}]{rombach2014core}%
	\BibitemOpen
	\bibfield  {author} {\bibinfo {author} {\bibfnamefont {M.~P.}\ \bibnamefont
			{Rombach}}, \bibinfo {author} {\bibfnamefont {M.~A.}\ \bibnamefont {Porter}},
		\bibinfo {author} {\bibfnamefont {J.~H.}\ \bibnamefont {Fowler}},\ and\
		\bibinfo {author} {\bibfnamefont {P.~J.}\ \bibnamefont {Mucha}},\ }\bibfield
	{title} {\bibinfo {title} {Core-periphery structure in networks},\ }\href
	{https://doi.org/10.1137/120881683} {\bibfield  {journal} {\bibinfo
			{journal} {SIAM Journal on Applied mathematics}\ }\textbf {\bibinfo {volume}
			{74}},\ \bibinfo {pages} {167} (\bibinfo {year} {2014})}\BibitemShut
	{NoStop}%
	\bibitem [{\citenamefont {Zhang}\ \emph {et~al.}(2015)\citenamefont {Zhang},
		\citenamefont {Martin},\ and\ \citenamefont
		{Newman}}]{zhang2015identification}%
	\BibitemOpen
	\bibfield  {author} {\bibinfo {author} {\bibfnamefont {X.}~\bibnamefont
			{Zhang}}, \bibinfo {author} {\bibfnamefont {T.}~\bibnamefont {Martin}},\ and\
		\bibinfo {author} {\bibfnamefont {M.~E.}\ \bibnamefont {Newman}},\ }\bibfield
	{title} {\bibinfo {title} {Identification of core-periphery structure in
			networks},\ }\href {https://doi.org/10.1103/PhysRevE.91.032803} {\bibfield
		{journal} {\bibinfo  {journal} {Physical Review E}\ }\textbf {\bibinfo
			{volume} {91}},\ \bibinfo {pages} {032803} (\bibinfo {year}
		{2015})}\BibitemShut {NoStop}%
	\bibitem [{\citenamefont {L{\"u}}\ \emph {et~al.}(2016)\citenamefont {L{\"u}},
		\citenamefont {Zhou}, \citenamefont {Zhang},\ and\ \citenamefont
		{Stanley}}]{lu2016h}%
	\BibitemOpen
	\bibfield  {author} {\bibinfo {author} {\bibfnamefont {L.}~\bibnamefont
			{L{\"u}}}, \bibinfo {author} {\bibfnamefont {T.}~\bibnamefont {Zhou}},
		\bibinfo {author} {\bibfnamefont {Q.-M.}\ \bibnamefont {Zhang}},\ and\
		\bibinfo {author} {\bibfnamefont {H.~E.}\ \bibnamefont {Stanley}},\
	}\bibfield  {title} {\bibinfo {title} {The h-index of a network node and its
			relation to degree and coreness},\ }\href
	{https://doi.org/10.1038/ncomms10168} {\bibfield  {journal} {\bibinfo
			{journal} {Nature communications}\ }\textbf {\bibinfo {volume} {7}},\
		\bibinfo {pages} {10168} (\bibinfo {year} {2016})}\BibitemShut {NoStop}%
	\bibitem [{\citenamefont {Battiston}\ \emph {et~al.}(2018)\citenamefont
		{Battiston}, \citenamefont {Guillon}, \citenamefont {Chavez}, \citenamefont
		{Latora},\ and\ \citenamefont {de~Vico~Fallani}}]{battiston2018multiplex}%
	\BibitemOpen
	\bibfield  {author} {\bibinfo {author} {\bibfnamefont {F.}~\bibnamefont
			{Battiston}}, \bibinfo {author} {\bibfnamefont {J.}~\bibnamefont {Guillon}},
		\bibinfo {author} {\bibfnamefont {M.}~\bibnamefont {Chavez}}, \bibinfo
		{author} {\bibfnamefont {V.}~\bibnamefont {Latora}},\ and\ \bibinfo {author}
		{\bibfnamefont {F.}~\bibnamefont {de~Vico~Fallani}},\ }\bibfield  {title}
	{\bibinfo {title} {Multiplex core--periphery organization of the human
			connectome},\ }\href {https://doi.org/10.1098/rsif.2018.0514} {\bibfield
		{journal} {\bibinfo  {journal} {Journal of the Royal Society Interface}\
		}\textbf {\bibinfo {volume} {15}},\ \bibinfo {pages} {20180514} (\bibinfo
		{year} {2018})}\BibitemShut {NoStop}%
	\bibitem [{\citenamefont {Tudisco}\ and\ \citenamefont
		{Higham}(2019{\natexlab{a}})}]{tudisco2019nonlinear}%
	\BibitemOpen
	\bibfield  {author} {\bibinfo {author} {\bibfnamefont {F.}~\bibnamefont
			{Tudisco}}\ and\ \bibinfo {author} {\bibfnamefont {D.~J.}\ \bibnamefont
			{Higham}},\ }\bibfield  {title} {\bibinfo {title} {A nonlinear spectral
			method for core--periphery detection in networks},\ }\href
	{https://doi.org/10.1137/18M118355} {\bibfield  {journal} {\bibinfo
			{journal} {SIAM J. Math. Data Sci.}\ }\textbf {\bibinfo {volume} {1}},\
		\bibinfo {pages} {269} (\bibinfo {year} {2019}{\natexlab{a}})}\BibitemShut
	{NoStop}%
	\bibitem [{\citenamefont {Tudisco}\ and\ \citenamefont
		{Higham}(2019{\natexlab{b}})}]{tudisco2019fast}%
	\BibitemOpen
	\bibfield  {author} {\bibinfo {author} {\bibfnamefont {F.}~\bibnamefont
			{Tudisco}}\ and\ \bibinfo {author} {\bibfnamefont {D.~J.}\ \bibnamefont
			{Higham}},\ }\bibfield  {title} {\bibinfo {title} {A fast and robust kernel
			optimization method for core--periphery detection in directed and weighted
			graphs},\ }\href {https://doi.org/10.1007/s41109-019-0173-9} {\bibfield
		{journal} {\bibinfo  {journal} {Applied Network Science}\ }\textbf {\bibinfo
			{volume} {4}},\ \bibinfo {pages} {1} (\bibinfo {year}
		{2019}{\natexlab{b}})}\BibitemShut {NoStop}%
	\bibitem [{\citenamefont {Higham}\ \emph {et~al.}(2022)\citenamefont {Higham},
		\citenamefont {Higham},\ and\ \citenamefont {Tudisco}}]{higham2022core}%
	\BibitemOpen
	\bibfield  {author} {\bibinfo {author} {\bibfnamefont {C.~F.}\ \bibnamefont
			{Higham}}, \bibinfo {author} {\bibfnamefont {D.~J.}\ \bibnamefont {Higham}},\
		and\ \bibinfo {author} {\bibfnamefont {F.}~\bibnamefont {Tudisco}},\
	}\bibfield  {title} {\bibinfo {title} {Core-periphery partitioning and
			quantum annealing},\ }in\ \href {https://doi.org/10.1145/3534678.3539261}
	{\emph {\bibinfo {booktitle} {Proceedings of the 28th ACM SIGKDD Conference
				on Knowledge Discovery and Data Mining}}}\ (\bibinfo {year} {2022})\ pp.\
	\bibinfo {pages} {565--573}\BibitemShut {NoStop}%
	\bibitem [{\citenamefont {Bergermann}\ \emph {et~al.}(2024)\citenamefont
		{Bergermann}, \citenamefont {Stoll},\ and\ \citenamefont
		{Tudisco}}]{bergermann2024nonlinear}%
	\BibitemOpen
	\bibfield  {author} {\bibinfo {author} {\bibfnamefont {K.}~\bibnamefont
			{Bergermann}}, \bibinfo {author} {\bibfnamefont {M.}~\bibnamefont {Stoll}},\
		and\ \bibinfo {author} {\bibfnamefont {F.}~\bibnamefont {Tudisco}},\
	}\bibfield  {title} {\bibinfo {title} {A nonlinear spectral core-periphery
			detection method for multiplex networks},\ }\href
	{https://doi.org/10.1098/rspa.2023.0914} {\bibfield  {journal} {\bibinfo
			{journal} {Proceedings of the Royal Society A}\ }\textbf {\bibinfo {volume}
			{480}},\ \bibinfo {pages} {20230914} (\bibinfo {year} {2024})}\BibitemShut
	{NoStop}%
	\bibitem [{\citenamefont {Polanco}\ and\ \citenamefont
		{Newman}(2023)}]{polanco2023hierarchical}%
	\BibitemOpen
	\bibfield  {author} {\bibinfo {author} {\bibfnamefont {A.}~\bibnamefont
			{Polanco}}\ and\ \bibinfo {author} {\bibfnamefont {M.}~\bibnamefont
			{Newman}},\ }\bibfield  {title} {\bibinfo {title} {Hierarchical
			core-periphery structure in networks},\ }\href
	{https://doi.org/10.1103/PhysRevE.108.024311} {\bibfield  {journal} {\bibinfo
			{journal} {Physical Review E}\ }\textbf {\bibinfo {volume} {108}},\ \bibinfo
		{pages} {024311} (\bibinfo {year} {2023})}\BibitemShut {NoStop}%
	\bibitem [{\citenamefont {Boyd}\ \emph {et~al.}(2010)\citenamefont {Boyd},
		\citenamefont {Fitzgerald}, \citenamefont {Mahutga},\ and\ \citenamefont
		{Smith}}]{boyd2010computing}%
	\BibitemOpen
	\bibfield  {author} {\bibinfo {author} {\bibfnamefont {J.~P.}\ \bibnamefont
			{Boyd}}, \bibinfo {author} {\bibfnamefont {W.~J.}\ \bibnamefont
			{Fitzgerald}}, \bibinfo {author} {\bibfnamefont {M.~C.}\ \bibnamefont
			{Mahutga}},\ and\ \bibinfo {author} {\bibfnamefont {D.~A.}\ \bibnamefont
			{Smith}},\ }\bibfield  {title} {\bibinfo {title} {Computing continuous
			core/periphery structures for social relations data with minres/svd},\
	}\href@noop {} {\bibfield  {journal} {\bibinfo  {journal} {Social Networks}\
		}\textbf {\bibinfo {volume} {32}},\ \bibinfo {pages} {125} (\bibinfo {year}
		{2010})}\BibitemShut {NoStop}%
	\bibitem [{\citenamefont {Mondrag{\'o}n}(2017)}]{mondragon2017network}%
	\BibitemOpen
	\bibfield  {author} {\bibinfo {author} {\bibfnamefont {R.~J.}\ \bibnamefont
			{Mondrag{\'o}n}},\ }\bibfield  {title} {\bibinfo {title} {Network partition
			via a bound of the spectral radius},\ }\href@noop {} {\bibfield  {journal}
		{\bibinfo  {journal} {Journal of Complex Networks}\ }\textbf {\bibinfo
			{volume} {5}},\ \bibinfo {pages} {513} (\bibinfo {year} {2017})}\BibitemShut
	{NoStop}%
	\bibitem [{\citenamefont {Presigny}\ \emph {et~al.}(2024)\citenamefont
		{Presigny}, \citenamefont {Corsi},\ and\ \citenamefont
		{De~Vico~Fallani}}]{presigny2024node}%
	\BibitemOpen
	\bibfield  {author} {\bibinfo {author} {\bibfnamefont {C.}~\bibnamefont
			{Presigny}}, \bibinfo {author} {\bibfnamefont {M.-C.}\ \bibnamefont
			{Corsi}},\ and\ \bibinfo {author} {\bibfnamefont {F.}~\bibnamefont
			{De~Vico~Fallani}},\ }\bibfield  {title} {\bibinfo {title} {Node-layer
			duality in networked systems},\ }\href
	{https://doi.org/10.1038/s41467-024-50176-5} {\bibfield  {journal} {\bibinfo
			{journal} {Nature Communications}\ }\textbf {\bibinfo {volume} {15}},\
		\bibinfo {pages} {6038} (\bibinfo {year} {2024})}\BibitemShut {NoStop}%
	\bibitem [{\citenamefont {Gautier}\ and\ \citenamefont
		{Tudisco}(2019)}]{gautier2019contractivity}%
	\BibitemOpen
	\bibfield  {author} {\bibinfo {author} {\bibfnamefont {A.}~\bibnamefont
			{Gautier}}\ and\ \bibinfo {author} {\bibfnamefont {F.}~\bibnamefont
			{Tudisco}},\ }\bibfield  {title} {\bibinfo {title} {The contractivity of
			cone-preserving multilinear mappings},\ }\href@noop {} {\bibfield  {journal}
		{\bibinfo  {journal} {Nonlinearity}\ }\textbf {\bibinfo {volume} {32}},\
		\bibinfo {pages} {4713} (\bibinfo {year} {2019})}\BibitemShut {NoStop}%
	\bibitem [{\citenamefont {Gautier}\ \emph {et~al.}(2019)\citenamefont
		{Gautier}, \citenamefont {Tudisco},\ and\ \citenamefont
		{Hein}}]{gautier2019perron}%
	\BibitemOpen
	\bibfield  {author} {\bibinfo {author} {\bibfnamefont {A.}~\bibnamefont
			{Gautier}}, \bibinfo {author} {\bibfnamefont {F.}~\bibnamefont {Tudisco}},\
		and\ \bibinfo {author} {\bibfnamefont {M.}~\bibnamefont {Hein}},\ }\bibfield
	{title} {\bibinfo {title} {The perron--frobenius theorem for multihomogeneous
			mappings},\ }\href@noop {} {\bibfield  {journal} {\bibinfo  {journal} {SIAM
				Journal on Matrix Analysis and Applications}\ }\textbf {\bibinfo {volume}
			{40}},\ \bibinfo {pages} {1179} (\bibinfo {year} {2019})}\BibitemShut
	{NoStop}%
	\bibitem [{\citenamefont {Gautier}\ \emph {et~al.}(2023)\citenamefont
		{Gautier}, \citenamefont {Tudisco},\ and\ \citenamefont
		{Hein}}]{gautier2023nonlinear}%
	\BibitemOpen
	\bibfield  {author} {\bibinfo {author} {\bibfnamefont {A.}~\bibnamefont
			{Gautier}}, \bibinfo {author} {\bibfnamefont {F.}~\bibnamefont {Tudisco}},\
		and\ \bibinfo {author} {\bibfnamefont {M.}~\bibnamefont {Hein}},\ }\bibfield
	{title} {\bibinfo {title} {Nonlinear perron--frobenius theorems for
			nonnegative tensors},\ }\href@noop {} {\bibfield  {journal} {\bibinfo
			{journal} {SIAM Review}\ }\textbf {\bibinfo {volume} {65}},\ \bibinfo {pages}
		{495} (\bibinfo {year} {2023})}\BibitemShut {NoStop}%
	\bibitem [{Note1()}]{Note1}%
	\BibitemOpen
	\bibinfo {note} {We found a value of $\alpha =\beta =10$ to be a reasonable
		compromise between the two objectives of the kernels in \protect \textup
		{\hbox {\mathsurround \z@ \protect \normalfont (\ignorespaces \ref
				{eq:objective_function}\unskip \@@italiccorr )}} to closely match the maximum
		kernel and not being too restrictive in terms of the choice of the parameters
		$p$ and $q$ with respect to the condition $\rho (\protect \bm {M})<1$ in
		\protect \Cref {thm:unique_solution}.}\BibitemShut {Stop}%
	\bibitem [{\citenamefont {Priem}\ \emph {et~al.}(2022)\citenamefont {Priem},
		\citenamefont {Piwowar},\ and\ \citenamefont {Orr}}]{priem2022openalex}%
	\BibitemOpen
	\bibfield  {author} {\bibinfo {author} {\bibfnamefont {J.}~\bibnamefont
			{Priem}}, \bibinfo {author} {\bibfnamefont {H.}~\bibnamefont {Piwowar}},\
		and\ \bibinfo {author} {\bibfnamefont {R.}~\bibnamefont {Orr}},\ }\bibfield
	{title} {\bibinfo {title} {Openalex: A fully-open index of scholarly works,
			authors, venues, institutions, and concepts},\ }\bibfield  {journal}
	{\bibinfo  {journal} {arXiv preprint arXiv:2205.01833}\ }\href
	{https://doi.org/10.48550/arXiv.2205.01833} {10.48550/arXiv.2205.01833}
	(\bibinfo {year} {2022})\BibitemShut {NoStop}%
	\bibitem [{ope()}]{openalex_note}%
	\BibitemOpen
	\href@noop {} {\bibinfo {title} {Openalex data organizes topics
			hierarchically, with level-0 representing the most general concepts and
			increasingly fine divisions at levels 1, 2, etc.}}\BibitemShut {Stop}%
	\bibitem [{\citenamefont {Cardillo}\ \emph {et~al.}(2013)\citenamefont
		{Cardillo}, \citenamefont {G{\'o}mez-Gardenes}, \citenamefont {Zanin},
		\citenamefont {Romance}, \citenamefont {Papo}, \citenamefont {Pozo},\ and\
		\citenamefont {Boccaletti}}]{cardillo2013emergence}%
	\BibitemOpen
	\bibfield  {author} {\bibinfo {author} {\bibfnamefont {A.}~\bibnamefont
			{Cardillo}}, \bibinfo {author} {\bibfnamefont {J.}~\bibnamefont
			{G{\'o}mez-Gardenes}}, \bibinfo {author} {\bibfnamefont {M.}~\bibnamefont
			{Zanin}}, \bibinfo {author} {\bibfnamefont {M.}~\bibnamefont {Romance}},
		\bibinfo {author} {\bibfnamefont {D.}~\bibnamefont {Papo}}, \bibinfo {author}
		{\bibfnamefont {F.~d.}\ \bibnamefont {Pozo}},\ and\ \bibinfo {author}
		{\bibfnamefont {S.}~\bibnamefont {Boccaletti}},\ }\bibfield  {title}
	{\bibinfo {title} {Emergence of network features from multiplexity},\ }\href
	{https://doi.org/10.1038/srep01344} {\bibfield  {journal} {\bibinfo
			{journal} {Scientific reports}\ }\textbf {\bibinfo {volume} {3}},\ \bibinfo
		{pages} {1344} (\bibinfo {year} {2013})}\BibitemShut {NoStop}%
	\bibitem [{\citenamefont {Bergermann}\ and\ \citenamefont
		{Stoll}(2021)}]{bergermann2021orientations}%
	\BibitemOpen
	\bibfield  {author} {\bibinfo {author} {\bibfnamefont {K.}~\bibnamefont
			{Bergermann}}\ and\ \bibinfo {author} {\bibfnamefont {M.}~\bibnamefont
			{Stoll}},\ }\bibfield  {title} {\bibinfo {title} {Orientations and matrix
			function-based centralities in multiplex network analysis of urban public
			transport},\ }\href {https://doi.org/10.1007/s41109-021-00429-9} {\bibfield
		{journal} {\bibinfo  {journal} {Applied Network Science}\ }\textbf {\bibinfo
			{volume} {6}},\ \bibinfo {pages} {1} (\bibinfo {year} {2021})}\BibitemShut
	{NoStop}%
	\bibitem [{\citenamefont {Bergermann}\ and\ \citenamefont
		{Stoll}(2022)}]{bergermann2022fast}%
	\BibitemOpen
	\bibfield  {author} {\bibinfo {author} {\bibfnamefont {K.}~\bibnamefont
			{Bergermann}}\ and\ \bibinfo {author} {\bibfnamefont {M.}~\bibnamefont
			{Stoll}},\ }\bibfield  {title} {\bibinfo {title} {Fast computation of matrix
			function-based centrality measures for layer-coupled multiplex networks},\
	}\href {https://doi.org/10.1103/PhysRevE.105.034305} {\bibfield  {journal}
		{\bibinfo  {journal} {Physical Review E}\ }\textbf {\bibinfo {volume}
			{105}},\ \bibinfo {pages} {034305} (\bibinfo {year} {2022})}\BibitemShut
	{NoStop}%
	\bibitem [{\citenamefont {Timmer}\ \emph {et~al.}(2015)\citenamefont {Timmer},
		\citenamefont {Dietzenbacher}, \citenamefont {Los}, \citenamefont {Stehrer},\
		and\ \citenamefont {De~Vries}}]{timmer2015illustrated}%
	\BibitemOpen
	\bibfield  {author} {\bibinfo {author} {\bibfnamefont {M.~P.}\ \bibnamefont
			{Timmer}}, \bibinfo {author} {\bibfnamefont {E.}~\bibnamefont
			{Dietzenbacher}}, \bibinfo {author} {\bibfnamefont {B.}~\bibnamefont {Los}},
		\bibinfo {author} {\bibfnamefont {R.}~\bibnamefont {Stehrer}},\ and\ \bibinfo
		{author} {\bibfnamefont {G.~J.}\ \bibnamefont {De~Vries}},\ }\bibfield
	{title} {\bibinfo {title} {An illustrated user guide to the world
			input--output database: the case of global automotive production},\ }\href
	{https://doi.org/10.1111/roie.12178} {\bibfield  {journal} {\bibinfo
			{journal} {Review of International Economics}\ }\textbf {\bibinfo {volume}
			{23}},\ \bibinfo {pages} {575} (\bibinfo {year} {2015})}\BibitemShut
	{NoStop}%
	\bibitem [{\citenamefont {WIOD}(2021)}]{wiod2021}%
	\BibitemOpen
	\bibfield  {author} {\bibinfo {author} {\bibnamefont {WIOD}},\ }\bibfield
	{title} {\bibinfo {title} {World input-output database 2016 release,
			2000-2014, v2},\ }\bibfield  {journal} {\bibinfo  {journal} {DataverseNL}\
	}\href {https://doi.org/10.34894/PJ2M1C} {10.34894/PJ2M1C} (\bibinfo {year}
	{2021})\BibitemShut {NoStop}%
	\bibitem [{\citenamefont {MacMahon}\ and\ \citenamefont
		{Garlaschelli}(2015)}]{macmahon2015community}%
	\BibitemOpen
	\bibfield  {author} {\bibinfo {author} {\bibfnamefont {M.}~\bibnamefont
			{MacMahon}}\ and\ \bibinfo {author} {\bibfnamefont {D.}~\bibnamefont
			{Garlaschelli}},\ }\bibfield  {title} {\bibinfo {title} {Community detection
			for correlation matrices},\ }\href
	{https://doi.org/10.1103/PhysRevX.5.021006} {\bibfield  {journal} {\bibinfo
			{journal} {Phys. Rev. X}\ }\textbf {\bibinfo {volume} {5}},\ \bibinfo {pages}
		{021006} (\bibinfo {year} {2015})}\BibitemShut {NoStop}%
\end{thebibliography}
\end{document}